\documentclass[11pt,arabic,envcountsame]{llncs}
\usepackage{multirow}
\let\doendproof\endproof
\renewcommand\endproof{~\hfill\qed\doendproof}
\usepackage{verbatim}
\usepackage{amssymb}
\usepackage{amsmath}
\usepackage{url}
\usepackage{hyperref}
  
\usepackage{subfigure}
\usepackage{setspace}
\usepackage{xspace}
\usepackage{mdwlist}
\usepackage{textcomp}
\usepackage{graphicx}

\usepackage{algorithmic}
\usepackage{algorithm}

\newcommand{\pd}[0]{\tau}
\newcommand{\rs}[0]{\mbox{RS}}
\newcommand{\pth}[0]{\mbox{PS}}
\newcommand{\ds}[0]{\mbox{DS}}
\newcommand{\cd}[0]{\mbox{CD}}
\newcommand{\sig}[0]{\mbox{sig}}
\newcommand{\hevol}[0]{\mbox{hist}}
\newcommand{\del}[0]{d}
\newcommand{\mdel}[0]{\bar{d}}

\begin{document}

\title{On the Effectiveness of Punishments in a Repeated Epidemic Dissemination Game}

\author{Xavier Vila\c{c}a \and Lu\'{\i}s Rodrigues}

\institute{INESC-ID, Instituto Superior T\'{e}cnico, Universidade de Lisboa\\
xvilaca@gsd.inesc-id.pt\\
ler@ist.utl.pt}

\date{\today}

\maketitle

\begin{abstract}
This work uses Game Theory to study the effectiveness of punishments as an incentive
for rational nodes to follow an epidemic dissemination protocol.
The dissemination process is modeled as an infinite repetition of a stage game. At the end
of each stage, a monitoring mechanism informs each player of the actions of other nodes.
The effectiveness of a punishing strategy is measured as the range of values for the benefit-to-cost ratio that sustain cooperation.
This paper studies both public and private monitoring.  Under public monitoring, we show that direct reciprocity is not an effective incentive,
whereas full indirect reciprocity provides a nearly optimal effectiveness. Under private monitoring,
we identify necessary conditions regarding the topology of the graph in order for punishments to be effective.
When punishments are coordinated, full indirect reciprocity is also effective with private monitoring.
\keywords{Epidemic Dissemination, Game Theory, Peer-to-Peer.}
\end{abstract}

\section{Introduction}
Epidemic broadcast protocols are known to be extremely scalable and
robust\,\cite{Birman:99,Deshpande:06,Libo:08}. As a result, they are particularly
well suited to support the dissemination of information in large-scale
peer-to-peer systems, for instance, to support live
streaming\,\cite{Li:06,Li:08}. In such an environment, nodes do not
belong to the same administrative domain. On the contrary, many of
these systems rely on resources made available by self-interested
nodes that are not necessarily obedient to the protocol. In
particular, participants may be rational and aim at maximizing their
utility, which is a function of the benefits obtained from receiving
information and the cost of contributing to its dissemination.

Two main incentive mechanisms may be implemented to
ensure that rational nodes are not interested in deviating from the
protocol: one is to rely on balanced exchanges\,\cite{Li:06,Li:08};
other is to monitor the degree of cooperation of every node
and punish misbehavior\,\cite{Guerraoui:10}. When balanced exchanges are
enforced, in every interaction, nodes must exchange an equivalent
amount of messages of interest to each other. This approach has the
main disadvantage of requiring symmetric interactions between
nodes. In some cases, more efficient protocols may be achieved with
asymmetric interactions\,\cite{Deshpande:06,Libo:08,Guerraoui:10}, 
where balanced exchanges become infeasible.
Instead, nodes are expected to forward messages without immediately receiving any benefit
in return. Therefore, one must consider repeated interactions for nodes to able to
collect information about the behavior of their neighbors, which may
be used to detect misbehavior and trigger punishments.

Although monitoring has been used to detect and expel free-riders from epidemic dissemination protocols\,\cite{Guerraoui:10}, 
no theoretical analysis studied the ability of punishments to sustain cooperation among rational nodes.
Therefore, in this paper, we tackle this gap by using Game Theory\,\cite{Osborne:94}. The aim is to study the existence of equilibria
in an infinitely repeated Epidemic Dissemination game. The stage
game consists in a sequence of messages disseminated by the source,
which are forwarded by every node $i$ to each neighbor $j$ with an independent
probability $p_i[j]$. At the end of each stage, a monitoring
mechanism provides information to each node regarding $p_i[j]$.

Following work in classical Game Theory that shows that cooperation in repeated games
can be sustained using punishing strategies\,\cite{Fudenberg:86}, we focus on this class of strategies.
We assume that there is a pre-defined target for the reliability of the epidemic dissemination process. To
achieve this reliability, each node should forward every
message to each of its neighbors with a probability higher than some 
threshold probability $p$, known a priory by the two neighboring nodes. We
consider that a player $i$ defects from a neighbor $j$ if it uses a probability lower
than $p$ when forwarding information to $j$. Each node $i$ receives
a benefit $\beta_i$ per received message, but incurs a cost $\gamma_i$ of forwarding
a message to a neighbor. Given this, we are particularly interested in determining the range 
of values of the ratio benefit-to-cost $(\beta_i/\gamma_i)$ that allows punishing
strategies to be equilibria. The wider is this range, the more likely it is for all
nodes to cooperate.

The main contribution of this paper is a quantification of the
effectiveness of different punishing strategies under two types of monitors: public and private.
Public monitors inform every node of the actions of every other node with no delays. On the other hand,
private monitoring inform only a subset of nodes of the actions of each node, and possibly with some delays. 
In addition, we study two particular types of punishing strategies: direct and full indirect reciprocity.
In the former type, each node is solely responsible for monitoring and punishing each neighbor, individually. The
latter type specifies that each misbehaving node should eventually be punished by every neighbor.
More precisely, we make the following contributions:

\begin{itemize}
  \item We derive a generic necessary and sufficient condition for a punishing strategy to be a Subgame Perfect Equilibrium under public monitoring.
  From this condition, we also derive an upper bound for the effectiveness of strategies that use direct reciprocity
  as an incentive. We observe that this value decreases very quickly with an increasing reliability, in many realistic scenarios.
  On the other hand, if full indirect reciprocity is used, then this problem can be avoided. We derive a lower bound for the effectiveness
  of these strategies, which is not only independent from the desired reliability, but also close to the theoretical
  optimum, under certain circumstances.

  \item Using private monitoring with delays, information collected by each node may be incomplete, even if local monitoring is perfect.
  We thus consider the alternative solution concept of Sequential Equilibrium, which requires the specification of a belief
  system that captures the belief held by each player regarding past events of which it has not been informed. For a punishing
  strategy to be an equilibrium, this belief must be consistent. We provide a definition of consistency that is sufficient
  to derive the effectiveness of punishing strategies.
  
  \item  Under private monitoring with a consistent belief system, 
  we show that certain topologies are ineffective when monitoring is fully distributed.
  Then, we prove that, unless full indirect reciprocity is possible,
  the effectiveness decreases monotonically with the reliability. To avoid this problem, punishments should be coordinated, i.e.,
  punishments applied to a misbehaving node $i$ by every neighbor of $i$ should overlap in time.
  We derive a lower bound for the effectiveness of full indirect reciprocity strategies with coordinated punishments.
  The results indicate that the number of stages during which punishments overlap should be at least of the order of
  the maximum delay of the monitoring mechanism. This suggests that, when implementing a distributed monitoring mechanism,
  delays should be minimized.
\end{itemize}

The remainder of the paper is structured as follows. Section~\ref{sec:related}
discusses some related work. The general model
is provided in Section~\ref{sec:model}. The analysis of public and private monitoring
are given in Sections~\ref{sec:public} and~\ref{sec:private}, respectively. Section~\ref{sec:conc}
concludes the paper and provides directions of future work.

\section{Related Work}
\label{sec:related}
There are examples of work that use monitoring to persuade rational nodes to
engage in a dissemination protocol. In Equicast\,\cite{Keidar:09}, the authors perform
a Game Theoretical analysis of a multicast protocol where nodes monitor the rate of messages
sent by their neighbors and apply punishments whenever the rate drops below a certain threshold. 
The protocol is shown to be a dominating strategy.
Nodes are disposed in an approximately random network, thus, the dissemination process
resembles epidemic dissemination. However, given that the network is connected and
nodes are expected to forward messages to every neighbor with probability $1$, there is no non-determinism
in the delivery of messages. Furthermore, the authors restrict the
actions available to each player by assuming that they only adjust the number of neighbors with which they interact
and a parameter of the protocol. This contrasts with our analysis, where we consider non-deterministic
delivery of messages and a more general set of strategies available to players.

Guerraoui et al.\,\cite{Guerraoui:10} propose a mechanism
that monitors the degree of cooperation of each node in epidemic dissemination protocols.
The goal is to detect and expel free-riders.
This mechanism performs statistical inferences on the reports provided by every node regarding its neighbors,
and estimates the cooperation level of each node.
 If this cooperation level is lower than a minimum value, then the node is expelled from the network. The authors perform a theoretical and experimental analysis
to show that this mechanism guarantees that free-riders only benefit by deviating from the protocol
if the degree of deviation is not significantly high. However, no Game Theoretical analysis is performed to
determine in what conditions are free-riders willing to abide to the protocol.

In~\cite{Li:06,Li:08}, the authors rely on balanced exchanges to provide incentives for
nodes to cooperate in dissemination protocols for data streaming. In BAR Gossip\,\cite{Li:06}, the proposed epidemic dissemination protocol
enforces strictly balanced exchanges. This requires the use of a pseudo-random
number generator to determine the set of interactions in every round of exchanged updates,
and occasionally nodes may have to send garbage as a payment for any unbalance in the amount
of information exchanged with a neighbor. A stepwise analysis shows that nodes cannot increase their utility by deviating in any step of the protocol.
In FlightPath\,\cite{Li:08}, the authors remove the need for sending garbage by allowing imbalanced exchanges.
By limiting the maximum allowed imbalance between every pair of nodes, the authors
show that it is possible for the protocol to be an $1/10$-Nash equilibrium, while still
ensuring a streaming service with high quality. Unfortunately, these results might not
hold for other dissemination protocols that rely on highly imbalanced exchanges. 
In these cases, a better alternative might be to rely on a monitoring approach.

Other game theoretical analysis have addressed a similar problem, but
in different contexts.  In particular, the tit-for-tat strategy used
in BitTorrent, a P2P file sharing system, has been subjected to a
wide variety of Game Theoretical analysis\,\cite{Feldman:04,Rahman:11,Qiu:04,Levin:08}. 
These works consider a set of $n$ nodes deciding with which nodes to
cooperate, given a limited number of available connections,
with the intent to share content. Therefore, contrary to our analysis, 
there is no non-determinism in content delivery.

Closer to our goal is the trend of work that applies game theory to
selfish routing\,\cite{Srivastava:06,Felegyhazi:06,Ji:06}. 
In this problem, each node may be a source of messages to be routed
along a fixed path of multiple relay nodes to a given destination. The benefit
of a node is to have its messages delivered to the destination, while it incurs
the costs of forwarding messages as a relay node. This results in a linear
relationship between the actions of a player and the utility of other nodes. In our case,
that relation is captured by the definition of reliability, which is non-linear. Consequently,
the utility functions of an epidemic dissemination and a routing game possess an inherently
different structure.

\section{Model}
\label{sec:model}
We now describe the System and Game Theoretical models, followed by the definition of effectiveness.
In Appendix~\ref{sec:epidemic}, we provide a more thorough description of the considered epidemics model and include
some auxiliary results that are useful for the analysis.

\subsection{System Model}
There is a set of nodes ${\cal N}$ organized into a directed graph $G$.
This models a P2P overlay network with a stable membership.
Each node has a set of in (${\cal N}_i^{-1}$)
and out-neighbors (${\cal N}_i$). Communication channels are assumed to be reliable.
We model the generation of messages in this network by considering the existence of a single
external source $s$. Its behavior is described by a
profile $\vec{p}_s$, which defines for each node $i \in
{\cal N}_s$ the probability $p_s[i] \in [0,1)$ of $i$ receiving a message
directly from $s$, with the restriction that $p_s[i] > 0$ for some $i$.
We consider the graph to be connected from the source $s$, i.e.,
there exists a path from $s$ to every node $i \in {\cal N}$.
Conversely, every node $i$ forwards messages to every neighbor $j \in {\cal N}_i$ with an independent
probability $p_{i}[j]$. Provided a profile of probabilities $\vec{p}$, which includes
the vector of probabilities $\vec{p}_s$ and $\vec{p}_i$ used by $s$ and by every node $i$, respectively,
we can define the reliability of the dissemination protocol as the probability of a node receiving
a message. For the analysis, it is convenient to consider the probability of a node $i$ \emph{not} receiving
a message, denoted by $q_i[\vec{p}]$. The reliability of the protocol is then defined by $1-q_i[\vec{p}]$.
The exact expression of $q_i$ is included in Appendix~\ref{sec:epidemic}.

\subsection{Monitoring Mechanism.}
The monitoring mechanism emits a signal $s \in {\cal S}$, where every player $i$ may observe a different private signal $s_i \in s$.
This signal can take two values for every pair of nodes $j \in {\cal N}$ and $k\in{\cal N}_j$:
$s_i[j,k] = \mbox{\emph{cooperate}}$ notifies $i$ that $j$ forwarded messages to $k$ with a probability higher than a specified threshold, and
$s_i[j,k] = \mbox{\emph{defect}}$ signals the complementary action. This signal may be public, if all nodes read the same signal,
or private, otherwise. Moreover, if the signal is perfectly correlated with the action taken by a node, then
monitoring is perfect; otherwise, monitoring is said to be imperfect.

We consider that monitoring is performed locally by every node. 
A possible implementation of such monitoring mechanism in the context of P2P networks can be based on the work of~\cite{Guerraoui:10}.
A simpler and cheaper mechanism would instead consist in every out-neighbor $j$ of a given node $i$ recording the fraction
of messages sent by $i$ to $j$ during the dissemination of a fixed number $M$ of messages. Then, $j$ may use this information along with an estimate of the reliability of the dissemination
of messages to $i$ ($1-q_i$) in order to determine whether $i$ is cooperating or defecting. When a defection is detected, $j$ is
disseminates an accusation against $i$ towards other nodes. If $i$ is expected to use $p_i[j]<1$ towards $j$, then monitoring is imperfect. Furthermore,
accusations may be blocked, disrupted, or wrongly emitted against one node due to both malicious and rational behavior. However, in this paper, we consider only perfect monitoring, faithful propagation of accusations, and that nodes are rational. 
Almost perfect monitoring can be achieved with a large $M$. Faithful propagation may be reasonable to assume if the impact of punishments on the reliability of each non-punished node is small
and the cost of sending accusations is not significant. We intend to relax these assumptions in future work.

In our model, an accusation emitted by a node $j$ against an in-neighbor $i$ may only be received by
the nodes that are reachable from $j$ by following paths in the graph. 
In addition, if we consider the obvious possibility that $i$ might block any accusation emitted by one of its neighbors,
then these paths cannot cross $i$. Finally, the number of nodes informed of each defection may be further reduced to minimize
the monitoring costs. This restricts the set of in-neighbors of $i$ that may punish
$i$ for defecting $j$. In this paper, we consider two alternative models.
First, we study public monitoring, where all nodes may be informed about any
defection with no delays. Then, we study the private monitoring case, taking into consideration the possible
delay of the dissemination of accusations.

\subsection{Game Theoretical Model}
Our model considers an infinite repetition of a stage game. Each stage consists in
the dissemination of a sequence of messages and is interleaved with the execution
of the monitoring mechanism, which provides every node with some information
regarding the actions taken by other nodes during the stage game.

\subsubsection{Stage Game.}
The stage game is modeled as a strategic game. An action of a player $i$ is 
a vector of probabilities $\vec{p}_i \in {\cal P}_i$, such that $p_i[j]>0$ only if $j \in {\cal N}_i$. 
Thus, $\vec{p}_i$ represents the average probability used by $i$ to forward messages during the stage.
It is reasonable to consider that $i$ adheres to $\vec{p}_i$ during the complete stage, since $i$ expects to be monitored by other nodes
with regard to a given $\vec{p}_i$. Hence, changing strategy is equivalent to following a different $\vec{p}_i$.
Despite $s$ not being a player, for simplicity, we consider that every profile $\vec{p} \in {\cal P}$ implicitly contains $\vec{p}_s$.
We can also define a mixed strategy $a_i \in {\cal A}_i$ as a probability distribution over ${\cal P}_i$,
and a profile of mixed strategies $\vec{a} \in {\cal A}$ as a vector containing the mixed strategies followed by every player.
The utility of a player $i$ is a function of the benefit $\beta_i$ obtained per received message and
the cost $\gamma_i$ of forwarding a message to each neighbor.
More precisely, this utility is given by the probability of receiving messages ($1-q_i[\vec{p}]$) multiplied by the difference
between the benefit per message ($\beta_i$) and the expected cost of forwarding that message to every neighbor ($\gamma_i \sum_{j \in {\cal N}_i} p_{i}[j]$):
$$u_i[\vec{p}] = (1-q_i[\vec{p}])(\beta_i - \gamma_i \sum_{j \in {\cal N}_i} p_{i}[j]).$$

If players follow a profile of mixed strategies $\vec{a}$, then the expected utility is denoted by $u_i[\vec{a}]$, which definition depends
on the structure of every $a_j$.

\subsubsection{Repeated Game.}
The repeated game consists in the infinite interleaving between the stage game and the execution of the monitoring
mechanism, where future payoffs are discounted by a factor $\omega_i$ for every player $i$. 
The game is characterized by (possibly infinite) sequences of previously observed signals, named histories.
The set of finite histories observed by player $i$ is represented by ${\cal H}_i$ and ${\cal H} = ({\cal H}_i)_{i \in {\cal N}}$
is the set of all histories observed by any player. A pure strategy for the repeated game $\sigma_i \in \Sigma_i$ maps each history to an action $\vec{p}_i$,
where $\vec{\sigma} \in \Sigma$ is a profile of strategies. Consequently, $\vec{\sigma}[h]$ specifies for some history $h \in {\cal H}$ the profile
of strategies $\vec{p}$ for the stage game to be followed by every node after history $h$ is observed. A behavioral strategy $\sigma_i$ differs 
from a pure strategy only in that $i$ assigns a probability distribution $a_i \in {\cal A}_i$
over the set of actions for the stage game. For simplicity, we will use the same notation for the two types of strategies.
The expected utility of player $i$ after having observed history $h_i$ is given by $\pi_i[\vec{\sigma}|h_i]$. The exact definitions
of equilibrium and expected utility depend on the type of monitoring being implemented. Hence, these definitions will be provided
in each of the sections regarding public and private monitoring.

\subsubsection{A Brief Note on Notation.}
Throughout the paper, we will conveniently simplify the notation as follows. Whenever referring
to a profile of strategies $\vec{\sigma}$, followed by all nodes except $i$,
we will use the notation $\vec{\sigma}_{-i}$. Also, $(\sigma_i,\vec{\sigma}_{-i})$ denotes the composite of a strategy $\sigma_i$ and a profile $\vec{\sigma}_{-i}$. 
The same reasoning applies to profiles of pure and mixed strategies of the stage game. 
Finally, we will let $(h,s)$ denote the history that follows $h$ after signal $s$ is observed.

\subsection{Effectiveness}
We know from Game Theoretic literature that certain punishing strategies can sustain cooperation if the discount factor $\omega_i$
is sufficiently close to $1$\,\cite{Fudenberg:86}. This minimum value is a function of the parameters $\beta_i$ and $\gamma_i$ for every player $i$.
More precisely, for larger values of the benefit-to-cost ratio $\beta_i/\gamma_i$, the minimum required value of $\omega_i$ is smaller.
In addition, for certain values of the benefit-to-cost ratio, no value of $\omega_i$ can sustain cooperation. Notice that these parameters
are specified by the environment and thus cannot be adjusted in the protocol. Thus, a strategy is more effective if it is an equilibrium
for wider ranges of $\omega_i$, $\beta_i$, and $\gamma_i$. In this paper, we only measure the effectiveness of
a profile of strategies $\vec{\sigma}$ as the allowed range of values for the benefit-to-cost ratio.
\begin{definition}
The effectiveness of a profile $\vec{\sigma} \in \Sigma$ is given by $\psi[\vec{\sigma}] \subseteq [0,\infty)$,
such that, if, for every $i \in {\cal N}$, $\frac{\beta_i}{\gamma_i} \in \psi[\vec{\sigma}]$,
then there exists $\omega_i \in (0,1)$ for every $i \in {\cal N}$ such that $\vec{\sigma}$ is an equilibrium.
\end{definition}

\section{Public Monitoring}
\label{sec:public}
In this section, we assume that the graph allows public monitoring to be implemented. 
That is, every node is informed about each defection at the end of the stage when the defection occurred.
We can thus simplify the notation by considering only public signals $s \in {\cal S}$ and histories $h \in {\cal H}$. 
With perfect monitoring, the public signal observed after players follow $\vec{p} \in {\cal P}$ is deterministic.
This type of monitoring requires accusations to be broadcast. However, since the dissemination of accusations is interleaved
with the dissemination of a sequence of messages, monitoring costs may not be relevant if the size of each accusation is small, compared to the size
of messages being disseminated.

The section is organized as follows. We start by providing a general definition of punishing strategies and then introduce the definition of expected utility
and the solution concept for public monitoring. We then proceed to a Game Theoretical analysis, where
we analyze punishing strategies that use direct reciprocity and full indirect reciprocity.

\subsection{Public Signal and Punishing Strategies}
We study a wide variety of punishing strategies, by considering a parameter $\pd$
that specifies the duration of punishments. Of particular interest to this analysis is the case where the duration of punishments
is infinite, which is known in the Game Theoretical literature as the Grim-trigger strategy.
Furthermore, a punishing strategy specifies a Reaction Set
$\rs[i,j] \subseteq {\cal N}$ of nodes that are expected to react to every defection of $i$ from $j$ during $\pd$ stages.
This set always contains $i$ and $j$, but it may also contain other nodes. In particular, a third node $k \in \rs[i,j]$
that is an in-neighbor of $i$ ($k \in {\cal N}_i^{-1}$) is expected to stop forwarding any messages to $i$, as a punishment. 
If $k$ is not a neighbor of $i$, then $k$ may also adapt the probabilities used towards its out-neighbors, for instance,
to keep the reliability high for every unpunished node.

In order for a node $j \in {\cal N}$ to monitor an in-neighbor $i \in {\cal N}_j^{-1}$, the protocol must define for every history $h \in {\cal H}$
a threshold probability $p_i[j|h]$ with which $i$ should forward messages to $j$. Since $h$ is public, $p_i[j|h]$ is common knowledge
between $i$ and $j$, allowing for an accurate monitoring. Given this, the public signal for perfect public monitoring is defined as follows.
\begin{definition}
\label{def:pubsig}
For every $h \in {\cal H}$ and $\vec{p}' \in {\cal P}$, let $s = \sig[\vec{p}'|h]$ be the public signal observed
when players follow $\vec{p}'$. For every $i \in {\cal N}$ and $j \in {\cal N}_i$,
$s[i,j] = \mbox{\emph{cooperate}}$ if and only if $p_i'[j] \geq p_i[j|h]$.
\end{definition}

Then, a punishing strategy becomes a set of rules specifying how every $p_i[j|h]$ should be defined.
Namely, let $\sigma_i^* \in \Sigma_i$ denote a punishing strategy, which specifies that after a history $h$
every node $i$ should forward messages to a neighbor $j$ with probability $p_i[j|h]$.
We will denote by $\vec{\sigma}^* \in \Sigma$ the profile of punishing strategies. The restrictions imposed on every $p_i[j|h]$ can be defined as follows.
Every node $i$ evaluates the set of defections observed in a history $h$ by $i$ and every neighbor $j$ to which both nodes should react. 
Basing on this information, $i$ uses a deterministic function to determine the probability $p_i[j|h]$.
For convenience, we will define $p_i[k|h] = 0$ for every node $k \in {\cal N} \setminus {\cal N}_i$ that is not an out-neighbor of $i$.

For the precise definition of punishing strategy, we need an additional data structure called Defection Set ($\ds_i[j|h]$) containing the set of defections
to which both $i$ and $j$ are expected to react, according to $\rs$. This information is specified in the form of tuples $(k_1,k_2,r)$ stating that both $i$ and $j$
are expected to react to a defection of $k_1$ from $k_2$ that occurred in the previous $r$-th stage. This way, $p_i[j|h]$
is defined as a function of $\ds_i[j|h]$. Namely, if $\ds_i[j|h]$ contains some defection of
$j$ from $k$ and $i$ should react to it ($i \in\rs[j,k]$) or if $i$ defected from $j$, then $p_i[j|h]$ is $0$.
Otherwise, $i$ forwards messages to $j$ with any positive probability that is a deterministic function of $\ds_i[j|h]$.

\begin{definition}
\label{def:thr}
Define $\ds_i[j|h] \subseteq {\cal N} \times {\cal N} \times \mathbb{Z}$ as follows for every $i \in {\cal N}$, $j \in {\cal N}_i$, and $h \in {\cal H}$:
\begin{itemize}
     \item $\ds_i[j|\emptyset] = \emptyset$.
     \item For $h = (h',s)$, $\ds_i[j|h] = L_1 \cup L_2$, where:
     \begin{enumerate}
       \item $L_1 = \{(k_1,k_2,r+1) | (k_1,k_2,r) \in \ds_i[j|h'] \land r +1 < \pd)\}$.
       \item $L_2 = \{(k_1,k_2,0) | k_1, k_2 \in {\cal N} \land i,j \in \rs[k_1,k_2] \land s[k_1,k_2] = \mbox{\emph{defect}}\}$.
     \end{enumerate}
\end{itemize}
For every $h \in {\cal H}$, $i \in {\cal N}$, and $j \in {\cal N}_i$:
  \begin{itemize}
    \item If there exists $r <\pd$ such that $(i,j,r) \in \ds_i[j|h]$, then $p_i[j|h] = 0$.
    \item If there exist $r < \pd$ and $k \in {\cal N}_j$ such that $(j,k,r) \in \ds_i[j|h]$, then $p_i[j|h] = 0$.
    \item Otherwise, $p_i[j|h]$ is a positive function of $\ds_i[j|h]$.
  \end{itemize}
\end{definition}

We consider that the source $s$ also abides to this strategy.
In Section~\ref{sec:pub-evol}, we will show that it follows by construction that if some node $k$ observes a defection of $i$ 
from a node $j$ and $k \in \rs[i,j]$, then $k$ reacts to this defection during $\pd$ stages,
regardless of the ensuing actions of $i$ and the current punishments being applied.
In addition, after defecting some neighbor $j$, $i$ does not forward messages to any node of
$\rs[i,j]$ in any of the following $\pd$ stages.

\subsection{Expected Utility and Solution Concept}
The expected utility of a profile of pure strategies for every player $i$ and history $h$ is given by:
\begin{equation}
\label{eq:pure-util}
\pi_i[\vec{\sigma}|h] = u_i[\vec{p}] + \omega_i \pi_i[\vec{\sigma}|(h,\sig[\vec{p}|h])],
\end{equation}
where $\vec{p} = \vec{\sigma}[h]$. Conversely, we can define the expected utility for a profile of behavioral strategies as follows:
\begin{equation}
\label{eq:behave-util}
\pi_i[\vec{\sigma}|h] = u_i[\vec{a}] + \omega_i \sum_{s \in {\cal S}} \pi_i[\vec{\sigma}|(h,s)]pr[s|\vec{a},h],
\end{equation}
where $\vec{a} = \vec{\sigma}[h]$ and $pr[s|\vec{a},h]$ is defined as
$$pr[s|\vec{a},h] = \prod_{j \in {\cal N}}pr_{j}[s|a_j,h],$$
where $pr_{j}[s|a_j,h]$ is the probability of the actions of $j$ in $a_j$ leading to $s$.

The considered solution concept for this model is the notion of Subgame Perfect Equilibrium (SPE)\,\cite{Osborne:94}, which refines
the solution concept of Nash Equilibrium (NE) for repeated games. In particular, a profile of strategies is a NE if no player can increase
its utility by deviating, given that other players follow the specified strategies. The solution concept of NE is adequate for instance
for strategic games, where players choose their actions prior to the execution of the game. However, in repeated games, players
have multiple decision points, where they may adapt their actions according to the observed history of signals. In this case, the notion
of NE ignores the possibility of players being faced with histories that are not consistent with the defined strategy, e.g., when some
defection is observed. This raises the possibility of the equilibrium only being sustained by non-credible threats.
To tackle this issue, the notion of SPE was proposed, which requires in addition
the defined strategy to be a NE after any history. This intuition is formalized as follows.

\begin{definition}
A profile of strategies $\vec{\sigma}^*$ is a Subgame Perfect Equilibrium if and only if for every player $i \in {\cal N}$,
history $h \in {\cal H}$, and strategy $\sigma_i' \in \Sigma_i$, 
$$\pi_i[\sigma_i^*,\vec{\sigma}_{-i}^*|h] \geq \pi_i[\sigma_i',\vec{\sigma}_{-i}^*|h].$$
\end{definition}

While this definition considers variations in strategies, it is possible to analyze only variations in the first
strategy for the stage game after any history and holding $\vec{\sigma}^*$ for the remaining stages, as stated by the one-deviation property.
For any $a_i' \in {\cal A}_i$ and $\vec{p}_i' \in {\cal P}_i$, let $\sigma_i^*[h|a_i']$ and $\sigma_i^*[h|\vec{p}_i']$
denote the strategies where $i$ always follows $\sigma_i^*$, except after history $h$,
when it chooses $a'_i$ and $\vec{p}_i'$ respectively. The same notation will be used for profiles $\vec{a}' \in {\cal A}$ and $\vec{p}' \in {\cal P}$,
namely, $\vec{\sigma}^*[h|\vec{a}']$ and $\vec{\sigma}^*[h|\vec{p}']$, respectively.
The following property captures the above intuition, which is known
to be true from Game Theoretic literature and can be proven in a similar fashion to~\cite{Blackwell:65}:

\begin{property}
\label{prop:one-dev}
\textbf{One-deviation.} 
A profile of strategies $\vec{\sigma}^*$ is a SPE if and only if
for every player $i \in {\cal N}$, history $h \in {\cal H}$, and $a_i' \in {\cal A}_i$,
$\pi_i[\sigma_i^*,\vec{\sigma}_{-i}^*|h] \geq \pi_i[\sigma_i',\vec{\sigma}_{-i}^*|h]$,
where $\sigma_i' = \sigma_i^*[h|a_i']$.
\end{property}

\subsection{Evolution of the Network}
\label{sec:pub-evol}
After any history $h \in {\cal H}$, the network induced by $h$ when players follow $\vec{\sigma}^*$ can be characterized by a subgraph,
where a link $(i,j)$ is active iff $i$ is not punishing $j$ and $i$ has not defected from $j$ in the last $\pd$ stages. All the remaining
links are inactive. When considering a profile of punishing strategies $\vec{\sigma}^*$, the evolution of this subgraph over time is deterministic.
That is, after a certain number of stages, inactive links become active, such that at most after $\pd$ stages we obtain the original graph.
Given this, we prove some correctness properties of the punishment strategy, which require the introduction of some auxiliary notation.
The complete proofs are in Appendix~\ref{sec:proof-pub-evol}. All the considered proofs are performed by induction.

For any profile of pure strategies $\vec{\sigma} \in \Sigma$ and $h \in {\cal H}$, let $\hevol[h,r|\vec{\sigma}]$ denote the history resulting
from players following $\vec{\sigma}$ during $r$ stages, after having observed $h$. That is:
$$\hevol[h,r|\vec{\sigma}] = (h,(s^{r'})_{r' \in \{1\ldots r\}}),$$
such that $s^1 = \sig[\vec{\sigma}[h]|h]$ and for every $r' \in \{1\ldots r-1\}$ we have $s^{r'+1} = \sig[\vec{\sigma}[h']|h']$, where $h' = \hevol[h,r'|\vec{\sigma}]$.
Notice that $\hevol[h,0|\vec{\sigma}] = h$ for every $h \in {\cal H}$ and $\vec{\sigma} \in \Sigma$.

The following notation will be useful in the analysis, where, for every $r>0$, $h' = \hevol[h,r-1|\vec{\sigma}]$ and $\vec{p}' = \vec{\sigma}[h']$:
\begin{itemize}
  \item $q_i[h,r|\vec{\sigma}] = q_i[\vec{p}']$.
  \item $\bar{p}_i[h,r|\vec{\sigma}] = \sum_{j \in {\cal N}_i} p_i'[j]$.
  \item $u_i[h,r|\vec{\sigma}] = u_i[\vec{p}'] = (1-q_i[h,r|\vec{\sigma}])(\beta_i - \gamma_i \bar{p}_i[h,r|\vec{\sigma}])$.
  \item ${\cal N}_i[h] = \{j \in {\cal N}_i|p_i[j|h]>0\}$.
\end{itemize}

The following lemma characterizes the evolution of the punishments being applied to any pair of nodes.

\begin{lemma}
\label{lemma:corr-0}
For every $h \in {\cal H}$, $r \in \{1 \ldots \pd-1\}$, $i \in {\cal N}$,  and $j \in {\cal N}_i$, 
$$\ds_i[j|h_r^*] = \{(k_1,k_2,r'+r) | (k_1,k_2,r') \in \ds_i[j|h] \land r' + r < \pd \},$$
where $h_r^* = \hevol[h,r|\vec{\sigma}^*]$.
\end{lemma}
\begin{proof}
The proof is by induction, where for the initial case it follows immediately from the definition of $\ds_i$ that 
for every $(k_1,k_2,r' ) \in \ds_i[j|h]$ such that $r' < 1$,  $r'$ is incremented by $1$ in the next stage.
Inductively, after $r$ stages, either $r' + r< \pd$ and  $(k_1,k_2,r' +r) \in \ds_i[j|h]$ or 
$(k_1,k_2,r')$ has been removed from $\ds_i$. 
(\hyperref[proof:lemma:corr-0]{Complete proof in Section~\ref{proof:lemma:corr-0}}).
\end{proof}

From this lemma, we obtain the following trivial corollary that simply states that
every punishment ends after $\pd$ stages. This is true by the
fact that for every $h \in {\cal H}$, $i \in {\cal N}$, $j \in {\cal N}_i$, and $(k_1,k_2,r') \in \ds_i^{h}[j]$, 
we have $r' < \tau$.

\begin{corollary}
\label{corollary:corr-0}
For every $h \in {\cal H}$, $r \geq \pd$,
$i \in {\cal N}$, and $j \in {\cal N}_i$, it holds $\ds_i[j|h_r^*] = \emptyset$,
where $h_r^* = \hevol[h,r|\vec{\sigma}^*]$.
\end{corollary}

The following lemma proves that every node $i$ that defects from a neighbor $j$
expects to be punished exactly during the next $\pd$ stages, regardless
of the following actions of $i$ or the punishments already being applied to $i$.
The auxiliary notation $\cd_i[\vec{p}'|h]$ is used to denote the characterization
of the defections performed by $i$ in $\vec{p}'$ after history $h$. More precisely, for every $i \in {\cal N}$,
$$\cd_i[\vec{p}'|h] =\{j \in {\cal N}_i | p_i'[j] < p_i[j|h]\}.$$

\begin{lemma}
\label{lemma:corr-1}
For every $h \in {\cal H}$, $\vec{p}' \in {\cal P}$, $r \in \{1 \ldots \pd\}$,
$i \in {\cal N}$, and $j \in {\cal N}_i$,
\begin{equation}
\label{eq:res-corr1}
\ds_i[j|h_r'] = \ds_i[j|h_r^*] \cup \{(k_1,k_2,r-1) | k_1,k_2 \in {\cal N} \land k_2 \in \cd_{k_1}[\vec{p}'|h] \land  i,j \in \rs[k_1,k_2]\},
\end{equation}
where $h_r^* = \hevol[h,r|\vec{\sigma}^*]$, $h_r' = \hevol[h,r| \vec{\sigma}']$, and
$\vec{\sigma}' =\vec{\sigma}^*[h|\vec{p}']$ is the profile of strategies where all players follow $\vec{p}'$ in the first stage.
\end{lemma}
\begin{proof}
By induction, the base case follows from the definition of $\ds_i[j|h]$ and the fact that
$i$ registers every defection of $k_1$ to $k_2$ detected in $\vec{p}'$ by
adding $(k_1,k_2,0)$ to $\ds_i[j|h]$. Inductively, after $r\leq\pd$ stages,
this pair is transformed into $(k_1,k_2,r-1)$. 
(\hyperref[proof:lemma:corr-1]{Complete proof in Section~\ref{proof:lemma:corr-1}}).
\end{proof}

From the previous lemmas, it follows that any punishment ceases after $\pd$ stages,
which is proven in Lemma~\ref{lemma:corr-2}.

\begin{lemma}
\label{lemma:corr-2}
For every $h \in {\cal H}$, $\vec{p}' \in {\cal P}$, $r > \pd$,
$i \in {\cal N}$, and $j \in {\cal N}_i$,
\begin{equation}
\label{eq:res-corr2}
\ds_i[j|h_r'] = \ds_i[j|h_r^*] = \emptyset,
\end{equation}
where $h_r^* = \hevol[r|\vec{\sigma}^*]$, $h_r' = \hevol[r| \vec{\sigma}']$,
and $\vec{\sigma}' = \vec{\sigma}^*[h|\vec{p}']$.
\end{lemma}
\begin{proof}
From Corollary~\ref{corollary:corr-0} and~Lemma~\ref{lemma:corr-2}, it follows that after $\pd$ stages, every pair
$(k_1,k_2,r)$ is removed from $\ds_i[j|h]$.
(\hyperref[proof:lemma:corr-2]{Complete proof in Section~\ref{proof:lemma:corr-2}}).
\end{proof}

\subsection{Generic Results}
\label{sec:gen-cond}
This section provides some generic results. Namely, we first derive the theoretically optimal
effectiveness, which serves as an upper bound for the effectiveness of any profile of strategies.
Then, we derive a simplified generic necessary and sufficient condition for any profile of strategies to be a SPE. 
The complete proofs are in Appendix~\ref{sec:proof:gen-cond}.

Proposition~\ref{prop:folk} establishes a minimum necessary benefit-to-cost ratio
for any profile of strategies to be a SPE of the repeated Epidemic Dissemination Game. Intuitively, the benefit-to-cost ratio must be greater
than the expected costs of forwarding messages to neighbors ($\bar{p}_i = \sum_{j \in {\cal N}_i}p_i[j|\emptyset]$), since otherwise a player has incentives
to not forward any messages. This is the minimum benefit-to-cost ratio that provides
an enforceable utility as defined by the Folk Theorems\,\cite{Fudenberg:86}, given that the utility that results from
nodes following any profile $\vec{p} \in {\cal P}$ is feasible and the minmax utility is $0$. Consequently, this establishes an upper bound for the
effectiveness of any strategy.

\begin{proposition}
\label{prop:folk}
For every profile of punishing strategies $\vec{\sigma}^*$, if $\vec{\sigma}^*$ is a SPE,
then, for every $i \in {\cal N}$, $\frac{\beta_i}{\gamma_i} >\bar{p}_i$.
Consequently, $\psi[\vec{\sigma}^*] \subseteq (v,\infty)$,
where $v = \max_{i \in {\cal N}} \bar{p}_i$.
\end{proposition}
\begin{proof}
(\hyperref[proof:prop:folk]{See Section~\ref{proof:prop:folk}}).
\end{proof}

A necessary and sufficient condition for any profile of punishing strategies to be an equilibrium
is that no node has incentives to stop forwarding messages to any subset of neighbors, i.e., to drop
those neighbors. This condition is named the DC Condition, which is defined as follows:

\begin{definition}
\textbf{DC Condition.}
\label{drop:condition}
For every player $i \in {\cal N}$, history $h \in {\cal H}$,  and $D \subseteq {\cal N}_{i}[h]$,
\begin{equation}
\label{eq:gen-cond}
\sum_{r=0}^\pd \omega_i^r (u_i[h,r|\vec{\sigma}^*] - u_i[h,r|\vec{\sigma}']) \geq 0,
\end{equation}
where $\vec{\sigma}' = (\sigma_i',\vec{\sigma}_{-i}^*)$,
$\sigma_i' = \sigma_i^*[h|\vec{p}_i']$, and $\vec{p}_i'$ is defined as:
\begin{itemize}
  \item For every $j \in D$, $p_i'[j] = 0$.
  \item For every $j \in {\cal N}_i \setminus D$, $p_i'[j] = p_i[j|h]$.
\end{itemize}
\end{definition}

The following Lemma shows that the DC Condition is necessary.

\begin{lemma}
\label{lemma:gen-cond-nec}
If $\vec{\sigma}^*$ is a SPE, then the DC Condition is fulfilled.
\end{lemma}
\begin{proof}
By the One-deviation property, for a profile to be a SPE, a player $i$ must not 
be able to increase its utility by unilaterally deviating in the first stage. In particular, this is true
if $i$ deviates by dropping any subset $D$ of neighbors.
Furthermore, since any punishment ends after $\pd$ stages, by Lemma~\ref{lemma:corr-2},
we have that if nodes follow the deviating profile $\vec{\sigma}'$, then for every $r > \pd$, 
$$u_i[h,r|\vec{\sigma}^*] =u_i[h,r|\vec{\sigma}'].$$
The DC Condition follows by the One-deviation property and the fact that
$$\pi_i[\vec{\sigma}^*|h] - \pi_i[\vec{\sigma}'|h] = \sum_{r = 0}^\infty \omega_i^r(u_i[h,r|\vec{\sigma}^*]-u_i[h,r|\vec{\sigma}']).$$
(\hyperref[proof:lemma:gen-cond-nec]{Complete proof in Section~\ref{proof:lemma:gen-cond-nec}}).
\end{proof}

In Lemma~\ref{lemma:gen-cond-suff}, we also show that the DC Condition is sufficient.
In order to prove this, we first need to show that every node $i$ cannot increase its utility by not following
a pure strategy in every stage game and by not forwarding messages with a probability
in $\{0,p_i[j|h]\}$ to every neighbor $j$. This is shown in two steps. First, Lemma~\ref{lemma:best-response1}
proves that any local best response mixed strategy only gives positive probability to an action in $\{0,p_i[j|h]\}$. Second,
Lemma~\ref{lemma:best-response2} proves that there is a pure strategy for the stage game that is a local best response.

Define the set of local best response strategies for history $h \in {\cal H}$ and any $i \in {\cal N}$ as:
$$BR[\vec{\sigma}_{-i}^*|h] = \{a_i \in {\cal A}_i | \forall_{a_i' \in {\cal A}_i} \pi_i[(\sigma_i^*[h|a_i],\vec{\sigma}_{-i}^*)|h] \geq \pi_i[(\sigma_i^*[h|a_i'],\vec{\sigma}_{-i}^*)|h]\}.$$
Notice that $BR[\vec{\sigma}_{-i}^*|h]$ is not empty. The following lemma first proves that every player $i$
always uses probabilities in $ \{0,p_i[j|h]\}$ towards a neighbor $j$.

\begin{lemma}
\label{lemma:best-response1}
For every $i \in {\cal N}$, $h \in {\cal H}$, $a_i \in BR[\vec{\sigma}_{-i}^*|h]$, and $\vec{p}_i \in {\cal P}_i$ such that $a_i[\vec{p}_i] > 0$,
it is true that for every $j \in {\cal N}_i$ we have $p_i[j] \in \{0,p_i[j|h]\}$.
\end{lemma}
\begin{proof}
The proof is by contradiction. Namely, assume that for some $\vec{p}_i$ and $a_i$ that is a best response
and $a_i[\vec{p}_i] > 0$, we have that $\vec{p}_i$ does not fulfill the restrictions defined above.
We can find $a_i'$ and $\vec{p}_i'$ such that:
\begin{itemize}
 \item $\vec{p}_i'$ fulfills the restrictions defined in the lemma.
 \item $a_i'[\vec{p}_i'] = a_i[\vec{p}_i] + a_i[\vec{p}_i']$ and $a_i'[\vec{p}_i] = 0$.
\end{itemize}
By letting $\vec{p}^* = \vec{\sigma}^*[h]$, we have that 
$$\sig[(\vec{p}_i,\vec{p}^*_{-i})|h] = \sig[(\vec{p}_i',\vec{p}^*_{-i})|h].$$
Consequently,
$$\pi_i[\sigma_i^*[h|a_i],\vec{\sigma}_{-i}^*|h] < \pi_i[\sigma_i^*[h|a_i'],\vec{\sigma}_{-i}^*|h].$$
Thus, $a_i$ cannot be a best response, which is a contradiction.
(\hyperref[proof:lemma:best-response1]{Complete proof in Section~\ref{proof:lemma:best-response1}}).
\end{proof}

We now have to show that there exists a pure strategy in $BR[\vec{\sigma}_{-i}^*|h]$.

\begin{lemma}
\label{lemma:best-response2}
For every $h \in {\cal H}$ and $i \in {\cal N}$, there exists $a_i \in BR[\vec{\sigma}_{-i}^*| h]$ and $\vec{p}_i \in {\cal P}_i$
such that $a_i[\vec{p}_i] = 1$.
\end{lemma}
\begin{proof}
First, notice that if only pure strategies are best-responses, then the result follows immediately. If there exists a mixed strategy 
$a_i$ that is a best response, then $i$ must be indifferent between following any profile $\vec{p}_i$ such that $a_i[\vec{p}_i]>0$.
Otherwise, $i$ could find a better strategy $a_i'$. In that case, any such profile $\vec{p}_i$ is a best response.
(\hyperref[proof:lemma:best-response2]{Complete proof in Section~\ref{proof:lemma:best-response2}}).
\end{proof}

Lemma~\ref{lemma:best-response} is a direct consequence of Lemmas~\ref{lemma:best-response1} and~\ref{lemma:best-response2}.

\begin{lemma}
\label{lemma:best-response}
For every $h \in {\cal H}$ and $i \in {\cal N}$, there exists $\vec{p}_i \in {\cal P}_i$ and a pure strategy $\sigma_i=\sigma_i^*[h|\vec{p}_i]$ such that:
\begin{enumerate}
 \item For every $j \in {\cal N}_i$, $p_i[j] \in \{0,p_{i}[j|h]\}$.
 \item For every $a_i \in {\cal A}_i$, $\pi_i[\sigma_i,\vec{\sigma}_{-i}^*|h] \geq \pi_i[\sigma_i',\vec{\sigma}_{-i}^*|h]$,
where $\sigma_i' = \sigma_i^*[h|a_i]$.
\end{enumerate}
\end{lemma}
\begin{proof}
(\hyperref[proof:lemma:best-response]{See Section~\ref{proof:lemma:best-response}}).
\end{proof}

It is now possible to show that the DC Condition is sufficient.

\begin{lemma}
\label{lemma:gen-cond-suff}
If the DC Condition is fulfilled, then $\vec{\sigma}^*$ is a SPE.
\end{lemma}
\begin{proof}
If the DC Condition holds, then no player $i$ can increase its utility by dropping
any subset of neighbors. By Lemma~\ref{lemma:best-response}, it follows that $i$
cannot increase its utility by following any alternative strategy for the first stage game,
which by the One-deviation property implies that the profile $\vec{\sigma}^*$ is a SPE.
(\hyperref[proof:lemma:gen-cond-suff]{Complete proof in Section~\ref{proof:lemma:gen-cond-suff}}).
\end{proof}

The following theorem merges the results from Lemmas~\ref{lemma:gen-cond-nec}
and~\ref{lemma:gen-cond-suff}.
\begin{theorem}
\label{theorem:gen-cond}
$\vec{\sigma}^*$ is a SPE if and only if the DC Condition holds.
\end{theorem}

\subsection{Direct Reciprocity is not Effective}
\label{sec:direct}
If $G$ is undirected, then it is possible to use direct reciprocity only, by defining $\rs[i,j] = \{i,j\}$
for every $i\in{\cal N}$ and $j \in {\cal N}_i$. That is, if $i$ defects from $j$, then only $j$ punishes $i$.
Direct reciprocity is the ideal incentive mechanism in a fully distributed environment, since it does not require
accusations to be sent by any node. The goal of this section is to show that punishments that use direct reciprocity
are not effective, even using public monitoring. To prove this, we first derive a generic necessary benefit-to-cost ratio
and then we identify the conditions under which direct reciprocity is ineffective. The complete proofs are included in Appendix~\ref{sec:proof:direct}.

Lemma~\ref{lemma:nec-btc} derives a minimum benefit-to-cost ratio for direct reciprocity.

\begin{lemma}
\label{lemma:nec-btc}
If $\vec{\sigma}^*$ is a SPE, then, for every $i \in {\cal N}$ and $j \in {\cal N}_i$, it is true that $q_i' > q_i^*$ and:
\begin{equation}
\label{eq:nec-btc}
\frac{\beta_i}{\gamma_i} > \bar{p}_i + \frac{p_i[j|\emptyset]}{q_i' - q_i^*}\left(1-q_i' + \frac{1-q_i^*}{\pd}\right),
\end{equation}
where $\vec{p}_i'$ is the strategy where $i$ drops $j$, $\vec{\sigma}' = (\sigma_i^*[\emptyset|\vec{p}_i'],\vec{\sigma}_{-i}^*)$,
$q_i'=q_i[\vec{\sigma}'[\emptyset]]$, and $q_i^* = q_i[\vec{\sigma}^*[\emptyset]]$.

\end{lemma}
\begin{proof}
By the definition of SPE and Theorem~\ref{theorem:gen-cond}, the DC Condition must hold for
the initial empty history and every deviation in the first stage where any player $i$ drops an out-neighbor $j$.
After some manipulations of the DC Condition for this specific scenario, Inequality~\ref{eq:nec-btc} is obtained.
(\hyperref[proof:lemma:nec-btc]{Complete proof in Section~\ref{proof:lemma:nec-btc}}).
\end{proof}

Lemma~\ref{lemma:dir-recip} also shows that direct reciprocity is not
an effective incentive mechanism under certain circumstances. Namely, by letting $q_i^*$ to be the probability
of delivery of messages in equilibrium ($q_i[\vec{\sigma}^*[\emptyset]]$), we find
that, if $p_i[j|\emptyset] + q_i^* \ll 1$, then the effectiveness is of the order $(1/q_i^*,\infty)$,
which decreases to $\emptyset$ very quickly with an increasing reliability.
The conditions under which direct reciprocity is ineffective are easily met, for instance, when a node has more neighbors than what is
strictly necessary to ensure high reliability.

\begin{lemma}
\label{lemma:dir-recip}
Suppose that for any $i \in {\cal N}$ and $j \in {\cal N}_i$, $p_i[j|\emptyset] + q_i^* \ll 1$.
If $\vec{\sigma}^*$ is a SPE, then:
\begin{equation}
\label{eq:dir-necessary}
\psi[\vec{\sigma}^*] \subseteq \left(\frac{1}{q_i^*},\infty\right),
\end{equation}
where $q_i^* = q_i[\vec{\sigma}^*[\emptyset]]$.
\end{lemma}
\begin{proof}
The proof follows directly from Lemma~\ref{lemma:nec-btc} and the fact that,
as we show in Lemma~\ref{lemma:single-impact} from Appendix~\ref{sec:epidemic},
we have that if $q_i'$ results from exactly $j$ punishing $i$ for defecting from $j$, then
$$q_i' \leq q_i^* \frac{1}{1-p_i[j|\emptyset]}.$$
(\hyperref[proof:lemma:dir-recip]{Complete proof in Section~\ref{proof:lemma:dir-recip}}).
\end{proof}

\subsection{Full Indirect Reciprocity is Sufficient}
\label{sec:indir}
Unlike direct reciprocity, if full indirect reciprocity is used, then the effectiveness may be independent of the reliability of the dissemination protocol.
This consists in the case where for every $i \in {\cal N}$ and $j \in {\cal N}_i$ we have 
\begin{equation}
{\cal N}_i^{-1} \subseteq \rs[i,j].
\end{equation}

The goal of this section is to show that, if full indirect reciprocity
is used, then the effectiveness is independent of the reliability of the dissemination protocol under certain circumstances.
To prove this, we proceed in two steps. First, we conveniently simplify
the DC Condition. Then, we derive a sufficient benefit-to-cost ratio for $\vec{\sigma}^*$ to be a SPE.
The complete proofs are included in Appendix~\ref{sec:proof:indir}.

The following lemma simplifies the DC Condition for this specific type of punishing strategies.
\begin{lemma}
\label{lemma:indir-equiv}
The profile of strategies $\vec{\sigma}^*$ is a SPE if and only if for every $h \in {\cal H}$ and $i \in {\cal N}$:
\begin{equation}
\label{eq:indir-equiv}
\sum_{r=1}^\pd (\omega_i^r u_i[h,r|\vec{\sigma}^*]) - (1-q_i[h,0|\vec{\sigma}^*])\gamma_i \bar{p}_i[h,0|\vec{\sigma}^*] \geq 0.
\end{equation} 
\end{lemma}
\begin{proof}
This simplification is obtained directly from the DC Condition and the fact that, if a node $i$ has incentives
to drop some out-neighbor $j$, then the best response strategy is to drop all out-neighbors.
This is proven by defining $\vec{p}''$, where $i$ drops a subset $D$ of out-neighbors,
and $\vec{p}'$, where $i$ drops every out-neighbor. We can prove that:
\begin{itemize}
  \item $\sig[\vec{p}'|h] = \sig[\vec{p}''|h]$.
  \item $u_i[\vec{p}'] > u_i[\vec{p}'']$.
  \item For any $r>0$, $u_i[h,r|\vec{\sigma}'] = u_i[h,r|\vec{\sigma}'']$, where $\vec{\sigma}'$ and $\vec{\sigma}''$
  differ from $\vec{\sigma}^*$ exactly in that $i$ follows $\vec{p}'$ and $\vec{p}''$ in the first stage, respectively.
\end{itemize}
This implies that $\pi_i[\vec{\sigma}'|h] >\pi_i[\vec{\sigma}''|h]$, and therefore the best response for $i$ is
to drop all neighbors.
(\hyperref[proof:lemma:indir-equiv]{Complete proof in Section~\ref{proof:lemma:indir-equiv}}).
\end{proof}

Theorem~\ref{theorem:indir-suff} derives a lower bound for the effectiveness of a full indirect reciprocity profile of strategies $\vec{\sigma}^*$.
This is done in two steps. First, it is shown that the history $h$ that minimizes the left side of Inequality~\ref{eq:indir-equiv}
results exactly in the same punishments being applied during the first $\pd-1$ stages. This is proven in Lemma~\ref{lemma:maxh}.

\begin{lemma}
\label{lemma:maxh}
Let $h \in {\cal H}$ be defined such that for every $h' \in {\cal H}$,
the value of the left side of Inequality~\ref{eq:indir-equiv} for $h$ is lower than or equal to the value for $h'$.
Then, for every $r \in \{1 \ldots \pd-2\}$,
$$u_i[h,r|\vec{\sigma}^*] = u_i[h,r+1|\vec{\sigma}^*].$$
\end{lemma}
\begin{proof}
The proof is performed by contradiction, where we assume that $h$ minimizes the left side of Inequality~\ref{eq:indir-equiv},
but for some $r \in \{1 \ldots \pd- 2\}$
$$u_i[h,r|\vec{\sigma}^*] \neq u_i[h,r+1|\vec{\sigma}^*].$$
This implies that in $h$ a set of punishments ends at the end of stage $r$. We can find $h'$ where those
punishments are either postponed or anticipated one stage and such that the left side of Inequality~\ref{eq:indir-equiv} is 
lower for $h'$ than for $h$, which is a contradiction.
(\hyperref[proof:lemma:maxh]{Complete proof in Section~\ref{proof:lemma:maxh}}).
\end{proof}

It is now possible to derive a sufficient benefit-to-cost ratio for full indirect reciprocity to be a SPE,
which constitutes a lower bound for the effectiveness of these strategies. However, this derivation
is only valid when the following assumption holds. There must exist a constant $c\geq 1$ such
that for every history $h$:
\begin{equation}
\label{eq:indir-assum}
q_i[h,0|\vec{\sigma}^*] \geq 1- c(1-q_i[h,1|\vec{\sigma}^*]).
\end{equation}

Intuitively, this states that, after some history $h$, if the value of $q_i$ varies from the first stage to the second
due to some punishments being concluded, then this variation is never too large. With this assumption,
we can derive a sufficient benefit-to-cost for $\vec{\sigma}^*$ to be a SPE.

\begin{theorem}
\label{theorem:indir-suff}
If there exists a constant $c \geq 1$ such that, for every $h \in {\cal H}$ and $i \in {\cal N}$, Assumption~\ref{eq:indir-assum} holds, then
$\psi[\vec{\sigma}^*] \supseteq (v,\infty)$, where
\begin{equation}
\label{eq:indir-suff}
v = \max_{h \in {\cal H}}\max_{i \in {\cal N}} \bar{p}_i[h,0|\vec{\sigma}^*]\left(1 + \frac{c}{\pd}\right).
\end{equation}
\end{theorem}
\begin{proof}
We consider the history $h$ that minimizes the left side of Inequality~\ref{eq:indir-equiv}.
Using the result of Lemma~\ref{lemma:maxh}, after some manipulations,
we can find that if for every $i \in {\cal N}$, $\beta_i/\gamma_i\in (v,\infty)$, then there exist $\omega_i \in (0,1)$ for every $i \in {\cal N}$ such that Inequality~\ref{eq:indir-equiv}
is true for every history $h'$, which implies by Lemma~\ref{lemma:indir-equiv} that $\vec{\sigma}^*$ is a SPE. This allows us to conclude
that $\psi[\vec{\sigma}^*] \supseteq (v,\infty)$.
(\hyperref[proof:theorem:indir-suff]{Complete proof in Section~\ref{proof:theorem:indir-suff}}).
\end{proof}

We can then conclude that if $c$ is small or $\pd$ is large, and the maximum of $\bar{p}_i[h,0|\vec{\sigma}^*]$ is never much larger than $\bar{p}_i$ for every $i$ and $h$, 
then the effectiveness of full indirect reciprocity is close to the optimum derived in Proposition~\ref{prop:folk}. In particular, if Grim-trigger is used ($\tau \to \infty$) 
and $\bar{p}_i$ is maximal for every $i$, then the effectiveness is optimal.
Furthermore, if for any $h$ both $q_i[h,0]$ and $q_i[h,1]$ are small,
then the effectiveness of full indirect reciprocity differs from the theoretical optimum only by
a factor $1+1/\pd$, which is upper bounded by $2$ for any $\pd\geq1$.

\section{Private Monitoring}
\label{sec:private}
When using public monitoring, we make the implicit assumption that the monitoring mechanism is able
to provide the same information instantly to every node, which requires the existence of a path from every out-neighbor $j$ of any node $i$ to every node of the graph,
that does not cross $i$. In addition, public monitoring is only possible if accusations are broadcast to every node.
We now consider private monitoring, where the dissemination of accusations may be restricted by the topology and scalability constraints.
However, any node that receives an accusation may react to it. Therefore, the definition of $\rs$ is no longer necessary.
In addition, accusations may be delayed.

\subsection{Private Signals}
In private monitoring, signals are determined by the history of previous signals $h$ and the profile $\vec{p}$
followed in the last stage. Namely, $\sig[\vec{p}|h]$ returns
a signal $s$, such that every node $i$ observes only its private signal $s_i \in s$, indicating for every other node $j \in {\cal N}$ whether $j$ cooperated or defected
with its out-neighbors in previous stages. The distinction between cooperation and defection is now determined
by a threshold probability $p_i[j|h_i]$. If a node $i$ defects an out-neighbor $j$ in stage $r$, then $k$ is informed of this defection
with a delay $\del_k[i,j]$, i.e., $k$ is informed only at the end of stage $r+\del_k[i,j]$.
We only assume that, for every node $i \in {\cal N}$ and $j \in {\cal N}_i$,
both $i$ and $j$ are informed instantly of the action of $i$ towards $j$ in the previous stage, i.e.:
$$\del_i[i,j] = \del_j[i,j] = 0.$$ 

We consider that these delays are common knowledge among players.
Moreover, we still assume that monitoring is perfect and that accusations are propagated faithfully.
We intend to relax the assumptions in future work. 
With this in mind, it is possible to provide a precise definition of a private signal.
For every player $i \in {\cal N}$ and history $h \in {\cal H}$,
we denote by $h_i \in h$ the private history observed by $i$ when all players observe the history $h \in {\cal H}$.
If $|h_i| \geq r \geq 1$, then let $h_i^r$ denote the last $r$-th signal observed by $i$, where $h_i^1$ is the last signal.
A private signal is defined such that if some node $j$ observes a defection of an in-neighbor $i \in {\cal N}_j^{-1}$,
every node $k \in {\cal N}$ such that $\del_k[i,j]$ is finite ($\del_k[i,j] < \infty$) observes this defection $\del_k[i,j]$ stages
after the end of the stage it occurred. The value $\del_k[i,j]$ is infinite if and only if accusations emitted by $j$ against $i$ may never reach $k$,
either due to every path from $j$ to $k$ crossing $i$ or the monitoring mechanism not disseminating the
accusation to $k$. However, if there exists a path from $j$ to $k$ without crossing $i$ and $k \in {\cal N}_i^{-1}$,
then $\del_k[i,j] < \infty$.

Formally:

\begin{definition}
\label{def:privsig}
For every $i \in {\cal N}$ and $h \in {\cal H}$, let $s_i' \in \sig[\vec{p}'|h]$ be the private signal observed by $i$
when players follow $\vec{p}' \in {\cal P}$ after having observed $h$. We have:
\begin{itemize}
  \item For every $j \in {\cal N}$ and $k \in {\cal N}_j$ such that $\del_i[j,k] = 0$, $s_i'[j,k] = \mbox{\emph{cooperate}}$ if and only if $p_j'[k] \geq p_j[k|h_j]$,
  where $h_j \in h$.
  \item For every $j \in {\cal N}$ and $k \in {\cal N}_j$ such that $0<\del_i[j,k] < \infty$, $s_i'[j,k] = \mbox{\emph{defect}}$ if and only if:
  \begin{itemize}
    \item $|h| \geq \del_i[j,k]$.
    \item For $h_k \in h$ and $s_k' = h_k^{\del_i[j,k]}$, $s_k'[j,k] = \mbox{\emph{defect}}$. 
  \end{itemize}
  \item For every $j \in {\cal N}$ and $k \in {\cal N}_j$ such that $\del_i[j,k] = \infty$, $s_i'[j,k] = \mbox{\emph{cooperate}}$.
\end{itemize}
\end{definition}

\subsection{Private Punishments}
In this context, we can define a punishing strategy $\sigma_i^*$ for every node $i$ as a 
function of the threshold probability $p_i[j|h_i]$ determined by $i$ for every private history $h_i$
and out-neighbor $j$. Notice that in the definition of private signals we assume that an accusation by $j$ is emitted against $i$
iff $i$ uses $p_i[j] < p_i[j|h_i]$ for any private history $h_i$. This was reasonable to assume in public monitoring,
where histories were public. Here, the strategy must also specify for every private history $h_j$
the threshold probability $p_i[j|h_j]$, since $h_i$ may differ from $h_j$. In order for $j$ to accurately monitor $i$,
for every $h \in {\cal H}$ and $h_i,h_j \in h$, we must have
\begin{equation}
\label{eq:common-knowledge}
p_i[j|h_i] = p_i[j|h_j].
\end{equation}

Therefore, both threshold probabilities must be computed as a function of the same set of signals.
The only issue with this requirement is that defection signals may arrive at different stages to $i$ and $j$.
For instance, if $k_1$ defects from $k_2$ in stage $r$ and $\del_i[k_1,k_2] < \del_j[k_1,k_2]$, then $i$ must wait for
stage $r+\del_j[k_1,k_2]$ before taking this defection into consideration in the computation of the threshold
probability. Furthermore, as in public monitoring, $i$ and $j$ are expected to react to a given defection
for a finite number of stages. However, as we will see later, this number should vary according to the delays
in order for punishments to be effective. Thus, for every $k_1 \in {\cal N}$ and $k_2 \in {\cal N}_{k_1}$,
we define $\pd[k_1,k_2|i,j]$ to be the number of stages during which $i$ and $j$ react to a given defection of $k_1$ from $k_2$.
Notice that $\pd[k_1,k_2|i,j] = \pd[k_1,k_2|j,i]$.

This intuition is formalized as follows. As in public monitoring, $\ds_i[j|h_i]$ denotes the 
set of defections observed by $i$ and that $j$ \emph{will eventually observe}. This set also contains tuples in the form $(k_1,k_2,r)$. 
The main difference is that now $i$ may have to wait before considering this tuple in the definition of $p_i[j|h_i]$. We signal
this by allowing $r$ to be negative and by using the tuple in the definition of $p_i[j|h_i]$ only when $r \geq0$. 
When $i$ observes a defection for the first time, it adds $(k_1,k_2,v)$ to  $\ds_i[j|h_i]$, where 
$$v= \min[\del_i[k_1,k_2] - \del_j[k_1,k_2],0].$$

Then, $i$ removes this pair when $r = \pd[k_1,k_2|i,j]$. For simplicity, we allow $v$ to take the value $\infty$
when $\del_j[k_1,k_2] = \infty$, resulting in that $i$ never takes into consideration this defection when determining $p_i[j|h_i]$.
This leads to the following definition.

\begin{definition}
\label{def:priv-thr}
For every $i \in {\cal N}$, $h_i \in {\cal H}_i$, and $j \in {\cal N}_i \cup {\cal N}_i^{-1}$, define $\ds_i[j|h_i] \subseteq {\cal N} \times \mathbb{Z}$ as follows:
  \begin{itemize}
     \item $\ds_i[j|\emptyset] = \emptyset$.
     \item For $h_i = (h_i',s_i')$, $\ds_i[j|h_i] = L_1 \cup L_2$, where:
     \begin{enumerate}
       \item $L_1 = \{(k_1,k_2,r+1) | (k_1,k_2,r) \in \ds_i[j|h_i'] \land r+1 < \pd[k_1,k_2|i,j]\}$.
       \item $L_2 = \{(k_1,k_2,v) | k_1,k_2 \in {\cal N} \land s_i'[k_1,k_2] = \mbox{\emph{defect}}\}$, where:
      $$v= \min[\del_i[k_1,k_2] - \del_j[k_1,k_2],0].$$     
     \end{enumerate}
  \end{itemize}
For every $i \in {\cal N}$, $h_i \in {\cal H}_i$, and  $j \in {\cal N}_i$, $\sigma_i^*[h_i] = p_i[j|h_i]$:
    \begin{itemize}
      \item Let $K = \{(k_1,k_2,r) \in \ds_i[j|h_i] | r \geq 0\}$.
      \item If there exists $r \geq 0$ such that $(i,j,r) \in K$, then $p_i[j|h_i] = 0$.
      \item If there exist $r \geq 0$ and $k \in {\cal N}_j$ such that $(j,k,r) \in K$, then $p_i[j|h_i] = 0$.
      \item Otherwise, $p_i[j|h]$ is a positive function of $K$.
    \end{itemize}

For every $j \in {\cal N}_i^{-1} \setminus {\cal N}_i$, $\sigma_i^*[h_i] = p_j[i|h_i]$ such that:
    \begin{itemize}
      \item Let $K = \{(k_1,k_2,r) \in \ds_i[j|h_i] | r \geq 0 \}$.
      \item If there exists $r \geq 0$ such that $(j,i,r) \in K$, then $p_j[i|h_i] = 0$.
      \item If there exist $r \geq 0$ and $k \in {\cal N}_i$ such that $(i,k,r) \in K$, then $p_j[i|h_i] = 0$.
      \item Otherwise, $p_j[i|h_i]$ is a positive function of $K$.
    \end{itemize}  
\end{definition}

\subsection{Expected Utility and Solution Concept}
We now model the interactions as a repeated game with imperfect information and perfect recall, for which
the solution concept of Sequential Equilibrium is adequate\,\cite{Kreps:85}.
Its definition requires the specification of a belief system $\vec{\mu}$. After a player $i$
observes a private history $h_i \in {\cal H}_i$, $i$ must form some expectation regarding
the history $h \in {\cal H}$ observed by every player, which must include $h_i$. This is captured by
a probability distribution $\mu_i[.|h_i]$ over ${\cal H}$. By defining $\vec{\mu}=(\mu_i)_{i \in {\cal N}}$,
we call a pair $(\vec{\sigma},\vec{\mu})$ an assessment, which is assumed to be common
knowledge among all players. The expected utility of a profile of strategies $\vec{\sigma}$ is then defined as:
\begin{equation}
\label{eq:priv-exp-util}
\pi_i[\vec{\sigma}|\vec{\mu},h_i] = \sum_{h \in {\cal H}} \mu_i[h|h_i]\pi_i[\vec{\sigma}|h],
\end{equation}
where $\pi_i[\vec{\sigma}|h]$ is defined as in the public monitoring case.

An assessment $(\vec{\sigma}^*,\vec{\mu}^*)$ is a Sequential Equilibrium if and only if
$(\vec{\sigma}^*,\vec{\mu}^*)$ is Sequentially Rational and Consistent.
The definition of sequential rationality is identical to that of subgame perfection:
\begin{definition}
An assessment $(\vec{\sigma}^*,\vec{\mu}^*)$ is Sequentially Rational if and only if for every $i \in {\cal N}$,
$h_i \in {\cal H}_i$, and $\sigma_i' \in \Sigma_i$, $\pi_i[\vec{\sigma}^*|\vec{\mu}^*,h_i] \geq \pi_i[\sigma_i',\vec{\sigma}^*_{-i}|\vec{\mu}^*,h_i]$.
\end{definition}

However, defining consistency for an assessment $(\vec{\sigma}^*,\vec{\mu}^*)$ is more intricate. The idea of defining
this concept was introduced in~\cite{Kreps:85}, intuitively defined as follows in our context. For any profile $\vec{\sigma}^*$, every private history $h_i$ that may
be reached with positive probability when players follow $\vec{\sigma}^*$ is said to be consistent with $\vec{\sigma}^*$;
otherwise, $h_i$ is inconsistent. For any consistent $h_i$, $\mu_i[h|h_i]$ must be defined using the Bayes rule.
The definition of $\mu_i$ for inconsistent private histories varies with the specific definition of Consistent Assessment.
It turns out that, in our case, the notion of Preconsistency introduced in~\cite{Hendon:96} is sufficient.

We now provide the formal definition of Preconsistency and later provide an interpretation in the context of punishment strategies.
Let $pr_i[h'|h,\vec{\sigma}]$ be the probability assigned by $i$ to $h' \in {\cal H}$ being reached from $h \subset h'$ when
players follow $\vec{\sigma} \in \Sigma$. Given $h_i \in {\cal H}_i$, we can define 
$$pr_i[h'|\vec{\mu},h_i,\vec{\sigma}] = \sum_{h \in {\cal H} } \mu_i[h|h_i]pr_i[h'|h,\vec{\sigma}].$$
For $h_i' \in {\cal H}_i$ such that $h_i \subset h_i'$, let:
$$pr_i[h_i'|\vec{\mu},h_i,\vec{\sigma}] = \sum_{h' \in {\cal H} : h_i' \in h'} pr_i[h'|\vec{\mu},h_i,\vec{\sigma}].$$

\begin{definition}
An assessment $(\vec{\sigma}^*,\vec{\mu}^*)$ is Preconsistent if and only if
for every $i \in {\cal N}$, $h_i \in {\cal H}_i$, and $h_i' \in {\cal H}_i$ such that $h_i \subset h_i'$,
if there exists $\sigma_i' \in \Sigma_i$ such that $pr_i[h_i'|\vec{\mu}^*,h_i,(\sigma_i',\vec{\sigma}^*_{-i})] > 0$,
then for every $h' \in {\cal H}$ such that $h_i' \in h'$: 
$$\mu_i[h'|h_i'] = \frac{pr_i[h'|\vec{\mu}^*,h_i,(\sigma_i',\vec{\sigma}^*_{-i})]}{pr_i[h_i'|\vec{\mu}^*,h_i,(\sigma_i',\vec{\sigma}^*_{-i})]}.$$
\end{definition}

The underlying intuition of this definition when considering a profile of punishing strategies $\vec{\sigma}^*$ is
as follows. First, notice that a history $h_i$ is consistent with $\vec{\sigma}^*$ if and only if no defections are observed in $h_i$.
For any $\sigma_i'$, the profile of strategies $\vec{\sigma}' = (\sigma_i',\vec{\sigma}_{-i}^*)$ may specify non-deterministic actions for the stage game,
by only due to $\sigma_i'$. Consider any $h_i'$. If $h_i'$ does not contain any defections, then the only strategies $\sigma_i'$
such that $h_i'$ is consistent with any such $\vec{\sigma}'$ are those where $i$ does not defect any node. Therefore, we can set $h_i$
to the empty history and from the above definition derive the conclusion that $\mu_i[h'|h_i']=1$ if and only if $h'$ does not contain
any defections. Similarly, if $h_i'$ only contains defections performed by $i$, then $h_i$ can be the empty set and there must
be only one $h'$ such that $\mu_i[h'|h_i'] =1$, which is the history where only defections performed by $i$ are observed by any player.

If $h_i'$ contains only defections performed by $i$, then we can use induction on the number of defections committed by other nodes to prove that
$\mu_i[h'|h_i'] =1$ if and only if $h'$ contains exactly the defections observed by $i$ in $h_i'$. The base case follows from the two previous scenarios.
As for the induction step, there are two hypothesis. If the last defections were performed in the last stage, then there is no $h_i \subset h_i'$ 
such that 
\begin{equation}
\label{eq:consistent-aux}
pr_i[h_i'|\vec{\mu},h_i,\vec{\sigma}'] > 0.
\end{equation}

Otherwise, consider that the last defections performed by other nodes occurred in the last $r$-th stage where $r>1$, 
when $i$ observed $h_i$. In this case, it is true that~\ref{eq:consistent-aux} holds. Here, there is only one history $h$
such that $h_i \in h$ and $\mu_i[h|h_i] = 1$, which is true by the induction hypothesis. This history contains exactly
the defections observed by $i$ in $h_i$, which are also included in $h_i'$. Thus, the only history $h'$ that may follow
$h_i$ fulfills the condition that no other defection was performed other than what $i$ observed in $h_i'$.

In summary, $\mu[h|h_i] = 1$ if and only if $h$ is the history containing
$h_i$ and the set of defections observed by any node $j \in {\cal H}$ in $h_j \in h$ is a subset of the set of defections
observed in $h_i$. The importance of this definition of consistency is that in~\cite{Kreps:85} the authors prove that the One-deviation
property also holds for Preconsistent assessments, which is sufficient for our analysis.

\begin{property}
\label{prop:priv-one-dev}
\textbf{One-deviation.} A Preconsistent assessment $(\vec{\sigma}^*,\vec{\mu}^*)$ is Sequentially Rational if and only if
for every player $i \in {\cal N}$, history $h_i \in {\cal H}_i$, and profile $\vec{a}_i' \in {\cal A}_i$,
$$\pi_i[\sigma_i^*,\vec{\sigma}_{-i}^*|\vec{\mu}^*,h_i] \geq \pi_i[\sigma_i',\vec{\sigma}_{-i}^*|\vec{\mu}^*,h_i],$$
where $\sigma_i' = \sigma_i^*[h_i|\vec{a}_i']$ is defined as in public monitoring.
\end{property}

\subsection{Evolution of the Network}
\label{sec:priv-evol}
When a player $i$ observes a private history $h_i \in {\cal H}_i$, only the histories $h \in {\cal H}$ such that $h_i \in h$ can be observed by other players.
Given this, we use the same notation as in public monitoring, when referring to the evolution of the network after a history $h$
is observed. Namely, $\hevol[h,r|\vec{\sigma}]$ is the resulting history starting from the observation of $h$ and when all players
follow the pure strategy $\vec{\sigma}$. Therefore, we continue to use the same notation for $q_i$ ($q_i[h,r|\vec{\sigma}]$),
$\bar{p}_i$ ($\bar{p}_i[h,r|\vec{\sigma}]$), and $u_i$ ($u_i[h,r|\vec{\sigma}]$). Now, we have
$${\cal N}_i[h_i] = \{j \in {\cal N}_i | p_i[j|h_i] > 0\}.$$

The definition of $\cd_i[\vec{p}|h] \subseteq {\cal N}$ is
almost identical to that of the public monitoring case. Namely, for every $i \in {\cal N}$, $h \in {\cal H}$, and $h_i \in h$,
$$\cd_i[\vec{p}|h] = \{j \in {\cal N}_i | p_i[j] < p_i[j|h_i]\}.$$

The following lemma proves that every node $k_1$ that defects from an out-neighbor $k_2$
expects $i$ and $j$ to react to this defection during the next $\pd[k_1,k_2|i,j]$ stages, regardless
of the following actions of $k_1$ or the punishments already being applied to $k_1$.

\begin{lemma}
\label{lemma:priv-corr-1}
For every $h \in {\cal H}$, $\vec{p}' \in {\cal P}$, $r >0$,
$i \in {\cal N}$, and $j \in {\cal N}_i$:
\begin{equation}
\label{eq:priv-res-corr1}
\begin{array}{ll}
\ds_i[j|h_{i,r}'] =& \ds_i[j|h_{i,r}^*] \cup \{(k_1,k_2,r -1 - \del_i[k_1,k_2]+v[k_1,k_2]) | k_1,k_2 \in {\cal N} \land \\
                            & k_2 \in \cd_{k_1}[\vec{p}'|h] \land r \in \{\del_i[k_1,k_2]+1 \ldots \del_i[k_1,k_2] + \pd[k_1,k_2|i,j]-v[k_1,k_2]\}\land\\
                            & v[k_1,k_2] = \min[\del_i[k_1,k_2] - \del_j[k_1,k_2],0]\},\\
\end{array}
\end{equation}
where $h_{i,r}^* \in \hevol[h,r|\vec{\sigma}^*]$, $h_{i,r}' \in \hevol[h,r| \vec{\sigma}']$, and
$\vec{\sigma}' =\vec{\sigma}^*[h|\vec{p}']$ is the profile of strategies where all players follow $\vec{p}'$ in the first stage.
\end{lemma}
\begin{proof}
By induction, the base case follows from the definition of $\ds_i$ and the fact that
$i$ registers every defection of $k_1$ to $k_2$ in stage $\del_i[k_1,k_2]$,
adding $(k_1,k_2,v[k_1,k_2])$ to $\ds_i[j|h]$. Inductively, after $r\leq\del_i[k_1,k_2] +\pd[k_1,k_2|i,j] - v[k_1,k_2]$ stages,
this pair is transformed into $(k_1,k_2,r-1)$.
(\hyperref[proof:lemma:priv-corr-1]{Complete proof in Section~\ref{proof:lemma:priv-corr-1}}).
\end{proof}

For the sake of completeness, we prove in Lemma~\ref{lemma:priv-history} that
the strategy is well defined, in terms of defining threshold probabilities that are always
common knowledge between pairs of players. This supports our assumption in the
definition of private signals that an accusation is emitted by $j$ against $i$
iff $i$ uses $p_i[j] < p_i[j|h_j]$ towards $j$. 

\begin{lemma}
\label{lemma:priv-history}
For every $i \in {\cal N}$, $j \in {\cal N}_i$, $h \in {\cal H}$, and $h_i,h_j \in h$:
$$p_i[j|h_i] = p_i[j|h_j].$$
\end{lemma}
\begin{proof}
It follows from the definition of $\ds$ that a node $i$ never includes in the set $K$
a tuple $(k_1,k_2,r)$ such that $r<0$, for any out-neighbor $j$. This value is only negative when $\del_i[k_1,k_2] < \del_j[k_1,k_2]$,
in which case $v$ is set to $-(\del_j[k_1,k_2] - \del_i[k_1,k_2])$. By Lemma~\ref{lemma:priv-corr-1},
this value only becomes $0$ when $j$ is also informed of this defection, in which case both nodes include the pair in $K$.
Consequently, $K$ is always defined identically by $i$ and $j$ after any history $h$, which implies the result.
(\hyperref[proof:lemma:priv-history]{Complete proof in Section~\ref{proof:lemma:priv-history}}).
\end{proof}

\subsection{Generic Results}
Proposition~\ref{prop:priv-folk} reestablishes the optimal effectiveness for private monitoring. 
The proof of this proposition is identical to that of Proposition~\ref{prop:folk}. The only difference
lies in the fact that now the effectiveness of a profile of strategies $\vec{\sigma}^*$ is conditional on
a belief system $\vec{\mu}^*$ ($\psi[\vec{\sigma}^*|\vec{\mu}^*]$).

\begin{proposition}
\label{prop:priv-folk}
For every assessment $(\vec{\sigma}^*,\vec{\mu}^*)$, if $(\vec{\sigma}^*,\vec{\mu}^*)$ is Sequentially Rational,
then, for every $i \in {\cal N}$, $\frac{\beta_i}{\gamma_i} \geq \bar{p}_i$. 
Consequently, $\psi[\vec{\sigma}^*] \subseteq (v,\infty)$,
where $v = \max_{i \in {\cal N}} \bar{p}_i$.
\end{proposition}
\begin{proof}
(\hyperref[proof:prop:priv-folk]{See Section~\ref{proof:prop:priv-folk}}).
\end{proof}

As in the public monitoring case, we can define a necessary and sufficient condition for the defined
profile of strategies to be Sequentially Rational, which we name PDC Condition.
The proof of both necessity and sufficiency is almost identical to that of public monitoring. 

\begin{definition}
\textbf{PDC Condition}.
For every $i \in {\cal N}$, $h_i \in {\cal H}_i$, and $D \subseteq {\cal N}_i[h_i]$,
\begin{equation}
\label{eq:priv-drop}
\sum_{h \in {\cal H}}\mu_i^*[h|h_i]\sum_{r=0}^{\infty} \omega_i^r (u_i[h,r|\vec{\sigma}^*] - u_i'[h,r|\vec{\sigma}']) \geq 0,
\end{equation}
where $\vec{\sigma}' = (\sigma_i^*[h_i|\vec{p}_i'],\vec{\sigma}_{-i}^*)$ and $\vec{p}_i'$ is defined as:
\begin{itemize}
  \item For every $j \in D$, $p_i'[j] = 0$.
  \item For every $j \in {\cal N} \setminus D$, $p_i'[j] = p_i[j|h_i]$.
\end{itemize}
\end{definition}

The following corollary captures the fact that the PDC condition is necessary, which
follows directly from the One-deviation property and the definition of $\pi_i$ for private monitoring.

\begin{corollary}
\label{corollary:priv-drop-nec}
If the assessment $(\vec{\sigma}^*,\vec{\mu}^*)$ is Sequentially Rational and Preconsistent, then the PDC Condition
is fulfilled.
\end{corollary}

In order to prove that the PDC Condition is sufficient, we proceed in the same fashion to public monitoring, always
implicitly assuming that the considered assessment is Preconsistent.
Redefine the set of local best responses for every node $i$ and private history $h_i \in {\cal H}_i$ as:
$$BR[\vec{\sigma}_{-i}^*|\vec{\mu}^*,h_i] = \{a_i \in {\cal A}_i | \forall_{a_i' \in {\cal A}_i} \pi_i[(\sigma_i^*[h_i|a_i],\vec{\sigma}_{-i}^*)|\vec{\mu}^*,h_i] \geq \pi_i[(\sigma_i^*[h_i|a_i'],\vec{\sigma}_{-i}^*)|\vec{\mu}^*,h_i]\}.$$

\begin{lemma}
\label{lemma:priv-best-response1}
For every $i \in {\cal N}$, $h_i \in {\cal H}_i$, $a_i \in BR[\vec{\sigma}_{-i}^*|\vec{\mu}^*,h_i]$, and $\vec{p}_i \in {\cal P}_i$ such that $a_i[\vec{p}_i] > 0$,
it is true that for every $j \in {\cal N}_i$ we have $p_i[j] \in \{0,p_i[j|h_i]\}$.
\end{lemma}
\begin{proof}
(\hyperref[proof:lemma:priv-best-response1]{See Section~\ref{proof:lemma:priv-best-response1}}).
\end{proof}

\begin{lemma}
\label{lemma:priv-best-response2}
For every $i \in {\cal N}$ and $h_i \in {\cal H}_i$, there exists $a_i \in BR[\vec{\sigma}_{-i}^*|\vec{\mu}^*,h_i]$ and $\vec{p}_i \in {\cal P}_i$
such that $a_i[\vec{p}_i] = 1$.
\end{lemma}
\begin{proof}
(\hyperref[proof:lemma:priv-best-response2]{See Section~\ref{proof:lemma:priv-best-response2}}).
\end{proof}

\begin{lemma}
\label{lemma:priv-best-response}
For every $i \in {\cal N}$ and $h_i \in {\cal H}_i$, there exists $\vec{p}_i \in {\cal P}_i$ and a pure strategy $\sigma_i=\sigma_i^*[h_i|\vec{p}_i]$ such that:
\begin{enumerate}
 \item For every $j \in {\cal N}_i$, $p_i[j] \in \{0,p_{i}[j|h_i]\}$.
 \item For every $a_i \in {\cal A}_i$, $\pi_i[\sigma_i,\vec{\sigma}_{-i}^*|\vec{\mu}^*,h_i] \geq \pi_i[\sigma_i',\vec{\sigma}_{-i}^*|\vec{\mu}^*,h_i]$,
where $\sigma_i' = \sigma_i^*[h_i|a_i]$.
\end{enumerate}
\end{lemma}
\begin{proof}
(\hyperref[proof:lemma:priv-best-response]{See Section~\ref{proof:lemma:priv-best-response}}).
\end{proof}

\begin{lemma}
\label{lemma:priv-drop-suff}
If the PDC Condition is fulfilled and $(\vec{\sigma}^*,\vec{\mu}^*)$ is Preconsistent, then $(\vec{\sigma}^*,\vec{\mu}^*)$ is Sequentially Rational.
\end{lemma}

\begin{proof}
(\hyperref[proof:lemma:priv-drop-suff]{See Section~\ref{proof:lemma:priv-drop-suff}}).
\end{proof}

The following theorem merges the results from Corollary~\ref{corollary:priv-drop-nec} and Lemma~\ref{lemma:priv-drop-suff}.

\begin{theorem}
\label{theorem:priv-drop}
If $(\vec{\sigma}^*,\vec{\mu}^*)$ is Preconsistent, then
$(\vec{\sigma}^*,\vec{\mu}^*)$ is Sequentially Rational if and only if the PDC Condition is fulfilled.
\end{theorem}

\subsection{Ineffective Topologies}
An important consequence of Theorem~\ref{theorem:priv-drop} is that not every topology allows the existence
of equilibria for punishing strategies. In fact, if there is some node $i$ and a neighbor $j$ such that every node $k$
that is reachable from $j$ without crossing $i$ is never in between $s$ and $i$, then the impact of the punishments
applied to $i$ after defecting from $j$ is null. This intuition is formalized in Lemma~\ref{lemma:paths},
where $\pth[i,j]$ denotes the set of paths from $i$ to $j$ in $G$.

\begin{lemma}
\label{lemma:paths}
If the assessment $(\vec{\sigma}^*,\vec{\mu}^*)$ is Preconsistent and 
Sequentially Rational, then for every $i \in {\cal N}$ and $j \in {\cal N}_i$,
there is $k \in {\cal N} \setminus \{i\}$, $x \in \pth[s,i]$, and $x' \in \pth[j,k]$, such that $k \in x$ and $i \notin x'$.
\end{lemma}
\begin{proof}
Assume by contradiction the opposite. Then, nodes can follow $\vec{p}'$
where $i$ drops $j$, such that, if $\vec{p}''$ is the profile resulting from nodes punishing $i$,
then by Lemma~\ref{lemma:bottleneck-impact} from Appendix~\ref{sec:epidemic},
$q_i[\vec{p}''] = q_i[\vec{p}^*]$, where $\vec{p}^* = \vec{\sigma}^*[\emptyset]$. 
Since $i$ increases its utility in the first stage by deviating, we have
$$u_i[\vec{p}''] > u_i[\vec{p}^*].$$
Moreover, by letting $\vec{\sigma}' =\vec{\sigma}^*[\vec{p}']$ to be the profile where exactly $i$ defects $j$, it is true that for every $r>0$
$$u_i[\emptyset,r|\vec{\sigma}'] = u_i[\emptyset,r|\vec{\sigma}^*].$$
This implies that $\pi_i[\vec{\sigma}'|\vec{\mu}^*,\emptyset] > \pi_i[\vec{\sigma}^*|\vec{\mu}^*,\emptyset]$, which is a contradiction.
(\hyperref[proof:lemma:paths]{Complete proof in Section~\ref{proof:lemma:paths}}).
\end{proof}

The main implication of this result is that many non-redundant topologies, i.e., that do not contain multiple paths between $s$
and every node, are ineffective at sustaining cooperation. This is not entirely surprising, since it was already known that cooperation cannot be sustained
using punishments as incentives in non-redundant graphs such as trees\,\cite{Ngan:04}. But even slightly redundant
structures, such as directed cycles, do not fulfill the necessary condition specified in Lemma~\ref{lemma:paths}.
Although redundancy is desirable to fulfill the above condition, it might decrease the effectiveness of punishments unless full indirect reciprocity may be used,
as shown in the following section.

\subsection{Redundancy may decrease Effectiveness}
In addition to the need to fulfill the necessary condition of Lemma~\ref{lemma:paths}, 
a higher redundancy increases tolerance to failures. We show in Theorem~\ref{theorem:problem} that if the graph is redundant and it does not allow full indirect reciprocity to be implemented, 
then the effectiveness decreases monotonically with the increase of the reliability.
We consider the graph to be redundant if there are multiple non-overlapping paths from the source to every node.
More precisely, if for every $i \in {\cal N}$ and $j \in {\cal N} \setminus \{i\}$, there exists $x \in \pth[s,i]$ such that
$j \notin x$, then $G$ is redundant. 

The reliability increases as the probabilities $p_i[j|\emptyset]$ approach $1$  for every node $i$ and out-neighbor $j$.
This is denoted by $\lim_{\vec{\sigma}^*\to \vec{1}}$. We find that the effectiveness of any punishing strategy that cannot implement full indirect
reciprocity converges to $\emptyset$, i.e., no benefit-to-cost ratio can sustain cooperation. This intuition is formalized as follows:

\begin{equation}
\label{eq:problem}
\lim_{\vec{\sigma}^*\to \vec{1}} \psi[\vec{\sigma}^*|\vec{\mu}^*] = \cdot_{i \in {\cal N},j \in {\cal N}_i} \lim_{p_i[j|\emptyset]\to 1} \psi[\vec{\sigma}^*|\vec{\mu}^*] = \emptyset.
\end{equation}

Theorem~\ref{theorem:problem} proves that Equality~\ref{eq:problem} holds for any graph that does not allow full indirect reciprocity
to be implemented, which in our model occurs when there is no path from some $j \in {\cal N}_i$ to some $k \in {\cal N}_i^{-1}$
that does not cross $i$.

\begin{theorem}
\label{theorem:problem}
If $G$ is redundant and there exist $i \in {\cal N}$, $j \in {\cal N}_i$, and $k \in {\cal N}_i^{-1}$ such that for every $x \in \pth[j,k]$ we have $i \in x$,
then Equality~\ref{eq:problem} holds.
\end{theorem}
\begin{proof}
The proof defines a deviating profile of strategies $\vec{\sigma}'$ where exactly $i$ drops $j$.
It follows that there is a path $x$ from $s$ to $i$ such that every node $k \in x$ never
reacts to this defection. By Definition~\ref{def:priv-thr}, $k$ uses
$p_k[l|\emptyset]$ towards every out-neighbor $l$; a value that converges to $1$. It follows from Lemma~\ref{lemma:prob-1}
in Appendix~\ref{sec:epidemic} that 
$$\lim_{\vec{\sigma}^* \to 1}( \pi_i[\vec{\sigma}^*|\vec{\mu}^*,\emptyset] - \pi_i[\vec{\sigma}'|\vec{\mu}^*,\emptyset] ) < 0.$$
Therefore, for any $\beta_i$ and $\gamma_i$, the left side of the PDC Condition converges to a negative value.
By Theorem~\ref{theorem:priv-drop}, this implies that $\psi[\vec{\sigma}^*|\vec{\mu}^*]$ converges to $\emptyset$.
(\hyperref[proof:theorem:problem]{Complete proof in Section~\ref{proof:theorem:problem}}).
\end{proof}

Notice that this result does not imply that only full indirect reciprocity is effective at incentivizing rational
nodes to cooperate in all scenarios. In fact, in many realistic scenarios, it might suffice for a majority of the in-neighbors
of a node $i$ to punish $i$ after any deviation. A more sensible analysis would take into consideration
the rate of converge to $\emptyset$ as the reliability increases. 
However, full indirect reciprocity is necessary in order to achieve an effectiveness fully independent of the desired reliability
in any redundant graph.

\subsection{Coordination is Desirable}
Although full indirect reciprocity is desirable for redundant graphs, 
we now show that for some definitions of punishing strategies, it might not be sufficient if monitoring incurs large delays. 
In particular, nodes also need to coordinate the punishments
being applied to any node, such that these punishments overlap during at least one stage after the defection, cancelling out any benefit obtained for receiving 
messages along some redundant path. This intuition is formalized as follows.
\begin{definition}
\label{def:coord}
An assessment $(\vec{\sigma}^*,\vec{\mu}^*)$ enforces coordination if and only if
for every $i \in {\cal N}$ and $j \in {\cal N}_i$, there exists $r>0$ such that,
for every $k \in {\cal N}_i^{-1}$,
$$r \in \{\del_k[i,j] + 1 \ldots \del_k[i,j] + \pd[i,j|k,i]\}.$$
\end{definition}

We prove a similar theorem to Theorem~\ref{theorem:problem}, which states that for
some definitions of punishing strategies, with a redundant graph, the effectiveness
decreases to $\emptyset$ with the reliability.

\begin{theorem}
\label{theorem:problem-coord}
If the graph is redundant and $\vec{\sigma}^*$ does not enforce coordination, then there is
a definition of $\vec{\sigma}^*$ such that Equality~\ref{eq:problem} holds.
\end{theorem}
\begin{proof}
Consider the punishing strategy where every node $i$ reacts only to the defections of out-neighbors
or to its own defections, and uses $p_i[j|\emptyset]$ in any other case. If $\vec{\sigma}^*$ does not
enforce coordination, then for some $i \in {\cal N}$ and $j \in {\cal N}_i$, and for every $r>0$,
we can find a path from $s$ to $i$ such that every node $k$ along the path uses the probability
$p_k[l|\emptyset]$ towards the next node $l$ in the path. These probabilities converge to $1$ as
the reliability increases, which by Lemma~\ref{lemma:prob-1} implies that
$$\lim_{\vec{\sigma}^* \to 1}( \pi_i[\vec{\sigma}^*|\vec{\mu}^*,\emptyset] - \pi_i[\vec{\sigma}'|\vec{\mu}^*,\emptyset] ) < 0,$$
where $\vec{\sigma}'$ is the alternative profile of strategies where exactly $i$ drops $j$.
Therefore, for any $\beta_i$ and $\gamma_i$, the left side of the PDC Condition converges to a negative value.
By Theorem~\ref{theorem:priv-drop}, this implies that $\psi[\vec{\sigma}^*|\vec{\mu}^*]$ converges to $\emptyset$.
(\hyperref[proof:theorem:problem-coord]{Complete proof in Section~\ref{proof:theorem:problem-coord}}).
\end{proof}

In order to allow nodes to obtain messages with high probabilities, while
keeping the effectiveness independent of the desired reliability, punishments must be coordinated,
such that after any deviation any node $i$ expects to
be punished by every in-neighbor during $\pd>0$ stages.  More precisely, for every node $i$, by letting
$$\mdel_i = \max_{j \in {\cal N}_i} \max_{k\in{\cal N}_i^{-1}} \del_k[i,j]$$
to be the maximum delay of accusations against $i$ towards any in-neighbor of $i$,
the protocol must define $\pd[i,j|k,i]$ for every $k \in {\cal N}_i^{-1}$ and $j \in {\cal N}_i$ in order to fulfill 
$$\pd[i,j|k,i] + \del_k[i,j] \geq \mdel_i  + \pd.$$ 
It is sufficient and convenient for the sake of simplicity to provide a definition of $\pd[i,j|k,l]$ for every $k\in{\cal N}$ 
and $l \in {\cal N}_k$ where $\del_k[i,j] < \infty$, such that every node stops reacting to a given defection in the same stage.
More precisely, 
\begin{definition}
\label{def:overlap}
For every $k \in {\cal N}$ and $l \in {\cal N}_k$ such that $d_k[i,j] < \infty$, if $k$ and $l$ observe the defection before
$\mdel_i + \pd$, i.e., $g = \max[\del_k[i,j],\del_l[i,j]] < \mdel_i + \pd$, then $k$ and $l$ react to a defection of $i$ from $j$:
$$\pd[i,j|k,l] = \mdel_i + \pd - g.$$
Otherwise,
$$\pd[i,j|k,l] = 0.$$
\end{definition}

This ensures that no node $k$ reacts to a defection of $i$ from $j$ after stage $\mdel_i + \pd$.

\subsection{Coordinated Full Indirect Reciprocity}
We now study the set of punishing strategies that use full indirect reciprocity. This requires the
existence of a path from every $j \in {\cal N}_i$ to every $k \in {\cal N}_i^{-1}$, which must not cross $i$.
Under some circumstances, the effectiveness of a Preconsistent assessment $(\vec{\sigma}^*,\vec{\mu}^*)$
that uses full indirect reciprocity does not increase with the reliability of the dissemination process. 
As seen in the previous section, this requires punishments to be coordinated, which we assume to be
defined as in Definition~\ref{def:overlap}.

The fact that accusations may be delayed has an impact on the effectiveness, which
is quantified in Lemma~\ref{lemma:priv-suff}. To prove this lemma, we first derive in Lemma~\ref{lemma:delay-equiv} 
an intermediate sufficient condition for the PDC Condition to be true. The lemma
simplifies the PDC Condition for the worst scenario, where all in-neighbors of a node $i$ begin
punishing $i$ for any defection simultaneously. The proofs assume that
punishing strategies are defined in a reasonable manner. More precisely, if in reaction to a defection of node $i$
other nodes increase the probabilities used towards out-neighbors other than $i$, then $i$ should never expect a large increase
in its reliability during the initial stages, before every in-neighbor starts punishing $i$. This intuition is captured
in Assumption~\ref{def:non-neg}.

\begin{definition}
\label{def:non-neg}
(\textbf{Assumption})
There exists a constant $\epsilon \in [0,1)$ such that,
for every $h\in {\cal H}$, $i \in {\cal N}$, $\vec{p}_i' \in {\cal P}_i$, $\vec{\sigma}' = (\sigma_i^*[h|\vec{p}_i'],\vec{\sigma}_{-i}^*)$,
and $r > 0$,
$$q_i[h,r|\vec{\sigma}^*] - q_i[h,r|\vec{\sigma}'] < \epsilon.$$
\end{definition}

\begin{lemma}
\label{lemma:delay-equiv}
If $(\vec{\sigma}^*,\mu^*)$ is Preconsistent, Assumption~\ref{def:non-neg} holds, and Inequality~\ref{eq:delay-equiv} 
is fulfilled for every $i \in {\cal N}$, $h_i \in {\cal H}_i$, and $h \in {\cal H}$ such 
that $\mu_i^*[h|h_i]>0$, then $(\vec{\sigma}^*,\vec{\mu}^*)$ is Sequentially Rational:
\begin{equation}
\label{eq:delay-equiv}
- \sum_{r=0}^{\mdel_i} \omega_i^r ((1-q_i[h,r|\vec{\sigma}^*])\gamma_i \bar{p}_i[h,r|\vec{\sigma}^*]  + \epsilon \beta_i)+ \sum_{r=\mdel_i+1}^{\mdel_i + \pd} \omega_i^r u_i[h,r | \vec{\sigma}^*]  \geq 0.
\end{equation}
\end{lemma}
\begin{proof}
By the definition of coordinated punishments if $i$ deviates in $\sigma_i'$ by dropping some subset of neighbors such that all players follow $\vec{\sigma}'=(\sigma_i',\vec{\sigma}_{-i}^*)$,
then in the worst scenario no node punishes $i$ in any of the first $\mdel_i$ stages. Therefore, by our assumptions, for
every $r \in \{1\ldots \mdel_i\}$,
$$u_i[h,r|\vec{\sigma}^*] - u_i[h,r|\vec{\sigma}'] \geq - \sum_{r=0}^{\mdel_i} \omega_i^r ((1-q_i[h,r|\vec{\sigma}^*])\gamma_i \bar{p}_i[h,r|\vec{\sigma}^*]  + \epsilon \beta_i).$$
Also, for every $r \in \{\mdel_i+1 \ldots \mdel_i + \pd\}$, every in-neighbor of $i$ punishes $i$, which by Lemma~\ref{lemma:noneib} from Appendix~\ref{sec:epidemic}
implies
$$u_i[h,r|\vec{\sigma}'] = 0.$$
Finally, for every $r \geq  \mdel_i +\pd +1$, every node ends its reaction to any defection of $i$ after $\mdel_i+\pd+1$ stages, implying that
$$u_i[h,r|\vec{\sigma}^*] = u_i[h,r|\vec{\sigma}'].$$

These three facts have the consequence that if Inequality~\ref{eq:delay-equiv} is fulfilled, then the PDC Condition holds.
Therefore, by Theorem~\ref{theorem:priv-drop}, $(\vec{\sigma}^*,\vec{\mu}^*)$ is Sequentially Rational.
(\hyperref[proof:lemma:delay-equiv]{Complete proof in Section~\ref{proof:lemma:delay-equiv}}).
\end{proof}

We can now derive a lower bound for the effectiveness
of the considered punishing strategies, in a similar fashion to Theorem~\ref{theorem:indir-suff}.
However, now a stronger assumption is made, defined in~\ref{def:priv-assum}.
The reasoning is similar in that after any history, if a node $i$ defects from some out-neighbor, then the reliability
$i$ would obtain during the initial stages when it is not being punished by all in-neighbors is not
significantly greater than the reliability $i$ would obtain in the subsequent stages, had $i$ not deviated
from the specified strategy. 
\begin{definition}
\label{def:priv-assum}
(\textbf{Assumption}).
There exists a constant $c>0$, such that, for every $i \in {\cal N}$, $r \in \{0 \ldots \mdel_i \}$, and $r' \in \{\mdel_i+1 \ldots \mdel_i+ \pd\}$,
$$q_i[h,r|\vec{\sigma}^*] \geq 1-c(1-q_i[h,r'|\vec{\sigma}^*]).$$
\end{definition}

\begin{lemma}
\label{lemma:priv-suff}
If $(\vec{\sigma}^*,\vec{\mu}^*)$ is Preconsistent, 
Assumptions~\ref{def:non-neg} and~\ref{def:priv-assum} hold, 
and Inequality~\ref{eq:priv-suff} is fulfilled for every $h$, $i \in {\cal N}$, and $r,r' \leq \mdel_i + \pd$ such that $q_i[h,r|\vec{\sigma}^*] <1$ and $q_i[h,r'|\vec{\sigma}^*] <1$,
then there exist $\omega_i \in (0,1)$ for every $i \in {\cal N}$ such that $(\vec{\sigma}^*,\vec{\mu}^*)$ is Sequentially Rational:
\begin{equation}
\label{eq:priv-suff}
\frac{\beta_i}{\gamma_i}> \bar{p}_i[h,r|\vec{\sigma}^*] \frac{1}{A}+\bar{p}_i[h,r'|\vec{\sigma}^*]\frac{1}{B - C},
\end{equation}
where
\begin{itemize}
  \item $A = 1 - \frac{\epsilon(\mdel_i+1)}{(1-q_i[h,r|\vec{\sigma}^*])\pd}$.
  \item $B=\frac{\pd}{c}$.
  \item $C=\frac{\epsilon(\mdel_i+1)}{1-q_i[h,r'|\vec{\sigma}^*]}$.
\end{itemize}
\end{lemma}
\begin{proof}
The proof considers two histories $h_1$ and $h_2$ that minimize the first and the second factors of Inequality~\ref{eq:delay-equiv},
respectively. Thus, if the following condition is true, then Inequality~\ref{eq:delay-equiv}  is true:
$$- \sum_{r=0}^{\mdel_i} \omega_i^r ((1-q_i^*[h_1,0|\vec{\sigma}^*])\gamma_i \bar{p}_i[h_1,0|\vec{\sigma}^*]  + \epsilon \beta_i)+ \sum_{r=\mdel_i+1}^{\mdel_i+\pd+1} \omega_i^r u_i[h_2,0 | \vec{\sigma}^*]  \geq 0.$$
After some manipulations, we conclude that the above condition is fulfilled if~\ref{eq:priv-suff} is true.
This implies by Lemma~\ref{lemma:delay-equiv} that if~\ref{eq:priv-suff} holds for every $h$, and $r,r' \leq \mdel_i + \pd$,
then Inequality~\ref{eq:delay-equiv} also holds for some $\omega_i \in (0,1)$ and $(\vec{\sigma}^*,\vec{\mu}^*)$ is Sequentially Rational.
(\hyperref[proof:lemma:priv-suff]{Complete proof in Section~\ref{proof:lemma:priv-suff}}).
\end{proof}

The main conclusion that can be drawn from this lemma is that if we pick the values of $\pd$ and $\epsilon$ such that
$\pd \geq \mdel + 1$ and $\epsilon \ll 1$, then we can simplify the above condition to what is expressed in Theorem~\ref{theorem:priv-effect},
where $\mdel = \max_{i \in {\cal N}} \mdel_i$ is the maximum delay of the monitoring mechanism.

\begin{theorem}
\label{theorem:priv-effect}
If $(\vec{\sigma}^*,\vec{\mu}^*)$ is Preconsistent, Assumptions~\ref{def:non-neg} and~\ref{def:priv-assum} hold for $\epsilon \ll 1$,
and $\pd \geq \mdel + 1$, then there exists a constant $c>0$
such that $\psi[\vec{\sigma}^*|\vec{\mu}^*] \supseteq (v,\infty)$, where
$$
v = \max_{i \in {\cal N}} \max_{h \in {\cal H}}\bar{p}_i[h,0|\vec{\sigma}^*](1+c).
$$
\end{theorem}
\begin{proof}
(\hyperref[proof:theorem:priv-effect]{See Section~\ref{proof:theorem:priv-effect}}).
\end{proof}

As in public monitoring, the effectiveness is close to optimal only if the initial expected costs of forwarding messages $\bar{p}_i$ are not
significantly smaller than the expected costs incurred after any history.
Provided this guarantee, if $\epsilon$ is small and $\pd$ is chosen to be at least of the order of the maximum delay between
out and in-neighbors of any node, then for any other punishing strategy, the effectiveness differs from the optimal by a constant factor.

Notice that, although we can adjust the value of $\pd$ to compensate for higher delays, it is desirable
to have a low maximum delay. First, this is due to the fact that higher values of $\pd$ correspond to harsher
punishments, which we may want to avoid, especially when monitoring is imperfect and honest nodes may wrongly
be accused of deviating. Second, a larger delay decreases the range of values of $\omega_i$ for each benefit-to-cost
ratio that ensures that punishing strategies are an equilibrium. In particular, we can derive from the proof of Lemma~\ref{lemma:priv-suff} the strict minimum $\omega_i$
for Grim-trigger to be an equilibrium. Under our assumptions, it is approximately given by
$$\omega_i \geq \sqrt[\mdel_i]{\frac{\gamma_i \bar{p}_i}{\beta_i}}.$$
For larger values of $\mdel_i$ and the same benefit-to-cost ratio, the minimum $\omega_i$ is also larger,
reducing the likelihood of punishments to persuade rational nodes to not deviate from the specified strategy.

\section{Discussion and Future Work}
\label{sec:conc}
From this analysis, we can derive several desirable properties of a fully distributed monitoring
mechanism for an epidemic dissemination protocol with asymmetric interactions, which uses
punishments as the main incentive. This mechanism is expected to operate on top of
an overlay network that provides a stable membership to each node. The results of this paper determine
that the overlay should optimally explore the tradeoff between maximal randomization and higher clustering coefficient.
The former is ideal for minimizing the latency of the dissemination process and fault tolerance, whereas the latter
is necessary to minimize the distances between the neighbors of each node, while maximizing the number of in-neighbors
of every node $i$ informed about any defection of $i$. The topology of this overlay should also fulfill the necessary
conditions identified in this paper. Furthermore, the analysis of private monitoring shows that
each accusation may be disseminated to a subset of nodes close to the accused node, without hindering
the effectiveness. As future work, we plan to extend this analysis by considering
imperfect monitoring, unreliable dissemination of accusations, malicious behavior, and churn. One possible application
of the considered monitoring mechanism would be to sustain cooperation in a P2P news recommendation system
such as the one proposed in~\cite{Boutet:13}. Due to the lower rate of arrival of news, a monitoring approach
may be better suited in this context.

\section*{Acknowledgements}
This work was partially supported by Funda\c{c}\~{a}o para a Ci\^{e}ncia e Tecnologia (FCT) 
via the INESC-ID multi-annual funding through the PIDDAC Program fund grant, under 
project PEst-OE/ EEI/ LA0021/ 2013, and via the project PEPITA (PTDC/EEI-SCR/2776/2012).

\bibliographystyle{splncs}
\bibliography{bibfile}

\newpage
\appendix


\section{Epidemic Model}
\label{sec:epidemic}
The probability that a node $i$ does not receive a message from $s$ when all players follow $\vec{p}$ can be defined recursively, as follows.
\begin{proposition}
\label{prop:non-delivery}
Define $\phi$ as follows: i) $\phi[R,\emptyset | \vec{p},L] = 1$ and ii) for $I \neq \emptyset$ and $R' = R \cup I \cup L$:
\begin{equation}
\label{eq:prod-non-delivery}
\begin{array}{ll}
\phi[R,I | \vec{p},L] = \sum_{H \subseteq {\cal N} \setminus R'} (P[I,H|\vec{p}] \cdot Q[{\cal N},R,I,H|\vec{p}]\cdot \phi[R \cup I, H | \vec{p},L]),
\end{array}                                     
\end{equation}
where
$$
\begin{array}{l}
P[I,H|\vec{p}] = \prod_{k \in H} (1 - \prod_{l \in I} (1 - p_l[k])).\\
Q[{\cal N},R,I,H|\vec{p}] = \prod_{k \in {\cal N} \setminus (H \cup R \cup I)} \prod_{l \in I} (1 - p_l[k]).
\end{array}
$$

Then, $\phi[\emptyset,\{s\}|\vec{p},L]$ is the probability that no node of $L$ receives a message disseminated by $s$. In particular, 
$q_i[\vec{p}] = \phi[\emptyset,\{s\}|\vec{p},\{i\}]$.
\end{proposition}
\begin{proof}
(\textbf{Justification})
The considered epidemics model is very similar to the Reed-Frost model\,\cite{Abbey:52}, where dissemination is performed by having nodes
forwarding messages with independent probabilities. The main difference is that the probability of forwarding
each message is determined by a vector $\vec{p}$, instead of being identical for every node. This implies that the
dissemination process can be modeled as a sequence of steps, such that, at every step, there is a set $I$ of nodes
infected in the last step, a set $R$ of nodes infected in previous steps other than the last, and a set $S$ of susceptible nodes.
Given $R$ and $I$, the probability of the set $H \subseteq {\cal N} \setminus (I \cup R \cup L)$
containing exactly the set of nodes infected at the current step is 
\begin{equation}
\label{eq:non-delivery}
P[I,H|\vec{p}]\cdot Q[{\cal N},I,H|\vec{p}] = \prod_{k \in H} (1 - \prod_{l \in I} (1 - p_l[k])) \cdot \prod_{k \in {\cal N} \setminus (H \cup R \cup I)} \prod_{l \in I} (1 - p_l[k]).
\end{equation}

That is, every node $i \in H$ is infected with a probability equal to $1$ minus the probability of no node of $I$ choosing $i$,
and these probabilities are all independent. Furthermore, all nodes of $I$ do not infect any node of ${\cal N} \setminus (H \cup R \cup I)$.
We can characterize $\phi$ with a weighted tree, where nodes correspond to a pair $(R,I)$.
Moreover, for every parent node $(R,I)$, each child node corresponds to a pair $(R\cup I,H)$ for every $H \subseteq {\cal N} \setminus (R \cup I \cup L)$.
The root node is the pair $(\emptyset,\{s\})$ and leaf nodes are in the form $(R,\emptyset)$ for every $R \subseteq \{s \} \cup {\cal N} \setminus L$.
The weight of the transition from $(R,I)$ to $(R\cup I,H)$ is given by~\ref{eq:non-delivery}, which is the probability of exactly
the nodes of $H$ being infected among every node of ${\cal N} \setminus (R \cup I \cup L)$. The sum of these factors for any path
from $(\emptyset,\{s\})$ to $(R,\emptyset)$ gives the probability of exactly the nodes of $R$ being infected in a specific order.
By summing over all leaf nodes in the form $(R,\emptyset)$, we have the total probability of exactly the nodes of $R$ being infected.
Finally, by summing over all possible $R \subseteq \{s\} \cup {\cal N} \setminus L$,  
we obtain the probability of no node in $L$ being infected. In particular, $q_i[\vec{p}] = \phi[\emptyset,\{s\} | \vec{p},\{i\}]$.
\end{proof}

The following are some useful axioms for the proofs, for any $\vec{p}$, $R$, $I$, $L$, and $R' = R \cup I \cup L$:
\begin{equation}
\label{eq:epid-1}
\sum_{H \subseteq {\cal N} \setminus R'} P[I,H|\vec{p}]\cdot Q[{\cal N},R,I,H|\vec{p}] = 1.
\end{equation}

If $(A,B)$ is a partition of ${\cal N}$, then
\begin{equation}
\label{eq:epid-2}
\begin{array}{ll}
&\sum_{H \subseteq {\cal N} \setminus R'} P[I,H|\vec{p}]\cdot Q[{\cal N},R,I,H|\vec{p}]=\\
=&\sum_{H_1 \subseteq A \setminus R'} P[I,H_1|\vec{p}]Q[A,R,I,H_1|\vec{p}] \cdot\\
\cdot&\sum_{H_2 \subseteq B \setminus R'} P[I,H_2|\vec{p}]\cdot Q[B,R,I,H_1|\vec{p}]\\
\end{array}
\end{equation}

\begin{equation}
\label{eq:epid-3}
j \in I \Rightarrow P[I,H|\vec{p}] \leq P[I \setminus \{j\},H|\vec{p}].
\end{equation}

\subsection{Deterministic Delivery}
Lemma~\ref{lemma:prob-1} proves the straightforward fact that if there is a path from $s$ to some node $i$
where all nodes along the path forward messages with probability $1$, then $q_i=0$.

\begin{lemma}
\label{lemma:prob-1}
For any $\vec{p} \in {\cal P}$, if there exists $i \in {\cal N}$ and $x \in \pth[s,i]$ such that for every $r \in \{0\ldots |x|-1\}$ we have
$p_{x_r}[x_{r+1}]= 1$, then $q_i[\vec{p}] = 0$.
\end{lemma}
\begin{proof}
The proof goes by induction on $r$ where the induction hypothesis is
that, for every $r \in \{0\ldots |x|-2\}$, $R \subseteq {\cal N} \cup \{s\} \setminus \{i\}$, and $I \subseteq {\cal N} \cup \{s\} \setminus \{i\} \cup R$,
such that:
\begin{itemize}
  \item $x_r \in I$,
  \item for every $r' \in \{0\ldots r\}$, $x_{r'} \in R \cup I$,
  \item for every $r' \in \{r+1\ldots |x|-1\}$, $x_{r'} \in {\cal N} \setminus (R \cup I)$,
\end{itemize} 
we have $\phi[R,I|\vec{p},\{i\}] = 0$.

Consider the base case for $r=|x|-2$ and let $R' = R \cup I \cup \{i\}$. In this case, for $j= x_{r}$,
$$
\begin{array}{lll}
\phi[R,I | \vec{p},\{i\}] &=& \sum_{H \subseteq {\cal N} \setminus R'} (P[I,H|\vec{p}] \cdot Q[{\cal N},R,I,H|\vec{p}] \cdot \phi[R \cup I, H | \vec{p},\{i\}])\\
                                     &\\
                                     &= &\sum_{H \subseteq {\cal N} \setminus R'} (P[,I,H|\vec{p}] \cdot Q[{\cal N}\setminus \{i\},R,I,H|\vec{p}]\\
                                     &&\prod_{l \in I \setminus \{j\}} (1 - p_l[i])(1-p_j[i])\phi[R \cup I, H | \vec{p},\{i\}])\\
				&\\
                                     &= &\sum_{H \subseteq {\cal N} \setminus R' } (P[,I,H|\vec{p}] \cdot Q[{\cal N}\setminus \{i\},R,I,H|\vec{p}]\\
				& & 0 \cdot \phi[R \cup I, H | \vec{p},\{i\}])\\
				&\\
                                     &= &0.
\end{array}
$$

This proves the base case. Assume now that the hypothesis is true for every $r' \in \{0 \ldots r\}$ and some $r \in \{1 \ldots |x|-2\}$.
Let $a = x_{r-1}$ and $b = x_{r}$, let $R_1 = R \cup I$ and $R_2 = R_1 \cup H$. We thus have
$$
\begin{array}{lll}
\phi[R,I | \vec{p},\{i\}] &= &\sum_{H \subseteq {\cal N} \setminus (R_1 \cup \{i\})}(P[I,H|\vec{p}]\cdot Q[{\cal N},R,I,H|\vec{p}]\cdot \phi[R_1, H | \vec{p},\{i\}])\\
                                      &&\\
                                      &= &\sum_{H \subseteq {\cal N} \setminus (R_1 \cup \{i,b\})} (P[I,H|\vec{p}] \cdot Q[{\cal N} \setminus \{b\},R,I,H|\vec{p}]\\
                                      & &\prod_{l \in I \setminus \{a\}} (1 - p_l[b]) (1 - p_a[b]) \phi[R_1, H | \vec{p},\{i\}])\\
                                      &+&\sum_{H \subseteq {\cal N} \setminus (R_1 \cup \{i,b\})} (P[I,H \cup \{b\}|\vec{p}] \cdot Q[{\cal N},R,I,H\cup \{b\}|\vec{p}] \cdot \phi[R_1, H\cup \{b\} | \vec{p},\{i\}])\\
                                      &&\\
                                      &= &\sum_{H \subseteq {\cal N} \setminus (R_1 \cup \{i,b\})} (P[I,H|\vec{p}] \cdot Q[{\cal N} \setminus \{b\},R,I,H|\vec{p}] \cdot 0\cdot \phi[R_1, H | \vec{p},\{i\}])\\
                                      &+&\sum_{H \subseteq {\cal N} \setminus (R_1 \cup \{i,b\})} (P[I,H \cup \{b\}|\vec{p}] \cdot Q[{\cal N},R,I,H\cup \{b\}|\vec{p}] \cdot 0)\\
                                      &&\\
                                      &=&0.
\end{array}
$$
This proves the induction step for $r-1$. Consequently, since $s = x_0$ and $x_r \in {\cal N}$ for every $r \in \{1 \ldots |x|-1\}$,
$$q_i[\vec{p}] = \phi[\emptyset,\{s\}|\vec{p},\{i\}] = 0.$$
\end{proof}
\newpage
\subsection{Positive Reliability}
Lemma~\ref{lemma:pprob} shows that if every node forwards messages with a positive probability, then every node
of the graph receives a message with positive probability as well.
\begin{lemma}
\label{lemma:pprob}
For any $\vec{p} \in {\cal P}$, if there exists $i \in {\cal N}$ and $x \in \pth[s,i]$ such that for every $r \in \{0\ldots |x|-1\}$ we have
$p_{x_r}[x_{r+1}]>0$, then $q_i[\vec{p}] <1$.
\end{lemma}
\begin{proof}
The proof goes by induction on $r$ where the induction hypothesis is
that, for every $r \in \{0\ldots |x|-2\}$, $R \subseteq {\cal N} \cup \{s\} \setminus \{i\}$, and $I \subseteq {\cal N} \cup \{s\} \setminus \{i\} \cup R$,
such that:
\begin{itemize}
  \item $x_r \in I$,
  \item for every $r' \in \{0\ldots r\}$, $x_{r'} \in R \cup I$,
  \item for every $r' \in \{r+1\ldots |x|-1\}$, $x_{r'} \in {\cal N} \setminus (R \cup I)$,
\end{itemize} 
we have $\phi[R,I|\vec{p},\{i\}] < 1$.

Consider the base case for $r=|x|-2$ and let $R_1 = R \cup I \cup \{i\}$. In this case, by Axiom~\ref{eq:epid-1}, for $j= x_{r}$,
$$
\begin{array}{lll}
\phi[R,I | \vec{p},\{i\}] &= &\sum_{H \subseteq {\cal N} \setminus R_1} (P[I,H|\vec{p}] \cdot Q[{\cal N},R,I,H|\vec{p}] \cdot \phi[R \cup I, H | \vec{p},\{i\}])\\
                                     & & \\
                                     &= &\sum_{H \subseteq {\cal N} \setminus R_1}(P[I,H|\vec{p}] \cdot Q[{\cal N}\setminus \{i\},R,I,H|\vec{p}] \cdot \\
                                     & & \prod_{l \in I \setminus \{j\}} (1 - p_l[i])(1-p_j[i])\phi[R \cup I, H | \vec{p},\{i\}])\\
				&&\\
                                     &< &\sum_{H \subseteq {\cal N} \setminus R_1}P[I,H|\vec{p}] \cdot Q[{\cal N},R,I,H|\vec{p}]\\
				&&\\
                                     &=&1.
\end{array}
$$
This proves the base case. Assume now that the hypothesis is true for every $r' \in \{0 \ldots r\}$ and some $r \in \{1 \ldots |x|-2\}$.
Let $a = x_{r-1}$ and $b = x_{r}$. Consider that $R_1 = R \cup I$ and $R_2 = R \cup H \cup I$.

We thus have by Axioms~\ref{eq:epid-1} and~\ref{eq:epid-3},
$$
\begin{array}{lll}
\phi[R,I | \vec{p},\{i\}] &= &\sum_{H \subseteq {\cal N} \setminus (R_1\cup \{i\})} (P[I,H|\vec{p}] \cdot Q[{\cal N},R,I,H|\vec{p}] \phi[R_1, H | \vec{p},\{i\}])\\
                                      &&\\
                                      &= &\sum_{H \subseteq {\cal N} \setminus (R_1 \cup \{i,b\})} (P[I,H|\vec{p}] \cdot Q[{\cal N} \setminus \{b\},R,I,H|\vec{p}]\\
                                      & & \prod_{l \in I} (1 - p_l[k])\prod_{l \in I \setminus \{a\}} (1 - p_l[k])(1 - p_a[b])\phi[R_1, H | \vec{p},\{i\}])\\
                                      &+&\sum_{H \subseteq {\cal N} \setminus (R_1 \cup \{i,b\})} (P[I,H\cup \{b\}|\vec{p}] \cdot Q[{\cal N},R,I,H \cup \{b\}|\vec{p}] \cdot \phi[R_1, H\cup \{b\} | \vec{p},\{i\}])\\
                                      &&\\
                                      &< &\sum_{H \subseteq {\cal N} \setminus (R_1 \cup \{i,b\})} (P[I,H|\vec{p}] \cdot Q[{\cal N} \setminus \{b\},R,I,H|\vec{p}]\\
                                      & & \prod_{l \in I} (1 - p_l[k])\prod_{l \in I \setminus \{a\}} (1 - p_l[k]) \cdot 1 \cdot \phi[R_1, H | \vec{p},\{i\}])\\
                                      &+&\sum_{H \subseteq {\cal N} \setminus (R_1 \cup \{i,b\})} (P[I \setminus \{a\},H\cup \{b\}|\vec{p}] \cdot Q[{\cal N},R,I \setminus \{a\},H \cup \{b\}|\vec{p}])\\
                                      &&\\
                                      &\leq &\sum_{H \subseteq {\cal N} \setminus (R_1 \cup \{i,b\})} (P[I \setminus \{a\},H|\vec{p}] \cdot Q[{\cal N},R,I \setminus \{a\},H|\vec{p}])\\
                                      &+&\sum_{H \subseteq {\cal N} \setminus (R_1 \cup \{i\})} (P[I \setminus \{a\},H\cup \{b\}|\vec{p}] \cdot Q[{\cal N},R,I \setminus \{a\},H \cup \{b\}|\vec{p}])\\
                                      &&\\
                                      &= &\sum_{H \subseteq {\cal N} \setminus (R_1 \cup \{i\})} (P[I \setminus \{a\},H|\vec{p}] \cdot Q[{\cal N},R,I \setminus \{a\},H|\vec{p}])\\
                                      &&\\
                                      & = &1.
\end{array}
$$
This proves the induction step for $r-1$. Consequently, since $s = x_0$ and $x_r \in {\cal N}$ for every $r \in \{1 \ldots |x|-1\}$,
$$q_i[\vec{p}] = \phi[\emptyset,\{s\}|\vec{p},\{i\}] < 1.$$
\end{proof}
\newpage
\subsection{Null Reliability}
Lemma~\ref{lemma:noneib} shows that if every in-neighbor of a node $i$ does not forward messages to $i$, then $q_i = 0$.
\begin{lemma}
\label{lemma:noneib}
If $\vec{p} \in {\cal P}$ is defined such that for some $i \in {\cal N}$ and for every $j \in {\cal N}_i^{-1}$
$p_j[i] = 0$, then $q_i[\vec{p}] = 1$.
\end{lemma}
\begin{proof}

The proof goes by induction on $r$ where the induction hypothesis is
that for every $r \in \{0\ldots |{\cal N}|\}$, $R \subseteq {\cal N} \cup \{s\} \setminus \{i\}$, and $I \subseteq {\cal N} \cup \{s\} \setminus \{i\} \cup R$
such that $|R|+|I| \leq |{\cal N}|+1-r$, we have $\phi[R,I|\vec{p},\{i\}] = 1$.

The base case is for $r=0$, where we have by the definition of $\vec{p}$.
$$\phi[R,I|\vec{p},\{i\}] =  \prod_{l \in I} (1 - p_l[i]) \phi[R \cup I, \emptyset | \vec{p},\{i\}] =  \prod_{l \in I} (1 - 0)\cdot 1 = 1.$$

Assume the induction hypothesis for any $r \in \{0 \ldots|{\cal N}|-1\}$. Consider any two $R$ and $I$ defined
as above for $r+1$, such that $|R|+|I| = |{\cal N}|-r$. Let $R_1 = R \cup I \cup \{i\}$ and $R_2 = H \cup R \cup I$.
It is true by Axiom~\ref{eq:epid-1} that:
$$
\begin{array}{ll}
\phi[R,I | \vec{p},\{i\}] = &\sum_{H \subseteq {\cal N} \setminus R_1}P[I,H|\vec{p}] Q[{\cal N},R,I,H|\vec{p}]\phi[R \cup I, H | \vec{p},\{i\}])\\
                                     &\\
                                     &=\sum_{H \subseteq {\cal N} \setminus R_1}P[I,H|\vec{p}] Q[{\cal N},R,I,H|\vec{p}]\cdot 1\\
                                     &\\
                                     & = 1.
\end{array}                                     
$$
Therefore, for $r=|{\cal N}|$, $q_i[\vec{p}] = \phi[\emptyset,\{s\}|\vec{p},\{i\}]=1$.
\end{proof}
\newpage
\subsection{Uniform Reliability}
Lemma~\ref{lemma:bottleneck-impact} shows that for some node $i$ and for every path $x$ from $s$ to and $i$,
every node of $x$ does not change its probability from $\vec{p}$ to $\vec{p}'$, then $q_i[\vec{p}] = q_i[\vec{p}']$.
As an intermediate step, Lemma~\ref{lemma:bt-aux} proves that the reliability is the same whenever $I$ and $R$
contain the same set of nodes from any path from $s$ to $i$.

For any $i \in {\cal N}$, $\vec{p} \in {\cal P}$, and $K \subseteq {\cal N}$, let 
$$D_i = \{j \in {\cal N} \cup \{s\} \setminus \{i\}| \pth[j,i] \neq \emptyset\},$$
and:
$$p[I,k] = 1 - \prod_{j \in I}(1-p_j[k]).$$

Therefore, it is possible to write:
$$
\begin{array}{l}
P[I,H|\vec{p}] = \prod_{k \in H} p[I,k].\\
Q[{\cal N},R,I,H|\vec{p}] = \prod_{k \in {\cal N} \setminus (R \cup I \cup H)} (1-p[I,k]).
\end{array}
$$

For any $L_1,L_2 \subseteq {\cal N} \setminus D_i$ such that $L_1 \cap L_2 = \emptyset$, 
since for every $j \in L$ we have $\pth[j,i] = \emptyset$ and $p_j[k] =0$ for every $k \in D_i$, and for every $H \subseteq D_i$,
\begin{equation}
\label{eq:bta-1}
\begin{array}{l}
p[I \cup L_1,k] = p[I,k].\\
P[I \cup L_1,H|\vec{p}] = P[I,H|\vec{p}].\\
Q[{\cal N},R \cup L_1,I \cup L_2,H|\vec{p}] = Q[{\cal N},I,H|\vec{p}].
\end{array}
\end{equation}

\begin{lemma}
\label{lemma:bt-aux}
For every $\vec{p} \in {\cal P}$, $i \in {\cal N}$, $R \subseteq {\cal N} \cup \{s\} \setminus \{i\}$, $I \subseteq D_i \setminus R$,
and $L_1,L_2 \subseteq {\cal N} \setminus (D_i \cup R)$ such that $L_1 \cap L_2 = \emptyset$,
$$\phi[R\cup L_1,I\cup L_2|\vec{p},\{i\}] = \phi[R,I|\vec{p},\{i\}].$$
\end{lemma}
\begin{proof}
Fix $\vec{p}$ and $i$. First notice that for every $R \subseteq {\cal N}$ and $L \subseteq {\cal N} \setminus (D_i \cup R)$,
\begin{equation}
\label{eq:bta-2}
\phi[R,L|\vec{p},\{i\}] = 1.
\end{equation}

We now prove by induction that for every $R \subseteq {\cal N} \cup \{s\}$, $I \subseteq D_i \setminus R$,
and $L_1, L_2\subseteq {\cal N} \setminus (D_i \cup R)$ such that $L_1\cap L_2 = \emptyset$,
$$\phi[R, L_1,I\cup L_2|\vec{p},\{i\}] = \phi[R,I|\vec{p},\{i\}].$$

The induction goes on $r \in \{0\ldots |{\cal N}|\}$, where $|R| + |I| + |L| = |{\cal N}| + 1 - r$, where $L = L_1 \cup L_2$.
For $r=0$, by Axiom~\ref{eq:epid-1}, \ref{eq:bta-1} and~\ref{eq:bta-2}, we can write:
\begin{equation}
\label{eq:bta-3}
\begin{array}{ll}
\phi[R \cup L_1,I\cup L_2 | \vec{p},\{i\}] &= (1 - p[I\cup L_1,i])\phi[R\cup I \cup L,\emptyset | \vec{p},\{i\}]\\
                                                &= (1-p[I,i])\\
                                                &= (1 - p[I,i])\sum_{H \subseteq L } \phi[R \cup I,H|\vec{p},\{i\}]\\
                                                &= (1 - p[I,i])\sum_{H \subseteq {\cal N} \setminus (R \cup I \cup \{i\})}(P[I,H|\vec{p}]\cdot Q[{\cal N},R,I,H|\vec{p}]\\
                                                &\cdot \phi[R \cup I,H|\vec{p},\{i\}])\\
                                                &= \phi[R,I | \vec{p},\{i\}].
\end{array}
\end{equation}
This proves the base case.

Assume the induction hypothesis for every $r' \in  \{0 \ldots r\}$  and $r \in \{0\ldots |{\cal N}|-1\}$.
By Axioms~\ref{eq:epid-1} and~\ref{eq:epid-2}, and by~\ref{eq:bta-1} and the induction hypothesis, we can write:
$$
\begin{array}{ll}
\phi[R \cup L_1,I\cup L_2 | \vec{p},\{i\}] &=\sum_{H \subseteq {\cal N} \setminus (R \cup I \cup L \cup \{i\})}(\\
                                                &P[I \cup L_2,H|\vec{p}] \cdot Q[{\cal N},R \cup L_1,I \cup L_2,H|\vec{p}]  \cdot\\
                                                & \phi[R\cup I \cup L,H | \vec{p},\{i\}])\\
                                                &\\
                                                &=\sum_{H_1 \subseteq D_i \setminus (R \cup I \cup \{i\})}(\\
                                                &P[I \cup L_2,H_1|\vec{p}] \cdot Q[D_i,R \cup L_1,I \cup L_2,H_1|\vec{p}]\\
                                                &\sum_{H_2 \subseteq {\cal N} \setminus (R \cup L \cup \{i\} \cup D_i)}(\\
                                                &P[I \cup L_2,H_2|\vec{p}] \cdot Q[{\cal N} \setminus D_i,R \cup L_1,I \cup L_2,H_2|\vec{p}]\\
                                                & \phi[R\cup I \cup L,H_1 \cup H_2 | \vec{p},\{i\}]))\\
                                                &\\
                                                &=\sum_{H_1 \subseteq D_i \setminus (R \cup I \cup L \cup \{i\})}(\phi[R\cup I,H_1| \vec{p},\{i\}]\\
                                                &P[I,H_1|\vec{p}] \cdot Q[D_i,R \cup L_1,I \cup L_2,H_1|\vec{p}]\\
                                                &\sum_{H_2 \subseteq {\cal N} \setminus (R \cup L \cup \{i\} \cup D_i)}(\\
                                                &P[I \cup L_2,H_2|\vec{p}] \cdot Q[{\cal N} \setminus D_i,R \cup L_1,I \cup L_2,H_2|\vec{p}]))\\
                                                &\\
                                                &=\sum_{H_1 \subseteq D_i \setminus (R \cup I \cup L \cup \{i\})}(\phi[R\cup I,H_1| \vec{p},\{i\}]\\
                                                &P[I,H_1|\vec{p}] \cdot Q[D_i,R,I,H_1|\vec{p}])\\
                                                &\\
                                                &=\sum_{H_1 \subseteq D_i \setminus (R \cup I \cup L \cup \{i\})}(\\
                                                &P[I ,H_1|\vec{p}] \cdot Q[D_i,R,I ,H_1|\vec{p}]\\
                                                &\sum_{H_2 \subseteq {\cal N} \setminus (R \cup L \cup \{i\} \cup D_i)}(\\
                                                &P[I,H_2|\vec{p}] \cdot Q[{\cal N} \setminus D_i,R,I,H_2|\vec{p}]\phi[R\cup I,H_1 \cup H_2| \vec{p},\{i\}]))\\
                                                &\\
                                                &=\sum_{H \subseteq {\cal N} \setminus (R \cup I \cup \{i\})}(\\
                                                &P[I,H|\vec{p}] \cdot Q[{\cal N},R,I,H|\vec{p}] \cdot \phi[R \cup I,H | \vec{p},\{i\}])\\
                                                &\\
                                                &=\phi[R,I|\vec{p},\{i\}].                                   
\end{array}
$$
This proves the result.
\end{proof}

\newpage
\begin{lemma}
\label{lemma:bottleneck-impact}
Let $\vec{p},\vec{p}' \in {\cal P}$ be any two profiles of probabilities such that 
for some $i\in{\cal N}$, for every $x \in \pth[s,i]$, and for every $j \in x$, 
$\vec{p}_j = \vec{p}'_j$. Then, $q_i[\vec{p}] = q_i[\vec{p}']$.
\end{lemma}
\begin{proof}
Assume this to be the case for a fixed $i$, $\vec{p}$, and $\vec{p}'$. Then, for every $x \in \pth[s,i]$ and $j \in x$,
it is true that, for every $k \in {\cal N}_j^{-1}$, there exists $x' \in \pth[s,i]$
such that 
\begin{equation}
\label{eq:bi-1}
k \in x' \land p_k[j] = p_k'[j].
\end{equation}

Define $p'[I,k]$ as in~\ref{eq:bta-1}, but for $\vec{p}'$.
Condition~\ref{eq:bi-1} implies that for every $I \subseteq {\cal N} \cup \{s\}$ and $k \in D_i \cup \{i\}$:
\begin{equation}
\label{eq:bi-2}
p'[I,k] = p[I,k].
\end{equation}

The rest of the proof is performed by induction on $r$ where the induction hypothesis is
that for every $r \in \{0\ldots|{\cal N}|\}$, $R \subseteq {\cal N} \cup \{s\} \setminus \{i\}$, and $I \subseteq {\cal N} \cup \{s\} \setminus \{i\} \cup R$
such that $|R|+|I| \leq |{\cal N}|+1-r$, we have $\phi[R,I|\vec{p},\{i\}] = \phi[R,I|\vec{p}',\{i\}]$.

The base case is for $r=0$, where we have by~\ref{eq:bi-2} and the definition of $\vec{p}$ and $\vec{p}'$:
$$
\begin{array}{ll}
\phi[R,I|\vec{p},\{i\}] & = p[I,i] \phi[R \cup I, \emptyset | \vec{p},\{i\}]\\
                                   &= p'[I,i] \phi[R \cup I, \emptyset | \vec{p}',\{i\}]\\
                                   & = \phi[R,I|\vec{p}',\{i\}].
\end{array}
$$

This proves the base case. Now, assume the induction hypothesis for some $r \in \{0\ldots |{\cal N}|-1\}$.
It is true by Lemma~\ref{lemma:bt-aux}, by Axioms~\ref{eq:epid-1} and~\ref{eq:epid-2}, by~\ref{eq:bi-2}, and the induction hypothesis that:
$$
\begin{array}{ll}
\phi[R,I | \vec{p},\{i\}] &=\sum_{H \subseteq {\cal N} \setminus (R \cup I \cup \{i\})}( P[I,H|\vec{p}] \cdot Q[{\cal N},R,I,H|\vec{p}] \cdot\phi[R\cup I,H | \vec{p},\{i\}])\\
                                                &\\
                                                   &=\sum_{H_1 \subseteq D_i \setminus (R \cup I \cup \{i\})}(\\
                                                &P[I,H_1|\vec{p}] \cdot Q[D_i,R,I,H_1|\vec{p}]\\
                                                &\sum_{H_2 \subseteq {\cal N} \setminus (R \cup \{i\} \cup D_i)}(\\
                                                &P[I,H_2|\vec{p}] \cdot Q[{\cal N} \setminus D_i,R,I,H_2|\vec{p}]\\
                                                & \phi[R\cup I,H_1 \cup H_2 | \vec{p},\{i\}]))\\
                                                &\\
                                                &=\sum_{H_1 \subseteq D_i \setminus (R \cup I \cup \{i\})}(\phi[R\cup I,H_1| \vec{p},\{i\}]\\
                                                &P[I,H_1|\vec{p}] \cdot Q[D_i,R,I,H_1|\vec{p}]\\
                                                &\sum_{H_2 \subseteq {\cal N} \setminus (R \cup \{i\} \cup D_i)}(\\
                                                &P[I,H_2|\vec{p}] \cdot Q[{\cal N} \setminus D_i,R,I,H_2|\vec{p}]))\\
                                                &\\
                                                &=\sum_{H_1 \subseteq D_i \setminus (R \cup I \cup \{i\})}(\phi[R\cup I,H_1| \vec{p}',\{i\}]\\
                                                &P[I,H_1|\vec{p}'] \cdot Q[D_i,R,I,H_1|\vec{p}'])\\
                                                &\\
                                                &=\sum_{H_1 \subseteq D_i \setminus (R \cup I \cup \{i\})}(\\
                                                &P[I ,H_1|\vec{p}'] \cdot Q[D_i,R,I ,H_1|\vec{p}']\\
                                                &\sum_{H_2 \subseteq {\cal N} \setminus (R \cup L \cup \{i\} \cup D_i)}(\\
                                                &P[I,H_2|\vec{p}'] \cdot Q[{\cal N} \setminus D_i,R,I,H_2|\vec{p}']\phi[R\cup I,H_1 \cup H_2| \vec{p}',\{i\}]))\\
                                                &\\
                                                &=\sum_{H \subseteq {\cal N} \setminus (R \cup I \cup \{i\})}(\\
                                                &P[I,H|\vec{p}'] \cdot Q[{\cal N},R,I,H|\vec{p}'] \cdot \phi[R \cup I,H | \vec{p}',\{i\}])\\
                                                &\\
                                                &=\phi[R,I|\vec{p}',\{i\}].                                   
\end{array}
$$
This concludes the proof by induction. Therefore, for $r=|{\cal N}|$, 
$$q_i[\vec{p}] = \phi[\emptyset,\{s\}|\vec{p},\{i\}]=\phi[\emptyset,\{s\}|\vec{p}',\{i\}] = q_i[\vec{p}'].$$
\end{proof}

\subsection{Single Impact}
Lemma~\ref{lemma:single-impact} provides an upper bound for the impact in the reliability when a single in-neighbor $j$ punishes node $i$.

\begin{lemma}
\label{lemma:single-impact}
For every $i \in {\cal N}$, $j \in {\cal N}_i^{-1}$, $\vec{p} \in {\cal P}$ such that $p_j[i] < 1$ and $q_i[\vec{p}]>0$, if $\vec{p}'$ is the profile
where only $j$ deviates from $p_j[i]$ to $p_j[i]' < p_j[i]$, then
$$q_i[\vec{p}'] \leq q_i[\vec{p}] \frac{1-p_{j}[i]'}{1-p_{j}[i]}.$$
\end{lemma}
\begin{proof}
Fix $i$, $j$, $\vec{p}$, and $\vec{p}'$.
The proof shows by induction that for every $r \in \{0 \ldots |{\cal N}|\}$, $R \subseteq {\cal N} \cup \{s\} \setminus \{i\}$ and
$I \subseteq  {\cal N} \cup \{s\} \setminus \{i\} \cup R$ such that $|R|+|I|- r =|{\cal N}|-1$: 
$$\phi[R,I|\vec{p}',\{i\}] \leq \phi[R,I|\vec{p},\{i\}] \frac{1-p_{j}'[i]}{1-p_{j}[i]}.$$

Notice that by the definition of $\phi$, if $j \in R$, then
\begin{equation}
\label{eq:si-1}
\phi[R,I|\vec{p}',\{i\}] = \phi[R,I|\vec{p},\{i\}].
\end{equation}

For the base case $r=0$, if $j \in R$, then the result follows immediately.
Thus, consider that $j \in I$. We can write
$$
\begin{array}{ll}
\phi[R,I|\vec{p}',\{i\}] & = p[I,i] \phi[R \cup I, \emptyset | \vec{p},\{i\}]\\
                                   &= p[I \setminus \{j\},i] \cdot p'[\{j\},i]\\
                                   &= p[I \setminus \{j\},i] \cdot (1-p_j[i])\frac{1-p_{j}'[i]}{1-p_{j}[i]}\\
                                   &= p[I,i]\frac{1-p_{j}'[i]}{1-p_{j}[i]}\\
                                   & = \phi[R,I|\vec{p},\{i\}]\frac{1-p_{j}'[i]}{1-p_{j}[i]}.
\end{array}
$$
This proves the induction step for $r=0$. Assume now that the induction hypothesis
is true for every $r' \in \{0 \ldots r\}$ and for some $r \in \{0 \ldots|{\cal N}|-1\}$.

If $j \notin I$, then by the induction hypothesis and by~\ref{eq:si-1},
$$
\begin{array}{ll}
\phi[R,I|\vec{p}',\{i\}] &=\sum_{H \subseteq {\cal N} \setminus (R \cup I \cup \{i\})}(P[I,H|\vec{p}']\cdot Q[{\cal N},R,I,H|\vec{p}'] \cdot \phi[R\cup I,H | \vec{p}',\{i\}])\\
                                                &\\
                                                & = \sum_{H \subseteq {\cal N} \setminus (R \cup I \cup \{i\})}(P[I,H|\vec{p}]\cdot Q[{\cal N},R,I,H|\vec{p}] \cdot \phi[R\cup I,H | \vec{p}',\{i\}])\\
                                                &\\
                                               & \leq \sum_{H \subseteq {\cal N} \setminus (R \cup I \cup \{i\})}(P[I,H|\vec{p}]\cdot Q[{\cal N},R,I,H|\vec{p}] \cdot \phi[R\cup I,H | \vec{p},\{i\}] \frac{1-p_{j}'[i]}{1-p_{j}[i]})\\
                                                &\\
                                                & \leq \phi[R,I|\vec{p},\{i\}] \frac{1-p_{j}'[i]}{1-p_{j}[i]}.
\end{array}
$$
For the final case where $j \in I$, we have by~\ref{eq:si-1}:
$$
\begin{array}{ll}
\phi[R,I|\vec{p}',\{i\}] &=\sum_{H \subseteq {\cal N} \setminus (R \cup I \cup \{i\})}(P[I,H|\vec{p}']\cdot Q[{\cal N},R,I,H|\vec{p}'] \cdot \phi[R\cup I,H | \vec{p}',\{i\}])\\
                                   &\\
                                   &= \sum_{H \subseteq {\cal N} \setminus (R \cup I \cup \{i\})}(\prod_{k \in H}p[I,k] \\
                                   & \prod_{k \in {\cal N} \setminus (H \cup R \cup I \cup \{i\})} (1- p[I,k])p[I\setminus \{j\},i] (1 - p_j[i])\frac{1-p_{j}'[i]}{1-p_{j}[i]}\phi[R\cup I,H | \vec{p},\{i\}])\\
                                   &\\
                                   & = \phi[R,I|\vec{p},\{i\}] \frac{1-p_{j}'[i]}{1-p_{j}[i]}.
\end{array}
$$
This concludes the proof.
\end{proof}

\newpage
\section{Public Monitoring}
\label{sec:proof:public}

\subsection{Evolution of the Network}
\label{sec:proof-pub-evol}

\subsubsection{Proof of Lemma~\ref{lemma:corr-0}.}
\label{proof:lemma:corr-0}
For every $h \in {\cal H}$, $r \in \{1 \ldots \pd-1\}$, $i \in {\cal N}$,  and $j \in {\cal N}_i$, 
$$\ds_i[j|h_r^*] = \{(k_1,k_2,r'+r) | (k_1,k_2,r') \in \ds_i[j|h] \land r' + r < \pd \},$$
where $h_r^* = \hevol[h,r|\vec{\sigma}^*]$.

\begin{proof}
Fix $i$, $h$, and $j$. The proof goes by induction on $r$, where the induction hypothesis is that,
for every $r \in \{1 \ldots \pd -1\}$, 
$$\ds_i[j|h_r^*] = \{(k_1,k_2,r'+r) | (k_1,k_2,r') \in \ds_i[j|h] \land r' + r < \pd\}.$$

By Definition~\ref{def:thr}, we have that for every $r \in \{1 \ldots \pd-1\}$, $\vec{p}^* = \vec{\sigma}^*[h_r^*]$, and $s^* = \sig[\vec{p}^*|h]$,
\begin{equation}
\label{eq:corr0}
\ds_i[j|h_{r+1}^*] = L_1[r+1|\vec{\sigma}^*] \cup L_2[r+1|\vec{\sigma}^*],
\end{equation}
where 
\begin{equation}
\label{eq:corr0-0}
\begin{array}{l}
L_1[r+1|\vec{\sigma}^*] = \{(k_1,k_2,r'+1)|(k_1,k_2,r') \in \ds_{i}[j|h_r^*] \land r' +1 < \pd\}.\\
L_2[r+1|\vec{\sigma}^*]= \{(k_1,k_2,0) | k_1,k_2 \in {\cal N} \land i,j \in \rs[k_1,k_2] \land s^*[k_1,k_2] = \mbox{\emph{defect}}\}.
\end{array}
\end{equation}

First, note that by Definition~\ref{def:pubsig} it holds that $s^*[k_1,k_2]=\mbox{\emph{cooperate}}$ for every $k_1 \in{\cal N}$ and $k_2 \in {\cal N}_{k_1}$.
Thus, by~\ref{eq:corr0-0}, for every $r \in \{1\ldots \pd -1\}$,
\begin{equation}
\label{eq:corr0-1}
L_2[r|\vec{\sigma}^*]= \{(k_1,k_2,0) | k_1,k_2 \in {\cal N} \land i,j \in \rs[k_1,k_2] \land s^*[k_1,k_2] = \mbox{\emph{defect}}\}=\emptyset.
\end{equation}
By~\ref{eq:corr0-0},
$$L_1[1|\vec{\sigma}^*] = \{(k_1,k_2,r'+1)|(k_1,k_2,r') \in \ds_{i}[j|h] \land r'  +1 < \pd\},$$
which, along with~\ref{eq:corr0-1} and~\ref{eq:corr0}, proves the base case.

Now, consider that the induction hypothesis is valid for any $r \in \{1\ldots \pd -2\}$.
We have by this assumption and by~\ref{eq:corr0-0} that
$$
\begin{array}{ll}
L_1[r+1|\vec{\sigma}^*] & = \{(k_1,k_2,r'+1)|(k_1,k_2,r') \in \ds_{i}[j|h_r^*] \land r' +1 < \pd\},\\
                                              & = \{(k_1,k_2,r'+1)|(k_1,k_2,r') \in \\
                                               &\{(l_1,l_2,r''+r) | (l_1,l_2,r'') \in \ds_i[j|h] \land r'' + r< \pd\} \land r' + 1 < \pd\},\\
                                               & = \{(k_1,k_2,r'+(r+1))|(k_1,k_2,r') \in \ds_{i}[j|h] \land r' + (r+1) < \pd\}.
\end{array}
$$
This fact, along with~\ref{eq:corr0-1} and~\ref{eq:corr0}, proves the induction step for $r+1$.
\end{proof}

\newpage

\subsubsection{Proof of Lemma~\ref{lemma:corr-1}.}
\label{proof:lemma:corr-1}
For every $h \in {\cal H}$, $\vec{p}' \in {\cal P}$, $r \in \{1 \ldots \pd\}$,
$i \in {\cal N}$, and $j \in {\cal N}_i$:
$$\ds_i[j|h_r'] = \ds_i[j|h_r^*] \cup \{(k_1,k_2,r-1) | k_1,k_2 \in {\cal N} \land k_2 \in \cd_{k_1}[\vec{p}'|h] \land  i,j \in \rs[k_1,k_2]\},$$
where $h_r^* = \hevol[h,r|\vec{\sigma}^*]$, $h_r' = \hevol[h,r| \vec{\sigma}']$, and
$\vec{\sigma}' =\vec{\sigma}^*[h|\vec{p}']$ is the profile of strategies where all players follow $\vec{p}'$ in the first stage.

\begin{proof}
Fix $h$, $\vec{p}'$, $i$, and $j$. The proof goes by induction on $r$, where the induction hypothesis is that
for every $r \in \{1 \ldots \pd\}$, Equality~\ref{eq:res-corr1} holds.

By Definition~\ref{def:thr}, we have that for every $r \leq \pd$, $\vec{p}^r = \vec{\sigma}^*[h_r^*]$, and $s^* = \sig[\vec{p}^r|h_r^*]$:
\begin{equation}
\label{eq:corr1}
\ds_i[j|h_{r+1}^*] = L_1[r+1|\vec{\sigma}^*] \cup L_2[r+1|\vec{\sigma}^*],
\end{equation}
where 
\begin{equation}
\label{eq:corr1-1}
\begin{array}{l}
L_1[r+1|\vec{\sigma}^*] = \{(k_1,k_2,r'+1)|(k_1,k_2,r') \in \ds_{i}[j|h_r^*] \land r' +1 < \pd\}.\\
L_2[r+1|\vec{\sigma}^*]= \{(k_1,k_2,0) | k_1,k_2 \in {\cal N} \land i,j \in \rs[k_1,k_2] \land s^*[k_1,k_2] = \mbox{\emph{defect}}\}.
\end{array}
\end{equation}
Similarly, for every $r \leq \pd$, $\vec{p}^r = \vec{\sigma}'[h_r']$, and $s' = \sig[\vec{p}^r|h_r']$,
\begin{equation}
\label{eq:corr1-2}
\ds_i[j|h_{r+1}'] = L_1[r+1|\vec{\sigma}'] \cup L_2[r+1|\vec{\sigma}'],
\end{equation}
where 
\begin{equation}
\label{eq:corr1-3}
\begin{array}{l}
L_1[r+1|\vec{\sigma}'] = \{(k_1,k_2,r'+1)|(k_1,k_2,r') \in \ds_{i}[j|h_r'] \land r' +1 < \pd\}.\\
L_2[r+1|\vec{\sigma}']= \{(k_1,k_2,0) | k_1,k_2 \in {\cal N} \land i,j \in \rs[k_1,k_2] \land s'[k_1,k_2] = \mbox{\emph{defect}}\}.
\end{array}
\end{equation}

First, note that for every $r \in \{1 \ldots \pd\}$ and $\vec{p}'' \in {\cal P}$, such that $\vec{\sigma}'' = \vec{\sigma}^*[h|\vec{p}'']$, 
we have $\vec{\sigma}''[h'] = \vec{\sigma}^*[h']$ for every $h' \in {\cal H} \setminus \{h\}$.

Thus, for $h_r'' = \hevol[h,r|\vec{\sigma}'']$, $\vec{p}^* = \vec{\sigma}''[h_r'']$, and $s^* = \sig[\vec{p}^*|h_r'']$,
we have by Definition~\ref{def:pubsig} that $s^*[k_1,k_2]=\mbox{\emph{cooperate}}$ for every $k_1 \in{\cal N}$ and $k_2 \in {\cal N}_{k_1}$.
Thus, by Definition~\ref{def:thr}, for every $r \in \{1 \ldots \pd\}$,
\begin{equation}
\label{eq:corr1-4}
L_2[r|\vec{\sigma}'']= \{(k_1,k_2,0) | k_1,k_2 \in {\cal N} \land i,j \in \rs[k_1,k_2] \land s^*[k_1,k_2] = \mbox{\emph{defect}}\} = \emptyset.
\end{equation}

It follows by~\ref{eq:corr1-1} and~\ref{eq:corr1-3} that, for every $r \in \{1 \ldots \pd\}$,
\begin{equation}
\label{eq:corr1-5}
L_2[r|\vec{\sigma}^*] = L_2[r|\vec{\sigma}'] = \emptyset.
\end{equation}

The base case is when $r=1$. Since $h=h_0' = h_0^*$, by~\ref{eq:corr1-1} and~\ref{eq:corr1-3},
it is true that: 
\begin{equation}
\label{eq:corr1-6}
L_1[1|\vec{\sigma}'] = L_1[1|\vec{\sigma}^*].
\end{equation}

Furthermore, if players follow $\vec{p}'$, then for $s' = \sig[\vec{p}'|h]$,
we have $s'[k_1,k_2]=\mbox{\emph{defect}}$ iff $k_2 \in \cd_{k_1}[\vec{p}'|h]$. Thus:
\begin{equation}
\label{eq:corr1-7}
\begin{array}{ll}
L_2[1|\vec{\sigma}'] & = \{(k_1,k_2,0) | k_1,k_2 \in {\cal N} \land i,j \in \rs[k_1,k_2] \land s'[k_1,k_2]=\mbox{\emph{defect}}\}\\
                 & = \{(k_1,k_2,0) | k_1,k_2 \in {\cal N} \land i,j \in \rs[k_1,k_2] \land k_2 \in \cd_{k_1}[\vec{p}'|h]\}.
\end{array}
\end{equation}
Consequently, the base case follows from~\ref{eq:corr1},~\ref{eq:corr1-2},~\ref{eq:corr1-5},~\ref{eq:corr1-6}, and~\ref{eq:corr1-7}.

Hence, assume the induction hypothesis for $r \in \{1\ldots \pd-1\}$.
By the induction hypothesis and by~\ref{eq:corr1-1}, since $r < \pd$, it is also true that
\begin{equation}
\label{eq:corr1-8}
\begin{array}{ll}
L_1[r+1|\vec{\sigma}'] & = \{(k_1,k_2,r'+1) | (k_1,k_2,r') \in \ds_i[j|h_r'] \land r' +1< \pd\} \\
&\\
                     & = \{(k_1,k_2,r'+1) | (k_1,k_2,r') \in \ds_i[j|h_r^*] \land r' +1 <\pd\} \cup \\
                     & \cup \{(k_1,k_2,r'+1) | (k_1,k_2,r') \in \{(l_1,l_2,r-1) | l_1,l_2 \in {\cal N} \\
                     &\land l_2 \in \cd_{l_1}[\vec{p}'|h] \land i,j \in \rs[l_1,l_2] \} \land r'+1<\pd\}\\
&\\
                     & = L_1[r+1|\vec{\sigma}^*] \cup \{(k_1,k_2,(r+1)-1) |k_1,k_2 \in {\cal N}\\
                      &\land  k_2 \in \cd_{k_1}[\vec{p}'|h] \land i,j \in \rs[k_1,k_2]\}.
\end{array}
\end{equation}

The induction hypothesis follows from~\ref{eq:corr1},~\ref{eq:corr1-2},~\ref{eq:corr1-5}, and~\ref{eq:corr1-8} for $r+1$, which proves the result.
\end{proof}
\newpage

\subsubsection{Proof of Lemma~\ref{lemma:corr-2}.}
\label{proof:lemma:corr-2}
For every $h \in {\cal H}$, $\vec{p}' \in {\cal P}$, $r > \pd$,
$i \in {\cal N}$, and $j \in {\cal N}_i$,
$$\ds_i[j|h_r'] = \ds_i[j|h_r^*] = \emptyset,$$
where $h_r^* = \hevol[r|\vec{\sigma}^*]$, $h_r' = \hevol[r| \vec{\sigma}']$,
and $\vec{\sigma}' = \vec{\sigma}^*[h|\vec{p}']$.

\begin{proof}
Fix $h$, $\vec{p}'$, $i$, and $j$. The proof goes by induction on $r$, where the induction hypothesis is that
for every $r \in \{1 \ldots \pd\}$, Equality~\ref{eq:res-corr2} holds.

By Definition~\ref{def:thr}, we have that for every $r >0$, $\vec{p}^r = \vec{\sigma}^*[h_r^*]$, and $s^* = \sig[\vec{p}^r|h_r^*]$:
\begin{equation}
\label{eq:corr2}
\ds_i[j|h_{r+1}^*] = L_1[r+1|\vec{\sigma}^*] \cup L_2[r+1|\vec{\sigma}^*],
\end{equation}
where 
\begin{equation}
\label{eq:corr2-1}
\begin{array}{l}
L_1[r+1|\vec{\sigma}^*] = \{(k_1,k_2,r'+1)|(k_1,k_2,r') \in \ds_{i}[j|h_r^*] \land r' +1 < \pd\}.\\
L_2[r+1|\vec{\sigma}^*]= \{(k_1,k_2,0) | k_1,k_2 \in {\cal N} \land i,j \in \rs[k_1,k_2] \land s^*[k_1,k_2] = \mbox{\emph{defect}}\}.
\end{array}
\end{equation}
Similarly, for every $r >0$, $\vec{p}^r = \vec{\sigma}'[h_r']$, and $s' = \sig[\vec{p}^r|h_r']$,
\begin{equation}
\label{eq:corr2-2}
\ds_i[j|h_{r+1}'] = L_1[r+1|\vec{\sigma}'] \cup L_2[r+1|\vec{\sigma}'],
\end{equation}
where 
\begin{equation}
\label{eq:corr2-3}
\begin{array}{l}
L_1[r+1|\vec{\sigma}'] = \{(k_1,k_2,r'+1)|(k_1,k_2,r') \in \ds_{i}[j|h_r'] \land r' +1 < \pd\}.\\
L_2[r+1|\vec{\sigma}']= \{(k_1,k_2,0) | k_1,k_2 \in {\cal N} \land i,j \in \rs[k_1,k_2] \land s'[k_1,k_2] = \mbox{\emph{defect}}\}.
\end{array}
\end{equation}

First, note that for every $r >0$, $\vec{p}'' \in {\cal P}$, and $\vec{\sigma}'' = \vec{\sigma}^*[h|\vec{p}'']$, 
we have $\vec{\sigma}''[h'] = \vec{\sigma}^*[h']$ for every $h' \in {\cal H} \setminus \{h\}$.

Thus, for $h_r'' = \hevol[r|\vec{\sigma}'']$, $\vec{p}^* = \vec{\sigma}''[h_r'']$, and $s^* = \sig[\vec{p}^*|h_r'']$,
we have by Definition~\ref{def:pubsig} that $s^*[k_1,k_2]=\mbox{\emph{cooperate}}$ for every $k_1 \in{\cal N}$ and $k_2 \in {\cal N}_{k_1}$.
Thus, by Definition~\ref{def:thr}, for every $r>0$,
\begin{equation}
\label{eq:corr2-4}
L_2[r|\vec{\sigma}'' ]= \{(k_1,k_2,0) | k_1,k_2 \in {\cal N} \land i,j \in \rs[k_1,k_2] \land s^*[k_1,k_2] = \mbox{\emph{defect}}\} = \emptyset.
\end{equation}

It follows by~\ref{eq:corr2-1} and~\ref{eq:corr2-3} that, for every $r>\pd$,
\begin{equation}
\label{eq:corr2-5}
L_2[r|\vec{\sigma}^*] = L_2[r|\vec{\sigma}'] = \emptyset.
\end{equation}

By Lemma~\ref{lemma:corr-1},
$$\ds_i[j|h_{\pd}'] = \ds_i[j|h_{\pd}^*] \cup \{(k_1,k_2,\pd-1)|k_1,k_2 \in {\cal N} \land k_2 \in \cd_{k_1}[\vec{p}'|h] \land i,j \in \rs[k_1,k_2]\}.$$

Therefore, by~\ref{eq:corr2-1} and~\ref{eq:corr2-3},
$$
\begin{array}{ll}
L_1[\pd|\vec{\sigma}'] & =  \{(k_1,k_2,r'+1) |(k_1,k_2,r') \in \ds_i[j|h_{\pd}'] \land r'+1<\pd\} \\
&\\
                         & = \{(k_1,k_2,r'+1) | (k_1,k_2,r') \in \ds_i[j|h_{\pd}^*]  \land r'+1 <\pd\} \cup \\
                         & \cup \{(k_1,k_2,r'+1) | (k_1,k_2,r') \in \{(l_1,l_2,\pd-1)| l_1,l_2 \in {\cal N} \land l_2 \in \cd_{l_1}[\vec{p}'|h]  \\
                          &\land i,j \in \rs[l_1,l_2]\} \land r' +1 < \pd\}\\
                          &\\
                          & = \{(k_1,k_2,r'+1) | (k_1,k_2,r') \in \ds_i[j|h_{\pd}^*]  \land r'+1 <\pd\} \cup \\
                         & \cup \{(k_1,k_2,r'+1) | k_1,k_2 \in {\cal N} \land k_2 \in \cd_{k_1}[\vec{p}'|h] \land i,j \in \rs[k_1,k_2] \land \pd < \pd\}\\
                          &\\
                         & = \{(k_1,k_2,r'-1) | (k_1,k_2,r') \in \ds_i[j|h_{\pd}^*]  \land r' +1 < \pd\} \cup \emptyset\\
                         & = L_1[\pd+1|\vec{\sigma}^*].
\end{array}
$$

By Corollary~\ref{corollary:corr-0}, 
\begin{equation}
\label{eq:corr2-6}
L_1[\pd+1|\vec{\sigma}'] = L_1[\pd+1|\vec{\sigma}^*] = \emptyset.
\end{equation}

By~\ref{eq:corr2},~\ref{eq:corr2-2}~\ref{eq:corr2-5},~\ref{eq:corr2-6},
the base case is true.

Now, assume the induction hypothesis for some $r\geq \pd+1$.
By this assumption,~\ref{eq:corr2-1}, and~\ref{eq:corr2-3}:
\begin{equation}
\label{eq:corr2-7}
\begin{array}{ll}
L_1[r+1|\vec{\sigma}'] &=  \{(k_1,k_2,r'+1) | (k_1,k_2,r') \in \ds_i[j|h_{r}'] \land r'+1 < \pd\}\\
                   & = \{(k_1,k_2,r'+1) | (k_1,k_2,r') \in \emptyset\} \\
                   & = \emptyset.
\end{array}
\end{equation}
By Definition~\ref{def:thr}, by~\ref{eq:corr2-2}, and~\ref{eq:corr2-4},
the induction step is true for $r+1$, which proves the result.
\end{proof}

\newpage
\subsection{Generic Results}
\label{sec:proof:gen-cond}

\subsubsection{Proof of Proposition~\ref{prop:folk}.}
\label{proof:prop:folk}
For every profile of punishing strategies $\vec{\sigma}^*$, if $\vec{\sigma}^*$ is a SPE,
then, for every $i \in {\cal N}$, $\frac{\beta_i}{\gamma_i} >\bar{p}_i$.
Consequently, $\psi[\vec{\sigma}^*] \subseteq (v,\infty)$, where $v = \max_{i \in {\cal N}} \bar{p}_i$.
\begin{proof}
Let $\vec{p}^* = \vec{\sigma}^*[h]$.
The equilibrium utility is 
$$\pi_i[\vec{\sigma}^*|\emptyset] = \sum_{r=0}^\infty\omega_i^r (1 - q_i[\vec{p}^*])(\beta_i - \gamma_i \bar{p}_i) = \frac{1 - q_i[\vec{p}^*]}{1-\omega_i}(\beta_i - \gamma_i \bar{p}_i).$$
If $\frac{\beta_i}{\gamma_i} \leq \bar{p}_i$, then 
\begin{equation}
\label{eq:od-1}
\pi_i[\vec{\sigma}^*|\emptyset]  \leq 0.
\end{equation}
Let $\sigma_i' \in \Sigma_i$ be a strategy such that, for every $h \in {\cal H}$, $\sigma_i'[h] = \vec{0}$,
and let $\vec{\sigma}' = (\sigma_i',\vec{\sigma}_{-i}^*)$, where $\vec{0}=(0)_{j \in {\cal N}_i}$. We have 
\begin{equation}
\label{eq:od-2}
\pi_i[\vec{\sigma}'|\emptyset] = (1-q_i[\vec{p}^*])\beta_i + \pi_i[\vec{\sigma}'|(h,\sig[\vec{p}'|h])] \geq (1-q_i[\vec{p}^*])\beta_i,
\end{equation}
where $\vec{p}' = (\vec{0},\vec{p}^*_{-i})$. By Lemma~\ref{lemma:pprob},
$q_i [\vec{p}^*] < 1$. Since $ \pi_i[\vec{\sigma}'|(h,\sig[\vec{p}'|h])] \geq 0$, it is true that
$$\pi_i[\vec{\sigma}^*|\emptyset] \leq 0 < \pi_i[\vec{\sigma}'|\emptyset].$$
This contradicts the assumption that $\vec{\sigma}^*$ is a SPE.
\end{proof}

\newpage
\subsubsection{Proof of Lemma~\ref{lemma:gen-cond-nec}.}
\label{proof:lemma:gen-cond-nec}
If $\vec{\sigma}^*$ is a SPE, then the DC Condition is fulfilled.

\begin{proof}
The proof consists in assuming that $\vec{\sigma}^*$ is a SPE and deriving~\ref{eq:gen-cond}.
By Property~\ref{prop:one-dev}, we must have, for every $h \in {\cal H}$ and $a_i' \in {\cal A}_i$,
\begin{equation}
\label{eq:drop-nec}
\pi_i[\vec{\sigma}^*|h] - \pi_i[\vec{\sigma}'|h]\geq 0,
\end{equation}
where $\sigma_i' = \sigma_i^*[h|a_i']$ and $\vec{\sigma}' = (\sigma_i',\vec{\sigma}_{-i}^*)$. 
This is true for any $\sigma_i'$, where $a_i'[\vec{p}_i'] = 1$ and $\vec{p}_i'$ differs
from $\sigma_i^*[h]$ exactly in that $i$ drops the nodes from any set $D \subseteq {\cal N}_{i}[h]$:
\begin{itemize}
  \item For every $j \in D$, $p_i'[j] = 0$.
  \item For every $j \in {\cal N}_i \setminus D$, $p_i'[j] = p_i[j|h]$  
\end{itemize}

For any pure profile of strategies $\vec{\sigma} \in \Sigma$, we can write
\begin{equation}
\label{eq:gcn-1}
\pi_i[\vec{\sigma}|h] = \sum_{r=0}^{\infty} \omega_i^r u_i[h,r|\vec{\sigma}].
\end{equation}
By Lemma~\ref{lemma:corr-2} and Definition~\ref{def:thr}, for every $r > \pd$, $j \in {\cal N}$, and $k \in {\cal N}_j$,
\begin{equation}
\label{eq:gcn-2}
\begin{array}{l}
\ds_j[k|h_r'] = \ds_j[k|h_r^*] = \emptyset,\\
\\
p_j[k|h_r'] = p_j[k|h_r^*],
\end{array}
\end{equation}
where $h_r' = \hevol[h,r|\vec{\sigma}']$ and $h_r^*=\hevol[h,r|\vec{\sigma}^*]$.
This implies that for every $r> \pd$:
$$q_i[h,r|\vec{\sigma}']=q_i[h,r|\vec{\sigma}^*],$$
$$\bar{p}_i[h,r|\vec{\sigma}']=\bar{p}_i[h,r|\vec{\sigma}^*],$$
$$u_i[h,r|\vec{\sigma}'] = u_i[h,r|\vec{\sigma}^*].$$ 
Thus, \ref{eq:drop-nec} and~\ref{eq:gcn-1} imply~\ref{eq:gen-cond}.
\end{proof}

\newpage
\subsubsection{Proof of Lemma~\ref{lemma:best-response1}.}
\label{proof:lemma:best-response1}
For every $i \in {\cal N}$, $h \in {\cal H}$, $a_i \in BR[\vec{\sigma}_{-i}^*|h]$, and $\vec{p}_i \in {\cal P}_i$ such that $a_i[\vec{p}_i] > 0$,
it is true that for every $j \in {\cal N}_i$ we have $p_i[j] \in \{0,p_i[j|h]\}$.

\begin{proof}
Suppose then that there exist $h \in {\cal H}$, $i \in {\cal N}$, $a_i^1 \in BR[\vec{\sigma}_{-i}^*|h]$, and $\vec{p}^1_i \in {\cal P}_i$ such that $a_i^1[\vec{p}^1_i]>0$
and there exists $j \in {\cal N}_i$ such that $p_i^1[j] \notin \{0,p_i[j|h]\}$.
Consider an alternative $a_i^2 \in {\cal A}_i$:

\begin{itemize}
  \item Define $\vec{p}_i^2 \in {\cal P}_i$ such that for every $j \in {\cal N}_i$, if $p_i^1[j] \geq p_i[j|h]$, then $p_i^2[j] = p_i[j|h]$, else, $p_i^2[j]=0$.
  \item Set $a_i^2[\vec{p}_i^2] = a_i^1[\vec{p}_i^1] + a_i^1[\vec{p}_i^2]$ and $a_i^2[\vec{p}_i^1] = 0$.
  \item For every $\vec{p}_i'' \in {\cal P}_i \setminus \{\vec{p}_i^1,\vec{p}_i^2\}$, set $a_i^2[\vec{p}_i''] = a_i^1[\vec{p}_i'']$.
\end{itemize}

Consider the following auxiliary definitions:
\begin{itemize}
  \item $\vec{a}^1 = (a_i^1,\vec{p}_{-i}^*)$ and $\vec{a}^2 = (a_i^2,\vec{p}_{-i}^*)$, where $\vec{p}^* = \vec{\sigma}^*[h]$.
  \item $\sigma_i^1 = \sigma_i^*[h|a_i^1]$ and $\sigma_i^2 = \sigma_i^*[h|a_i^2]$.
  \item $\vec{p}^1 = (\vec{p}_i^1,\vec{p}^*_{-i})$ and $\vec{p}^2 = (\vec{p}_i^2,\vec{p}^*_{-i})$.
  \item $\vec{\sigma}^1 = (\sigma_i^1,\vec{\sigma}^*_{-i})$ and $\vec{\sigma}^2 = (\sigma_i^2,\vec{\sigma}^*_{-i})$. 
  \item $s^1= \sig[\vec{p}^1|h]$ and $s^2=\sig[\vec{p}^2|h]$.
\end{itemize}

Notice that for any $j \in {\cal N}_i$, $p_i^1[j] \geq p_i^2[j]$ and $p_i^1[j] \geq p_i[j|h]$ iff $p_i^2[j] \geq p_i[j|h]$.
Thus, by Definition~\ref{def:pubsig}, for any $s \in {\cal S}$,
\begin{equation}
\label{eq:br-1}
\begin{array}{l}
pr_i[s|a_i^1,h] = pr_i[s|a_i^2,h].\\
pr[s|\vec{a}^1,h] = pr[s|\vec{a}^2,h].
\end{array}
\end{equation}

Moreover, for some $j \in {\cal N}_i$, $p_i^1[j|h] > p_i^2[j|h]$, thus, it is true that
\begin{equation}
\label{eq:br-2}
u_i[\vec{a}^1] < u_i[\vec{a}^2].
\end{equation}

Recall that
$$\pi_i[\vec{\sigma}^1|h] = u_i[\vec{a}^1] + \omega_i\sum_{s \in {\cal S}} \pi_i[\vec{\sigma}^1|(h,s)]pr[s|\vec{a}^1,h],$$
$$\pi_i[\vec{\sigma}^2|h] = u_i[\vec{a}^2] + \omega_i\sum_{s \in {\cal S}} \pi_i[\vec{\sigma}^2|(h,s)]pr[s|\vec{a}^2,h].$$

By~\ref{eq:br-1} and the definition of $\vec{\sigma}^1$ and $\vec{\sigma}^2$,
$$\sum_{s \in {\cal S}} \pi_i[\vec{\sigma}|(h,s)]pr[s|\vec{a}^1,h] = \sum_{s \in {\cal S}} \pi_i[\vec{\sigma}|(h,s)]pr[s|\vec{a}^2,h].$$

By~\ref{eq:br-2},
$$\pi_i[\vec{\sigma}^1|h] < \pi_i[\vec{\sigma}^2|h].$$

This is a contradiction, since $a_i^1 \in BR[\vec{\sigma}_{-i}^*|h]$ by assumption, concluding the proof.
\end{proof}
\newpage

\subsubsection{Proof of Lemma~\ref{lemma:best-response2}.}
\label{proof:lemma:best-response2}
For every $h \in {\cal H}$ and $i \in {\cal N}$, there exists $a_i \in BR[\vec{\sigma}_{-i}^*| h]$ and $\vec{p}_i \in {\cal P}_i$
such that $a_i[\vec{p}_i] = 1$.

\begin{proof}
For any $h \in {\cal H}$ and $i \in {\cal N}$, if $BR[\vec{\sigma}_{-i}^*|h]$ only contains pure strategies for the stage game,
since $BR[\vec{\sigma}_{-i}^*|h]$ is not empty, the result follows.
Suppose then that there exists a mixed strategy $a_i^1 \in BR[\vec{\sigma}_{-i}^*|h]$. 
We know from Lemma~\ref{lemma:best-response1} that every such $a_i^1$ attributes positive probability to one of two probabilities in $\{0,p_i[j|h]\}$, for every $j \in {\cal N}_i$.
Let $\sigma_i^1 = \sigma_i^*[h|a_i^1]$, $\vec{\sigma}^1 = (\sigma_i^1,\vec{\sigma}^*_{-i})$, and denote by ${\cal P}^*[h]$ the finite set of profiles of
probabilities that fulfill the condition of Lemma~\ref{lemma:best-response1}, i.e., for every $\vec{p} \in {\cal P}^*[h]$, $j\in{\cal N}$, and $k \in {\cal N}_j$,
$p_j[k] \in \{0,p_j[k|h]\}$. We can write
\begin{equation}
\label{eq:br2-0}
\pi_i[\vec{\sigma}^1|h]  = \sum_{\vec{p}_i \in {\cal P}^*_i[h]} (u_i[\vec{p}] + \omega_i \pi_i[\vec{\sigma}^1|(h,\sig[\vec{p}|h])])a_i^1[\vec{p}_i],
\end{equation}
where $\vec{p} = (\vec{p}_i,\vec{p}_{-i}^*)$ and $\vec{p}^* = \vec{\sigma}^*[h]$.

For any $\vec{p}_i^1 \in {\cal P}_i^*[h]$ such that $a_i^1[\vec{p}_i^1] >0$, let $\vec{p}^* = \vec{\sigma}^*[h]$, $\vec{p}^1 = (\vec{p}_i^1,\vec{p}_{-i}^*)$,
 $\sigma_i' = \sigma_i^*[h|\vec{p}_i^1]$,
and $\vec{\sigma}' = (\sigma_i',\vec{\sigma}^*_{-i})$.

There are three possibilities:
\begin{enumerate}
 \item $\pi_i[\vec{\sigma}^1|h] = \pi_i[\vec{\sigma}'|h]$.
 \item $\pi_i[\vec{\sigma}^1|h] < \pi_i[\vec{\sigma}'|h]$.
 \item $\pi_i[\vec{\sigma}^1|h] > \pi_i[\vec{\sigma}'|h]$.
\end{enumerate}

In possibility~1, it is true that there is $a_i' \in BR[\vec{\sigma}_{-i}^*|h]$ such that $a_i'[\vec{p}_i^1] = 1$ and the result follows. Possibility~2 contradicts
the assumption that $a_i^1 \in BR[\vec{\sigma}_{-i}^*|h]$. 

Finally, consider that possibility~3 is true. Recall that $a_i^1$ being mixed implies $a_i^1[\vec{p}_i^1] < 1$.
Thus, there must exist $\vec{p}_i^2 \in {\cal P}_i^*[h] \setminus \{\vec{p}_i^1\}$, $\sigma_i''=\sigma_i^*[h|\vec{p}_i^2]$, and $\vec{\sigma}'' = (\sigma_i',\vec{\sigma}^*_{-i})$,
such that $a_i^1[\vec{p}_i^2] > 0$ and
\begin{equation}
\label{eq:br2-1}
\pi_i[\vec{\sigma}'|h] < \pi_i[\vec{\sigma}''|h].
\end{equation}

Here, we can define $a_i^2 \in {\cal A}_i$ such that:
\begin{itemize}
  \item $a_i^2[\vec{p}_i^2]=a_i^1[\vec{p}_i^1] + a_i^1[\vec{p}_i^2]$; 
  \item $a_i^2[\vec{p}_i^1] = 0$.
  \item For every $\vec{p}_i'' \in {\cal P}^*[h] \setminus \{\vec{p}_i^1,\vec{p}_i^2\}$, $a_i^2[\vec{p}_i''] = a_i^1[\vec{p}_i'']$.  
\end{itemize}

Now, let:
\begin{itemize}
  \item $\sigma_i^2 = \sigma_i^*[h|a_i^2]$ and $\vec{\sigma}^2 = (\sigma_i^2,\vec{\sigma}^*_{-i})$.
  \item $\vec{p}^2 = (\vec{p}_i^2,\vec{p}^*_{-i})$.
  \item $\vec{\sigma}' = (\sigma_i[h|\vec{p}_i''],\vec{\sigma}_{-i}^*)$.
\end{itemize}

By~\ref{eq:br2-0},
$$\pi_i[\vec{\sigma}^1|h] = l_1 + \pi_i[\vec{\sigma}'|h]a_i^1[\vec{p}_i^1] + \pi_i[\vec{\sigma}''|h]a_i^1[\vec{p}_i^2],$$
$$\pi_i[\vec{\sigma}^2|h] = l_2 + \pi_i[\vec{\sigma}''|h]a_i^2[\vec{p}_i^2],$$
where 
$$l_1 = \sum_{\vec{p}'' \in {\cal P}^*[h] \setminus \{\vec{p}_i^1,\vec{p}_i^2\}}  (u_i[\vec{p}''] + \omega_i \pi_i[\vec{\sigma}^1|h,s''])a_i^1[\vec{p}_i''],$$
$$l_2 = \sum_{\vec{p}'' \in {\cal P}^*[h] \setminus \{\vec{p}_i^1,\vec{p}_i^2\}}  (u_i[\vec{p}''] + \omega_i \pi_i[\vec{\sigma}^2|h,s''])a_i^2[\vec{p}_i''],$$
and $s'' = \sig[\vec{p}''|h]$.

By the definition of $a_i^2$, we have that $l_1 = l_2$. 
It follows that:
$$
\begin{array}{ll}
\pi_i[\vec{\sigma}^1|h] - \pi_i[\vec{\sigma}^2|h] & = \pi_i[\vec{\sigma}'|h]a_i^1[\vec{p}_i^1]  + \pi_i[\vec{\sigma}''|h]a_i^1[\vec{p}_i^2] - \pi_i[\vec{\sigma}''|h](a_i^1[\vec{p}_i^2] + a_i^1[\vec{p}_i^1]) \\
                                                                                   & = (\pi_i[\vec{\sigma}'|h] - \pi_i[\vec{\sigma}''|h])a_i^1[\vec{p}_i^1].
\end{array}
$$

If follows from~\ref{eq:br2-1} that:
$$\pi_i[\vec{\sigma}^1|h] < \pi_i[\vec{\sigma}^2|h],$$
contradicting the assumption that $a_i^1 \in BR[\vec{\sigma}_{-i}^*| h]$. This concludes the proof.
\end{proof}

\newpage
\subsubsection{Proof of Lemma~\ref{lemma:best-response}.}
\label{proof:lemma:best-response}
For every $h \in {\cal H}$ and $i \in {\cal N}$, there exists $\vec{p}_i \in {\cal P}_i$ and a pure strategy $\sigma_i=\sigma_i^*[h|\vec{p}_i]$ such that:
\begin{enumerate}
 \item For every $j \in {\cal N}_i$, $p_i[j] \in \{0,p_{i}[j|h]\}$.
 \item For every $a_i \in {\cal A}_i$, $\pi_i[\sigma_i,\vec{\sigma}_{-i}^*|h] \geq \pi_i[\sigma_i',\vec{\sigma}_{-i}^*|h]$,
where $\sigma_i' = \sigma_i^*[h|a_i]$.
\end{enumerate}

\begin{proof}
Consider any $h \in {\cal H}$ and $i \in {\cal N}$.
From Lemma~\ref{lemma:best-response2}, it follows that there exists $a_i \in BR[\vec{\sigma}_{-i}^*|h]$
and $\vec{p}_i \in {\cal P}_i$ such that $a_i[\vec{p}_i] = 1$. By Lemma~\ref{lemma:best-response1},
every such $a_i$ and $\vec{p}_i$ such that $a_i[\vec{p}_i]=1$ fulfill Condition~$1$.
Condition~$2$ follows from the definition of $BR[\vec{\sigma}_{-i}^*|h]$.
\end{proof}

\subsubsection{Proof of Lemma~\ref{lemma:gen-cond-suff}.}
\label{proof:lemma:gen-cond-suff}
If the DC Condition is fulfilled, then $\vec{\sigma}^*$ is a SPE.

\begin{proof}
Assume that Inequality~\ref{eq:gen-cond} holds for every history $h$ and
$D \subseteq {\cal N}_{i}[h]$. In particular, these assumptions imply that,
for each $\vec{p}_i \in {\cal P}_i$ such that $p_i[j] \in \{0,p_i[j|h]\}$ for every $j \in {\cal N}_i$,
we have 
\begin{equation}
\label{eq:gcs}
\pi_i[\vec{\sigma}^*|h] \geq \pi_i[\sigma_i,\vec{\sigma}^*_{-i}|h],
\end{equation}
where $\sigma_i = \sigma_i^*[h|\vec{p}_i]$. By Lemma~\ref{lemma:best-response}, 
there exists one such $\vec{p}_i$ such that $\sigma_i$ is a local best response.
Consequently, by~\ref{eq:gcs}, for every $a_i \in {\cal A}_i$ and $\sigma_i'=\sigma_i^*[h|a_i]$,
$$\pi_i[\vec{\sigma}^*|h] \geq \pi_i[\sigma_i',\vec{\sigma}^*_{-i}|h].$$
By Property~\ref{prop:one-dev}, $\vec{\sigma}^*$ is a SPE.
\end{proof}

\newpage

\subsection{Direct Reciprocity is not Effective}
\label{sec:proof:direct}

\subsubsection{Proof of Lemma~\ref{lemma:nec-btc}.}
\label{proof:lemma:nec-btc}
If $\vec{\sigma}^*$ is a SPE, then, for every $i \in {\cal N}$ and $j \in {\cal N}_i$, it is true that $q_i' > q_i^*$ and:
$$\frac{\beta_i}{\gamma_i} > \bar{p}_i + \frac{p_i[j|\emptyset]}{q_i' - q_i^*}\left(1-q_i' + \frac{1-q_i^*}{\pd}\right),$$
where $\vec{p}_i'$ is the strategy where $i$ drops $j$, $\vec{\sigma}' = (\sigma_i^*[\emptyset|\vec{p}_i'],\vec{\sigma}_{-i}^*)$,
$q_i'=q_i[\vec{\sigma}'[\emptyset]]$, and $q_i^* = q_i[\vec{\sigma}^*[\emptyset]]$.

\begin{proof}
The assumption that $\vec{\sigma}^*$ is a SPE implies by Theorem~\ref{theorem:gen-cond} that the DC Condition
is true for the history $\emptyset$, any node $i \in {\cal N}$, and $D =\{j\}$, where $j \in {\cal N}_i$. Define $\vec{p}_i' \in {\cal P}_i$ as:
\begin{itemize}
  \item $p_i'[j] =0$.
  \item For every $k \in {\cal N}_i \setminus \{j\}$, $p_i'[k] = p_i[k|\emptyset]$.
\end{itemize}

We have that:
\begin{equation}
\label{eq:nbtc-0}
\begin{array}{ll}
\bar{p}_i - \bar{p}_i[\emptyset,0|\vec{\sigma}'] & = \sum_{k \in {\cal N}_i} p_i[k|\emptyset] - \sum_{k \in {\cal N}_i \setminus \{j\}} p_i[k|\emptyset]\\
                                                                                & = p_i[j|\emptyset].
                                                                               \end{array}
\end{equation}

Let $\sigma_i' = \sigma_i^*[\emptyset| \vec{p}_i']$, $\vec{\sigma}' = (\sigma_i',\vec{\sigma}^*_{-i})$, $h_r^* = \hevol[\emptyset,r|\vec{\sigma}^*]$,
and $h_r' = \hevol[\emptyset,r|\vec{\sigma}']$. It is true by the definition of $u_i$ and by~\ref{eq:nbtc-0} that
\begin{equation}
\label{eq:nbtc-1}
\begin{array}{ll}
u_i[\emptyset,0|\vec{\sigma}^*] - u_i[\emptyset,0|\vec{\sigma}'] &= (1-q_i^*)(\beta_i - \gamma_i \bar{p}_i) - (1-q_i^*)(\beta_i - \gamma_i \bar{p}_i[\emptyset,0|\vec{\sigma}'])\\
                                                                  & = (1-q_i^*)\gamma_i ( \bar{p}_i[\emptyset,0|\vec{\sigma}'] -  \bar{p}_i)\\
                                                                  & = - (1-q_i^*)\gamma_i p_i[j|\emptyset].\\
                                                                  & = -c,
\end{array}
\end{equation}
where $c = (1-q_i^*)\gamma_i p_i[j|\emptyset]$. Notice that for every $k \in {\cal N} \setminus \{i\}$
\begin{equation}
\label{eq:nb-1}
\cd_k[\vec{p}'|h] = \emptyset,
\end{equation}
and 
\begin{equation}
\label{eq:nb-2}
\cd_i[\vec{p}'|h] = \{j\}.
\end{equation}

From~\ref{eq:nb-1} and~\ref{eq:nb-2}, and by Lemma~\ref{lemma:corr-1}, and Definition~\ref{def:thr}, for every $r \in \{1 \ldots \pd\}$,
\begin{equation}
\label{eq:nbtc-2}
\begin{array}{ll}
\ds_i[j|h_r'] & = \ds_i[j|h_r^*] \cup \{(k_1,k_2,r - 1)|  k_1,k_2 \in {\cal N} \land i,j \in \rs[k_1,k_2] \land k_2 \in \cd_{k_1}[\vec{p}'|h]\}\\
                       & = \ds_i[j|h_r^*] \cup \{(i,j,r-1)| i,j \in \rs[i,j]\}\\
                       & = \{(i,j,r-1)\}.
\end{array}
\end{equation}
and for every $k \in {\cal N} \setminus \{j\}$
\begin{equation}
\label{eq:nbtc-3}
\begin{array}{ll}
\ds_i[k|h_r'] & = \ds_i[k|h_r^*] \cup \{(k_1,k_2,r-1)| k_1,k_2 \in {\cal N} \land i,k \in \rs[k_1,k_2] \land k_2 \in \cd_{k_1}[\vec{p}'|h]\}\\
                       & = \ds_i[k|h_r^*] \cup \{(i,k,r-1)| i,k \in \rs[i,j]\}\\
                       & = \ds_i[k|h_r^*] = \emptyset.
\end{array}
\end{equation}
By~\ref{eq:nbtc-2} and~\ref{eq:nbtc-3}, Definition~\ref{def:thr}, and the definition of $\vec{p}'$, for every $r \in \{1 \ldots \pd\}$
\begin{equation}
\label{eq:nbtc-4}
\bar{p}_i[\emptyset,r|\vec{\sigma}'] = \sum_{k \in {\cal N}_i \setminus\{j\}} p_i[k|h_r']  =\sum_{k \in {\cal N}_i \setminus\{j\}} p_i[k|\emptyset] = \bar{p}_i - p_i[j|\emptyset].
\end{equation}

By~\ref{eq:nbtc-2} and~\ref{eq:nbtc-4}, for every $r \in \{1 \ldots \pd\}$,
\begin{equation}
\label{eq:nbtc-5}
\begin{array}{ll}
u_i[\emptyset,r|\vec{\sigma}^*] - u_i[\emptyset,r|\vec{\sigma}'] &= (1-q_i[\emptyset,r|\vec{\sigma}^*])(\beta_i - \gamma_i \bar{p}_i[\emptyset,r|\vec{\sigma}^*]) - \\
						&(1-q_i[\emptyset,r|\vec{\sigma}'])(\beta_i - \gamma_i \bar{p}_i[\emptyset,r|\vec{\sigma}'])\\
						&\\
                                                                   & = (1-q_i^*)(\beta_i - \gamma_i \bar{p}_i) - (1-q_i')(\beta_i -  \gamma_i \bar{p}_i + p_i[j|\emptyset])\\
                                                                   & = (q_i' - q_i^*)(\beta_i - \gamma_i \bar{p}_i)  - (1-q_i')\gamma_i p_i[j|\emptyset]\\
                                                                   & = a - b,
\end{array}
\end{equation}
where:
\begin{itemize}
  \item $a = (q_i' - q_i^*)(\beta_i - \gamma_i \bar{p}_i)$.
  \item $b = (1-q_i')\gamma_i p_i[j|\emptyset]$.
\end{itemize}

By Theorem~\ref{theorem:gen-cond}, the assumption that $\vec{\sigma}^*$ is a SPE,~\ref{eq:nbtc-1}, and~\ref{eq:nbtc-5},
\begin{equation}
\label{eq:nbtc-6}
\begin{array}{ll}
\sum_{r=0}^{\pd} \omega_i^r(u_i[h,r|\vec{\sigma}^*] - u_i[h,r|\vec{\sigma}']) &\geq 0\\
-c + \sum_{r=1}^{\pd} \omega_i^r (a-b) & \geq 0\\
-c + \frac{\omega_i-\omega_i^{\pd+1}}{1-\omega_i}(a-b) & \geq 0\\
-c(1-\omega_i) + (\omega_i-\omega_i^{\pd+1})(a-b) & \geq 0\\
\omega_i(a - b + c) - \omega_i^{\pd+1}(a-b) -c & \geq 0.
\end{array}
\end{equation}

This is a polynomial function of degree $\pd+1 \geq 2$ that has a zero in $\omega_i = 1$ and is negative for $\omega_i = 0$, since, by Lemma~\ref{lemma:pprob}, $c > 0$.
For any $\omega_i \in (0,1)$, a solution to~\ref{eq:nbtc-6} exists only if the polynomial is strictly concave. 
The second derivative is 
$$-(\pd+1)\pd\omega_i^{\pd-1}(a-b),$$ 
so we must have $a > b$ for this to be true, i.e.:

\begin{equation}
\label{eq:nbtc-7}
\begin{array}{ll}
(q_i' - q_i^*)(\beta_i - \gamma_i \bar{p}_i) &> (1-q_i')\gamma_i p_i[j|\emptyset]\\
\Rightarrow (q_i' - q_i^*)(\beta_i - \gamma_i \bar{p}_i) &> 0.
\end{array}
\end{equation}

By Proposition~\ref{prop:folk}, we know that $\beta_i > \gamma_i \bar{p}_i$. Thus,~\ref{eq:nbtc-7} implies $q_i' > q_i^*$, which concludes the first part of the proof.

Furthermore, a solution to~\ref{eq:nbtc-6} exists if and only if
there is a maximum of the polynomial for $\omega_i \in (0,1)$. We can find the zero of the first derivative in order to $\omega_i$, 
obtaining 
\begin{equation}
\label{eq:nbtc-8}
\omega_i^{\pd} = \frac{a-b+c}{(a-b)(\pd+1)}. 
\end{equation}

For a solution of~\ref{eq:nbtc-8} to exist for some $\omega_i \in (0,1)$, it must be true that: 
\begin{enumerate}
 \item $a\pd > b\pd + c$.
 \item $a - b + c > 0$.
\end{enumerate}

Condition~1) yields:
\begin{equation}
\label{eq:nbtc-9}
\begin{array}{ll}
\pd(q_i' - q_i^*)(\beta_i - \gamma_i \bar{p}_i) \pd  & > \pd(1-q_i')\gamma_i p_i[j|\emptyset]+ (1-q_i^*)\gamma_i p_i[j|\emptyset]\\
\pd(q_i' - q_i^*)\beta_i & > \pd(q_i' - q_i^*)\gamma_i \bar{p}_i + \pd(1-q_i')\gamma_i p_i[j|\emptyset] + (1-q_i^*)\gamma_i p_i[j|\emptyset]\\
\beta_i & > \gamma_i \bar{p}_i + \frac{1-q_i'}{q_i'-q_i^*}\gamma_i p_i[j|\emptyset] + \frac{1-q_i^*}{\pd(q_i' - q_i^*)} \gamma_i p_i[j|\emptyset]\\
\frac{\beta_i}{\gamma_i} & > \bar{p}_i + \frac{1}{q_i'-q_i^*}p_i[j|\emptyset]( 1-q_i' + \frac{1-q_i^*}{\pd})\\
\end{array}
\end{equation}
which is equivalent to~\ref{eq:nec-btc}. If Condition~1 is true, then:
$$a > b + \frac{c}{\pd} \Rightarrow a> b - c.$$
That is, Condition~1 implies Condition~2. Therefore, if $\vec{\sigma}^*$ is a SPE for some $\omega_i \in (0,1)$,
then~\ref{eq:nec-btc} must hold, which proves the result.
\end{proof}


\newpage
\subsubsection{Proof of Lemma~\ref{lemma:dir-recip}.}
\label{proof:lemma:dir-recip}
Suppose that for any $i \in {\cal N}$ and $j \in {\cal N}_i$, $p_i[j|\emptyset] + q_i^* \ll 1$.
If $\vec{\sigma}^*$ is a SPE, then:
$$\psi[\vec{\sigma}^*] \subseteq \left(\frac{1}{q_i^*},\infty\right).$$

\begin{proof}
Fix $i$ and $j$ for which the assumption holds.
Let $\vec{p}_i'$ be defined as:
\begin{itemize}
  \item $p_i'[j] = 0$,
  \item $p_i'[k] = p_i[k|\emptyset]$ for every $k \in {\cal N}_i \setminus \{j\}$.
\end{itemize}

Let $q_i' = q_i[(\vec{p}_i',\vec{p}^*_{-i})]$ and $q_i^* = q_i[\vec{\sigma}^*[\emptyset]]$,
where $\vec{p}^* = \vec{\sigma}[\emptyset]$.
By Lemma~\ref{lemma:nec-btc}, we must have
\begin{equation}
\label{eq:dir-recip-1}
\frac{\beta_i}{\gamma_i} > \bar{p}_i + \frac{p_i[j|\emptyset]}{q_i' - q_i^*}\left(1-q_i' + \frac{1-q_i^*}{\pd}\right).
\end{equation}
By Lemma~\ref{lemma:single-impact} from Appendix~\ref{sec:epidemic}, it is true that $q_i'  \leq \frac{q_i^*}{1-p_i[j|\emptyset]}$.
By including this fact in~\ref{eq:dir-recip-1}, we obtain:
\begin{equation}
\label{eq:dr-5}
\begin{array}{ll}
\frac{\beta_i}{\gamma_i} &> \bar{p}_i + \frac{p_i[j|\emptyset]}{q_i' - q_i^*}\left(1-q_i' + \frac{1-q_i^*}{\pd}\right) > \bar{p}_i + \frac{p_i[j|\emptyset]}{q_i' - q_i^*}(1 - q_i')\\
                                            &\geq \bar{p}_i + \frac{p_i[j|\emptyset](1-p_i[j|\emptyset])}{(1-p_i[j|\emptyset])(q_i^* - q_i^*(1-p_i[j|\emptyset]))}(1- p_i[j|\emptyset] - q_i^*)\\
                                            &= \bar{p}_i + \frac{1}{q_i^*} (1 - p_i[j|\emptyset] - q_i^*)\\
                                            &\approx \bar{p}_i + \frac{1}{q_i^*}\\
                                            &> \frac{1}{q_i^*}.
\end{array}
\end{equation}
The result follows from the fact that if for every $i$ there exists $\omega_i \in (0,1)$ such that $\vec{\sigma}^*$ is a SPE, then $[0,\frac{1}{q_i^*}] \cap \psi[\vec{\sigma}^*] = \emptyset$.

\end{proof}

\newpage

\subsection{Full Indirect Reciprocity is Sufficient}
\label{sec:proof:indir}

\subsubsection{Proof of Lemma~\ref{lemma:indir-equiv}.}
\label{proof:lemma:indir-equiv}
The profile of strategies $\vec{\sigma}^*$ is a SPE if and only if for every $h \in {\cal H}$ and $i \in {\cal N}$:
$$\sum_{r=1}^\pd (\omega_i^r u_i[h,r|\vec{\sigma}^*]) - (1-q_i[h,0|\vec{\sigma}^*])\gamma_i \bar{p}_i[h,0|\vec{\sigma}^*] \geq 0.$$

\begin{proof}
Using the result from Theorem~\ref{theorem:gen-cond}, it is true that $\vec{\sigma}^*$ is a SPE if and only if the DC Condition
is true for every $i \in {\cal N}$, $h \in {\cal H}$, and $D \subseteq {\cal N}_i[h]$.

Consider any strategy $\sigma_i'= \sigma_i^*[h|\vec{p}']$ where $\vec{p}' = (\vec{0},\vec{p}_{-i}^*)$ and $\vec{p}^*=\vec{\sigma}^*[h]$.
Alternatively, define $\sigma_i''=\sigma_i^*[h|\vec{p}'']$ where $\vec{p}'' = (\vec{p}_i'',\vec{p}^*_{-i})$ such that for some $D \subset {\cal N}_i[h]$:
\begin{itemize}
  \item For every $j \in D$, $p_i''[j] = 0$.
  \item For every $j \in {\cal N}_i \setminus D$, $p_i''[j] = p_i[j|h]$.
\end{itemize}
Define $h_r^* = \hevol[r|\vec{\sigma}^*]$.
Let $\vec{\sigma}' = (\sigma_i',\vec{\sigma}_{-i}^*)$, $h_r' = \hevol[r|\vec{\sigma}']$,
$\vec{\sigma}'' = (\sigma_i'',\vec{\sigma}_{-i}^*)$, and $h_r'' = \hevol[r|\vec{\sigma}']$.
By Lemma~\ref{lemma:corr-1} and the definition of full indirect reciprocity, 
for every $r \in \{1 \ldots \pd\}$ and $j \in {\cal N}_i^{-1}$, since for every $k \in {\cal N}\setminus \{i\}$,
$$
\begin{array}{l}
\cd_i[\vec{p}'|h] = {\cal N}_i,\\
\cd_i[\vec{p}''|h] = D,\\
\cd_k[\vec{p}'|h] = \cd_k[\vec{p}''|h] = \emptyset,
\end{array}
$$
it holds that
\begin{equation}
\label{eq:indir-2}
\begin{array}{ll}
\ds_j[i|h_r'] & = \ds_j[i|h_r^*] \cup \{(k_1,k_2,r-1)|k_1,k_2 \in {\cal N} \land k_2 \in \cd_{k_1}[\vec{p}'|h] \land j,i \in \rs[k_1,k_2]\}\\
                      & = \ds_j[i|h_r^*] \cup \{(i,k,r-1)|k \in {\cal N}_i \land j,i \in \rs[i,k]\}.\\
                      & = \ds_j[i|h_r^*] \cup \{(i,k,r-1)|k \in {\cal N}_i\}.
\end{array}
\end{equation}
and
\begin{equation}
\label{eq:indir-3}
\begin{array}{ll}
\ds_j[i|h_r''] & = \ds_j[i|h_r^*] \cup \{(k_1,k_2,r)|k_1,k_2 \in {\cal N} \land k_2 \in \cd_{k_1}[\vec{p}''|h] \land j,i \in \rs[k_1,k_2]\}\\
             & = \ds_j[i|h_r^*] \cup \{(i,k,r-1)|k \in D \land j,i \in \rs[i,k]\}.\\
                      & = \ds_j[i|h_r^*] \cup \{(i,k,r-1)|k \in D\}.
\end{array}
\end{equation}

By Definition~\ref{def:thr}, and by~\ref{eq:indir-2} and~\ref{eq:indir-3}, for every $r \in \{1 \ldots \pd\}$ and $j \in {\cal N}_i^{-1}$,
\begin{equation}
\label{eq:indir-4}
p_j[i|h_r'] = p_j[i|h_r''] = 0.
\end{equation}

It follows from~\ref{eq:indir-4} and Lemma~\ref{lemma:noneib} that for every $r \in \{1 \ldots \pd\}$
\begin{equation}
\label{eq:indir-5}
\begin{array}{ll}
q_i[h,r|\vec{\sigma}'] &= q_i[h,r|\vec{\sigma}''] = 1.\\
u_i[h,r|\vec{\sigma}'] &= u_i[h,r|\vec{\sigma}''] = 0.\\
\end{array}
\end{equation}

Furthermore, we have
$$\bar{p}_i[h,0|\vec{\sigma}'] \leq \bar{p}_i[h,0|\vec{\sigma}''],$$
which implies that
\begin{equation}
\label{eq:indir-6}
u_i[h,0|\vec{\sigma}'] > u_i[h,0|\vec{\sigma}''].
\end{equation}

By~\ref{eq:indir-5} and~\ref{eq:indir-6},
\begin{equation}
\label{eq:indir-7}
\sum_{r=0}^\pd \omega^r (u_i[h,r|\vec{\sigma}^*] - u_i[h,r|\vec{\sigma}']) < \sum_{r=0}^\pd \omega^r(u_i[h,r|\vec{\sigma}^*] - u_i[h,r|\vec{\sigma}'']).
\end{equation}

Finally, we have
\begin{equation}
\label{eq:indir-8}
\begin{array}{ll}
\sum_{r=0}^\pd \omega^r (u_i[h,r|\vec{\sigma}^*] - u_i[h,r|\vec{\sigma}']) &\geq 0\\
(1-q_i[h,0|\vec{\sigma}^*])(\beta_i - \bar{p}_i[h,0|\vec{\sigma}^*]) - (1-q_i[h,0|\vec{\sigma}^*])\beta_i + \sum_{r = 1}^{\pd} \omega^r u_i[h,r|\vec{\sigma}^*] & \geq 0\\
\sum_{r = 1}^{\pd} \omega^r u_i[h,r|\vec{\sigma}^*] - (1-q_i[h,0|\vec{\sigma}^*])\bar{p}_i[h,0|\vec{\sigma}^*]& \geq 0.
\end{array}
\end{equation}

It is direct to conclude by Theorem~\ref{theorem:gen-cond} that if $\vec{\sigma}^*$ is a SPE, then DC Condition is fulfilled for
$D= {\cal N}_i$. By~\ref{eq:indir-7}, if DC Condition is fulfilled for $D={\cal N}_i$, then it is also fulfilled for
every $D \subset {\cal N}_i$, and by Theorem~\ref{theorem:gen-cond}, $\vec{\sigma}^*$ is a SPE.
Thus, $\vec{\sigma}^*$ is a SPE iff~\ref{eq:indir-8} holds. This concludes the proof.
\end{proof}

\newpage
\subsubsection{Proof of Lemma~\ref{lemma:maxh}}
\label{proof:lemma:maxh}
Let $h \in {\cal H}$ be defined such that for every $h' \in {\cal H}$,
the left side of Inequality~\ref{eq:indir-equiv} for $h$ is lower than or equal to the value for $h'$.
Then, for every $r \in \{1 \ldots \pd-2\}$,
$$u_i[h,r|\vec{\sigma}^*] = u_i[h,r+1|\vec{\sigma}^*].$$

\begin{proof}
The proof goes by contradiction. First, assume that $h$ minimizes the left side of Inequality~\ref{eq:indir-equiv}:
$$\sum_{r=1}^\pd (\omega_i^r u_i[h,r|\vec{\sigma}^*]) - (1-q_i[h,0|\vec{\sigma}^*])\gamma_i \bar{p}_i[h,0|\vec{\sigma}^*] \geq 0.$$
This implies that for every $h' \in {\cal H}$, we have:
\begin{equation}
\label{eq:mh}
\begin{array}{ll}
  \sum_{r=1}^\pd (\omega_i^r u_i[h,r|\vec{\sigma}^*])  - (1-q_i[h,0|\vec{\sigma}^*])\gamma_i \bar{p}_i[h,0|\vec{\sigma}^*]& \\
  - \sum_{r=1}^\pd (\omega_i^r u_i[h',r|\vec{\sigma}^*]) - (1-q_i[h',0|\vec{\sigma}^*])\gamma_i \bar{p}_i[h',0|\vec{\sigma}^*] & \leq 0\\
  \\
  \sum_{r=1}^\pd \omega_i^r (u_i[h,r|\vec{\sigma}^*] - u_i[h',r|\vec{\sigma}^*]) &\\
   + (1-q_i[h,0|\vec{\sigma}^*])\gamma_i (\gamma_i \bar{p}_i[h',0|\vec{\sigma}^*] - \bar{p}_i[h,0|\vec{\sigma}^*]) & \leq 0\\
\end{array}
\end{equation}

Assume by contradiction that for every $h$ that minimizes the above condition, there is some $r \in \{1 \ldots \pd -3\}$, 
$$u_i[h,r|\vec{\sigma}^*] \neq u_i[h,r+1|\vec{\sigma}^*].$$
Fix $h$. Without loss of generality, suppose
\begin{equation}
\label{eq:mh-0}
u_i[h,r|\vec{\sigma}^*] < u_i[h,r+1|\vec{\sigma}^*].
\end{equation}

By  Lemma~\ref{lemma:corr-0} and Definition~\ref{def:thr}, this implies the existence of $D \subset {\cal N}^2$, such that, for every $j \in {\cal N}$ and $k \in {\cal N}_j$,
\begin{equation}
\label{eq:mh-1}
\begin{array}{l}
\ds_j[k|h_{r+1}^*] = \ds_j[k|h_{r}^*] \setminus \{(k_1,k_2,\pd-1)|(k_1,k_2) \in D \land j,k \in \rs[k_1,k_2]\},\\
\ds_j[k|h_{r}^*] = \ds_j[k|h_{r+1}^*]\cup \{(k_1,k_2,\pd-1)|(k_1,k_2) \in D \land j,k \in \rs[k_1,k_2]\},\\
(k_1,k_2,\pd - r-1) \in \ds_j[k|h],
\end{array}
\end{equation}
where $h_r^* = \hevol[h,r|\vec{\sigma}^*]$, and there exist $j,k$ and $(k_1,k_2) \in D$ such that $i,j \in \rs[k_1,k_2]$, for instance, $k_1$ and $k_2$. 

Notice that, if~\ref{eq:mh-0} is true, then $\pd - r -1\geq1$.
Thus, we can define $h'$ such that:
\begin{itemize}
  \item $|h'| = |h|$.
  \item $\ds_j[k|h'] = \ds_j[k|h] \setminus \{(k_1,k_2,\pd - r-1) | (k_1,k_2) \in D\} \cup \{(k_1,k_2,\pd - r -2) | (k_1,k_2) \in D\}$.
\end{itemize}
Let $h_r' = \hevol[h',r|\vec{\sigma}^*]$. By Lemma~\ref{lemma:corr-0}, for every $r' \in \{0\ldots r-1\}$, $j \in {\cal N}$, and $k \in {\cal N}_j$,
\begin{equation}
\label{eq:mh-2}
\begin{array}{ll}
\ds_j[k|h_{r'+1}'] & = \{(l_1,l_2,r''+1) | (l_1,l_2,r'') \in \ds_j[k|h_{r'}'] \land r'' +1 < \pd \}\\
&\\
			   & = \{(l_1,l_2,r''+r') | (l_1,l_2,r'') \in (\ds_j[k|h] \setminus \{(k_1,k_2,\pd - r- 1) | (k_1,k_2) \in D\}\\
			   & \cup \{(k_1,k_2,\pd - r- 2) | (k_1,k_2) \in D\land j,k \in \rs[k_1,k_2]\}) \land r'' +r' < \pd \}\\
			   &\\
			   & = \{(k_1,k_2,r''+r') | (k_1,k_2,r'') \in \ds_j[k|h] \land (k_1,k_2) \notin D \land r'' +r' <\pd \}\\
			   & \cup \{(k_1,k_2,\pd - r+r'-2) | (k_1,k_2) \in D\land j,k \in \rs[k_1,k_2] \land \pd - r - 2 + r' < \pd \}\\
			   &\\
			   & = \ds_j[k|h_{r'+1}^*] \setminus \{(k_1,k_2,\pd-r+r'-1) | (k_1,k_2) \in D \land j,k \in \rs[k_1,k_2]\} \\
			   &\cup \{(k_1,k_2,\pd - r + r' - 2) | (k_1,k_2) \in D\land j,k \in \rs[k_1,k_2] \land r' -2 < r\}.\\
			   &\\
          		   & = \ds_j[k|h_{r'+1}^*] \setminus \{(k_1,k_2,\pd-r+r'-1) | (k_1,k_2) \in D \land j,k \in \rs[k_1,k_2]\} \\
			   &\cup \{(k_1,k_2,\pd - r + r' - 2) | (k_1,k_2) \in D\land j,k \in \rs[k_1,k_2] \land r' -1 < r\}.\\
			   &\\
			   & = \ds_j[k|h_{r'+1}^*] \setminus \{(k_1,k_2,\pd-r+r'-1) | (k_1,k_2) \in D \land j,k \in \rs[k_1,k_2]\} \\
			   &\cup \{(k_1,k_2,\pd - r + r' - 1) | (k_1,k_2) \in D\land j,k \in \rs[k_1,k_2]\}.\\
			   &\\
			   &=\ds_j[k|h_{r'+1}^*].
\end{array}
\end{equation}

By the definition of $\ds_j[k|h]$ and $\ds_j[k|h']$ and by Definition~\ref{def:thr}, for every $r' \in \{0\ldots r\}$, $j \in {\cal N}$, and $k \in {\cal N}_j$,
\begin{equation}
\label{eq:mh-3}
\begin{array}{l}
p_j[k|h_{r'}'] = p_j[k|h_{r'}^*].\\
\bar{p}_i[h',r'|\vec{\sigma}^*] = \bar{p}_i[h,r'|\vec{\sigma}^*].\\
q_i[h',r'|\vec{\sigma}^*] = q_i[h,r'|\vec{\sigma}^*].\\
u_i[h',r'|\vec{\sigma}^*] = u_i[h,r'|\vec{\sigma}^*].
\end{array}
\end{equation}

Furthermore, by~\ref{eq:mh-1},
\begin{equation}
\label{eq:mh-4}
\begin{array}{ll}
\ds_j[k|h_{r+1}'] & = \{(l_1,l_2,r'+1) | (l_1,l_2,r') \in \ds_j[k|h_{r}'] \land r' +1 <\pd \}\\
			   &\\
			   & = \{(l_1,l_2,r'+1) | (l_1,l_2,r') \in (\ds_j[k|h_{r}^*] \cup \{(k_1,k_2,\pd-2)|(k_1,k_2) \in D \land j,k \in \rs[k_1,k_2]\}) \\
			   &\land r' +1 < \pd\})\\
			   &\\
			   & = \{(k_1,k_2,r'+1) | (k_1,k_2,r') \in \ds_j[k|h_{r}^*] \} \cup \\
			   & \{(k_1,k_2,\pd-1)|(k_1,k_2) \in D \land j,k \in \rs[k_1,k_2] \land \pd -1 < \pd\}\\
			   &\\
                               & = \ds_j[k|h_{r+1}^*] \cup \{(k_1,k_2,\pd - 1)|(k_1,k_2) \in D\land j,k \in \rs[k_1,k_2]\}\\
                               & = \ds_j[k|h_{r}^*].
\end{array}
\end{equation}

By Definition~\ref{def:thr} and~\ref{eq:mh-4},
\begin{equation}
\label{eq:mh-5}
\begin{array}{l}
p_j[k|h_{r+1}'] = p_j[k|{h_{r}^*}].\\
\bar{p}_i[h',r+1|\vec{\sigma}^*] = \bar{p}_i[h,r|\vec{\sigma}^*].\\
q_i[h',r+1|\vec{\sigma}^*] = q_i[h,r|\vec{\sigma}^*].\\
u_i[h',r+1|\vec{\sigma}^*] = u_i[h,r|\vec{\sigma}^*].
\end{array}
\end{equation}

Finally, by Lemma~\ref{lemma:corr-0}, for every $r'>r$, $j \in {\cal N}$, and $k \in {\cal N}_j$,
\begin{equation}
\label{eq:mh-6}
\begin{array}{ll}
\ds_j[k|h_{r'+1}'] & = \{(l_1,l_2,r''+1) | (l_1,l_2,r'') \in \ds_j[k|h_{r'}'] \land r''+1 <\pd \}\\
			   &= \{(l_1,l_2,r'' +1)|  (l_1,l_2,r'') \in (\ds_j[k|h] \setminus \\
			   &\{(k_1,k_2,\pd - r -1) | (k_1,k_2) \in D \land  j,k \in \rs[k_1,k_2]\}\\ 
			   &\cup \{(k_1,k_2,\pd - r- 2) | (k_1,k_2) \in D \land j,k \in \rs[k_1,k_2]\}) \land r'' +1 <\pd \}\\
			   &\\
			   &= \{(l_1,l_2,r''+r'+1)| (l_1,l_2,r'') \in \ds_j[k|h] \land r'' + r' +1< \pd\}\\ 
			   & \setminus \{(k_1,k_2,\pd - r+r' - 1) | (k_1,k_2) \in D \land j,k \in \rs[k_1,k_2] \land r'< r\}\\
			   &\cup \{(k_1,k_2,\pd - r+r'-2)| (k_1,k_2) \in D \land j,k \in \rs[k_1,k_2] \land r'< r+1 \}\\
			   &\\
			   &= \{(l_1,l_2,r''+r')| (l_1,l_2,r'') \in \ds_j[k|h] \land r'' + r' < \pd\}\\ 
			   & \setminus \{(k_1,k_2,\pd - r+r' - 1) | (k_1,k_2) \in D \land j,k \in \rs[k_1,k_2]\land r'< r\}\\
			   &\cup \{(k_1,k_2,\pd - r+r'-2)| (k_1,k_2) \in D\land j,k \in \rs[k_1,k_2] \land r'< r\}\\	   
			   &\\
			   & = \{(l_1,l_2,r''-r') | (l_1,l_2,r'') \in \ds_j[k|h] \land r'' > r' \}\\
			   & = \ds_j[k|h_{r'+1}^*].
\end{array}
\end{equation}

Therefore, by Definition~\ref{def:thr}, for every $r' > r+1$, $j \in {\cal N}$, and $k \in {\cal N}_j$,
\begin{equation}
\label{eq:mh-7}
\begin{array}{l}
p_j[k|h_{r'}'] = p_j[k|h_{r'}^*].\\
\bar{p}_i[h',r'|\vec{\sigma}^*] = \bar{p}_i[h,r'|\vec{\sigma}^*].\\
q_i[h',r'|\vec{\sigma}^*] = q_i[h,r'|\vec{\sigma}^*].\\
u_i[h',r'|\vec{\sigma}^*] = u_i[h,r'|\vec{\sigma}^*].
\end{array}
\end{equation}

By~\ref{eq:mh-0},~\ref{eq:mh-3},~\ref{eq:mh-5}, and~\ref{eq:mh-7},
$$
\begin{array}{ll}
\sum_{r'=1}^\pd \omega_i^{r'}(u_i[h,r'|\vec{\sigma}^*] - u_i[h',r'|\vec{\sigma}^*])  & + (1-q_i[h,0|\vec{\sigma}^*])\gamma_i (\gamma_i \bar{p}_i[h',0|\vec{\sigma}^*] - \bar{p}_i[h,0|\vec{\sigma}^*])\\
 = u_i[h,r+1|\vec{\sigma}^*] - u_i[h',r+1|\vec{\sigma}^*] & = u_i[h,r+1|\vec{\sigma}^*] - u_i[h,r|\vec{\sigma}^*] > 0.
\end{array}
$$
This is a contradiction to~\ref{eq:mh}, which concludes the proof. For the case where 
$$u_i[h,r|\vec{\sigma}^*] > u_i[h,r+1|\vec{\sigma}^*],$$
the proof is identical, except that $h'$ is defined such that the punishments that end at stage $r$ are anticipated, i.e.:

$$\ds_j[k|h'] = \ds_j[k|h] \setminus \{(k_1,k_2,\pd - r-1) | (k_1,k_2) \in D\} \cup \{(k_1,k_2,\pd - r ) | (k_1,k_2) \in D\}.$$
\end{proof}

\newpage

\subsubsection{Proof of Theorem~\ref{theorem:indir-suff}.}
\label{proof:theorem:indir-suff}
If there exists a constant $c \geq 1$ such that, for every $h \in {\cal H}$ and $i \in {\cal N}$, Assumption~\ref{eq:indir-assum} holds, then
$\psi[\vec{\sigma}^*] \supseteq (v,\infty)$, where
$$v = \max_{h \in {\cal H}}\max_{i \in {\cal N}} \bar{p}_i[h,0|\vec{\sigma}^*]\left(1 + \frac{c}{\pd}\right).$$

\begin{proof}
Assume by contradiction that, for any player $i \in {\cal N}$, $\vec{\sigma}^*$ is not a SPE for any $\omega_i \in (0,1)$ and that
$$\frac{\beta_i}{\gamma_i} > \max_{h \in {\cal H}} \bar{p}_i[h,0|\vec{\sigma}^*]\left(1 + \frac{c}{\pd}\right).$$
Fix $i$. The proof considers a history $h$ that minimizes the left side of Inequality~\ref{eq:indir-equiv}. The reason for this is that, if Inequality~\ref{eq:indir-equiv}
is true for $h$, then it is also true for every other history $h'$. If $\bar{p}_i[h,0|\vec{\sigma}^*] = 0$, then the inequality is trivially fulfilled.
Hence, consider that $\bar{p}_i[h,0|\vec{\sigma}^*] > 0$. By Lemma~\ref{lemma:maxh}, for every $r \in \{1 \ldots \pd-1\}$,
$$u_i[h,r|\vec{\sigma}^*] = u_i[h,r+1|\vec{\sigma}^*].$$

We are left with stage $\pd$. Let $u_i^h = u_i[h,1|\vec{\sigma}^*]$. If for every $k \in {\cal N}$ and $l \in {\cal N}_k$ we have $\ds_k[l|h_{\pd}^*] = \ds_k[l|h_{\pd-1}^*]$, where $h_r^* = \hevol[h,r|\vec{\sigma}^*]$, then 
$$u_i^h = u_i[h,\pd|\vec{\sigma}^*].$$
Otherwise, by Lemma~\ref{lemma:maxh},
$$u_i^h \leq u_i[h,\pd|\vec{\sigma}^*].$$
Either way, by Lemma~\ref{lemma:maxh}, for every $h' \in {\cal H}$:
\begin{equation}
\label{eq:is-1}
\begin{array}{ll}
-(1-q_i[h,0|\vec{\sigma}^*])\gamma_i \bar{p}_i[h,0|\vec{\sigma}^*] + \sum_{r = 1}^\pd \omega_i^r u_i^h & \leq \\
-(1-q_i[h',0|\vec{\sigma}^*])\gamma_i \bar{p}_i[h',0|\vec{\sigma}^*] +  \sum_{r=1}^\pd \omega_i^r u_i[h',r|\vec{\sigma}^*].
\end{array}
\end{equation}

We can write:
\begin{equation}
\label{eq:is-2}
\begin{array}{ll}
-(1-q_i[h,0|\vec{\sigma}^*])\gamma_i \bar{p}_i[h,0|\vec{\sigma}^*] + \sum_{r = 1}^\pd \omega_i^r u_i^h & \geq 0\\
-a + \frac{\omega_i - \omega_i^{\pd+1}}{1 - \omega_i} u_i^h & \geq 0\\
-a  + \omega_i(u_i^h + a) - \omega_i^{\pd+1} u_i^h & \geq 0,
\end{array}
\end{equation}
where 
$$a = (1-q_i[h,0|\vec{\sigma}^*])\gamma_i \bar{p}_i[h,0|\vec{\sigma}^*].$$
Again, this Inequality corresponds to a polynomial with degree $\pd+1$. If $q_i[h,1|\vec{\sigma}^*] = 1$, then by our assumptions $q_i[h,0|\vec{\sigma}^*] =1$, $a=0$, and
the Inequality holds. Suppose then that 
$$q_i[h,1|\vec{\sigma}^*],q_i[h,0|\vec{\sigma}^*] < 1.$$

The polynomial has a zero in $\omega_i = 1$. If $\bar{p}_i[h,0|\vec{\sigma}^*] = 0$, then $a=0$ and the Inequality holds.
Consider, then, that $\bar{p}_i[h,0|\vec{\sigma}^*] > 0$, which implies that $a>0$. In these circumstances, a solution to~\ref{eq:is-2}
exists for $\omega_i \in (0,1)$ iff the polynomial is strictly concave and has another zero in $(0,1)$. This is true iff
the polynomial has a maximum in $(0,1)$. The derivatives yield the following conditions:
\begin{enumerate}
 \item $\exists_{\omega_i \in (0,1)} u_i^h + a - (\pd+1) \omega_i u_i^h = 0 \Rightarrow \exists_{\omega_i \in (0,1)} \omega_i = \frac{u_i^h+a}{(\pd+1)u_i^h}$.
 \item $-(\pd+1)\pd u_i^h < 0 \Rightarrow u_i^h > 0$.
\end{enumerate}

Condition~$1$ implies:
\begin{equation}
\label{eq:is-3}
\begin{array}{ll}
u_i^h\pd & > a\\
(1-q_i[h,1|\vec{\sigma}^*]) \beta_i \pd &> (1-q_i[h,1|\vec{\sigma}^*])\gamma_i\bar{p}_i[h,1|\vec{\sigma}^*] \pd + (1-q_i[h,0|\vec{\sigma}^*])\gamma_i\bar{p}_i[h,0|\vec{\sigma}^*]\\
\frac{\beta_i}{\gamma_i} &> \bar{p}_i[h,1|\vec{\sigma}^*] + \bar{p}_i[h,0|\vec{\sigma}^*]\frac{1}{\pd} \frac{1-q_i[h,0|\vec{\sigma}^*]}{1-q_i[h,1|\vec{\sigma}^*]}.
\end{array}
\end{equation}

By the assumption that $q_i[h,0|\vec{\sigma}^*] \leq 1-c(1-q_i[h,1|\vec{\sigma}^*])$, if Inequality~\ref{eq:indir-suff} is true, then so is~\ref{eq:is-3}.
Furthermore, it also holds that
$$\beta_i > \gamma_i \bar{p}_i[h,1|\vec{\sigma}^*] \Rightarrow u_i^h >0,$$
thus, Condition~$2$ and~\ref{eq:is-2} are also true for some $\omega_i \in (0,1)$. By~\ref{eq:is-1},
for every $h' \in {\cal H}$ and $i \in {\cal N}$, Inequality~\ref{eq:indir-equiv} is true,
implying by Lemma~\ref{lemma:indir-equiv} that $\vec{\sigma}^*$ is a SPE for some value $\omega_i \in (0,1)$.
This is a contradiction, proving that if for every $i \in {\cal N}$ we have $\frac{\beta_i}{\gamma_i} \in (v,\infty)$, then $\vec{\sigma}^*$ is a SPE.
By the definition of $\psi$, $\psi[\vec{\sigma}^*] \supseteq (v,\infty)$.
\end{proof}

\newpage
\section{Private Monitoring}
\label{sec:proof:private}

\subsection{Evolution of the Network}
\label{sec:proof-priv-evol}
\subsubsection{Auxiliary Lemma.}
\begin{lemma}
\label{lemma:pevol-aux}
For every $h \in {\cal H}$, $\vec{p}' \in {\cal P}$,
$i \in {\cal N}$, $j \in {\cal N}_i$, $k_1,k_2 \in {\cal N}$, and $r \in \{0 \ldots \del_i[k_1,k_2] + \pd[k_1,k_2|i,j] - v[k_1,k_2] -1\}$,
where $v[k_1,k_2] = \min[\del_i[k_1,k_2] - \del_j[k_1,k_2],0]$,
let $s_i' \in \sig[\vec{\sigma}'[h_{r}']|h_{r}']$ and $s_i^* \in \sig[\vec{\sigma}^*[h_{r}^*]|h_{r}^*]$,
where $h_r^* = \hevol[h,r|\vec{\sigma}^*]$, $h_r' =  \hevol[h,r|\vec{\sigma}']$, and $\vec{\sigma}' = \vec{\sigma}^*[h|\vec{p}']$.
Then, we have
\begin{enumerate}
  \item If $k_2 \notin \cd_{k_1}[\vec{p}'|h]$ or $k_2 \in \cd_{k_1}[\vec{p}'|h]$ and $\del_i[k_1,k_2] > r$,
we have $s_i^*[k_1,k_2] = s_i'[k_1,k_2]$.
 \item If $k_2 \in \cd_{k_1}[\vec{p}'|h]$ and $\del_i[k_1,k_2] = r$, then $s_i^*[k_1,k_2] = \mbox{\emph{cooperate}}$ and $s_i'[k_1,k_2] = \mbox{\emph{defect}}$.
 \item Else, $s_i^*[k_1,k_2] = s_i'[k_1,k_2] = \mbox{\emph{cooperate}}$.
\end{enumerate}
\end{lemma}
\begin{proof}
Fix $i$, $j$, $k_1$, $k_2$, $h$, and $\vec{p}'$.

Let $s_{k_2}' \in \sig[\vec{\sigma}'[h]|h]$ and $s_{k_2}^* \in \sig[\vec{\sigma}^*[h]|h]$. 
For $k_2 \notin \cd_{k_1}[\vec{p}'|h]$ and $h_{k_1} \in h$, 
$$p_{k_1}'[k_2] = p_{k_1}[k_2|h_{k_1}],$$
and, by Definition~\ref{def:privsig}, 
$$s_i^*[k_1,k_2] = s_i'[k_1,k_2].$$

If $k_2 \in \cd_{k_1}[\vec{p}'|h]$ and $r = \del_i[k_1,k_2]$, then, by the definition of $\vec{\sigma}'$ and $\vec{\sigma}^*$,
$$s_{k_2}'[k_1,k_2] = \mbox{\emph{defect}},$$
$$s_{k_2}^*[k_1,k_2] = \mbox{\emph{cooperate}},$$
and, by Definition~\ref{def:privsig}, 
$$s_i'[k_1,k_2]  = \mbox{\emph{defect}}.$$
$$s_i^*[k_1,k_2] = \mbox{\emph{cooperate}}.$$

If $r > \del_i[k_1,k_2]$, then
$$s_i'[k_1,k_2] = s_i^*[k_1,k_2]= \mbox{\emph{cooperate}}.$$ 
To see this, assume first that $s_i'[k_1,k_2] = \mbox{\emph{defect}}$. Then,  by Definition~\ref{def:privsig}, there must
exist a round $r'>0$ and history $\hevol[h,r'|\vec{\sigma}']$ after which $k_2$ defects $k_1$.
That is, define $r' = r - \del_i[k_1,k_2]$. For $s_{k_2}'' \in \sig[\vec{\sigma}'[h_{r'}']|h_{r'}']$,
$s_{k_2''}[k_1,k_2] = \mbox{\emph{defect}}$, which is true iff $p_{k_1}''[k_2] < p_{k_1}[k_2|h_{r'}']$.
Since $r' > 0$, this contradicts the definition of $\vec{\sigma}'$. Hence, $s_i^*[k_1,k_2] = \mbox{\emph{cooperate}}$. 

If we assume that $s_i^*[k_1,k_2] = \mbox{\emph{defect}}$, then by Definition~\ref{def:privsig}, there must
exist a round $r'>0$ and history $\hevol[h,r'|\vec{\sigma}']$ after which $k_2$ defects $k_1$. As before, since $r'>0$,
another contradiction is reached and we can conclude that we must have $s_i^*[k_1,k_2]= \mbox{\emph{cooperate}}$.

For $r < \del_i[k_1,k_2]$,the result holds immediately. This is because by Definition~\ref{def:privsig}, $s_i'[k_1,k_2] = \mbox{\emph{defect}}$
iff $|h| + r\geq \del_i[k_1,k_2]$ and for $s_{k_2}'' = h_{k_2}^{\del_i[k_1,k_2]-r}$, $s_{k_2}''[k_1,k_2] = \mbox{\emph{defect}}$. This implies that 
$$s_{i}^*[k_1,k_2] = \mbox{\emph{defect}}.$$
Similarly, if $s_i'[k_1,k_2] = \mbox{\emph{cooperate}}$, then $s_{k_2}''[k_1,k_2] = \mbox{\emph{cooperate}}$, implying that
$$s_{i}^*[k_1,k_2] = \mbox{\emph{cooperate}}.$$
This proves the result.
\end{proof}

\newpage

\subsubsection{Proof of Lemma~\ref{lemma:priv-corr-1}.}
\label{proof:lemma:priv-corr-1}
For every $h \in {\cal H}$, $\vec{p}' \in {\cal P}$, $r >0$,
$i \in {\cal N}$, and $j \in {\cal N}_i$:
$$
\begin{array}{ll}
\ds_i[j|h_{i,r}'] =& \ds_i[j|h_{i,r}^*] \cup \{(k_1,k_2,r -1 - \del_i[k_1,k_2]+v[k_1,k_2]) | k_1,k_2 \in {\cal N} \land \\
                            & k_2 \in \cd_{k_1}[\vec{p}'|h] \land r \in \{\del_i[k_1,k_2]+1 \ldots \del_i[k_1,k_2] + \pd[k_1,k_2|i,j]-v[k_1,k_2]\}\land\\
                            & v[k_1,k_2] = \min[\del_i[k_1,k_2] - \del_j[k_1,k_2],0]\},\\
\end{array}
$$
where $h_{i,r}^* \in \hevol[h,r|\vec{\sigma}^*]$, $h_{i,r}' \in \hevol[h,r| \vec{\sigma}']$, and
$\vec{\sigma}' =\vec{\sigma}^*[h|\vec{p}']$ is the profile of strategies where all players follow $\vec{p}'$ in the first stage.

\begin{proof}
Fix $h$, $\vec{p}'$, $i$, and $j$. The proof goes by induction on $r$, where the induction hypothesis is that
for every $r \geq 0$, Equality~\ref{eq:priv-res-corr1} holds for $r+1$:
$$
\begin{array}{ll}
\ds_i[j|h_{i,r+1}'] =& \ds_i[j|h_{i,r+1}^*] \cup \{(k_1,k_2,r  - \del_i[k_1,k_2]+v[k_1,k_2]) | k_1,k_2 \in {\cal N} \land \\
                            & k_2 \in \cd_{k_1}[\vec{p}'|h] \land r+1 \in \{\del_i[k_1,k_2]+1 \ldots \del_i[k_1,k_2] + \pd[k_1,k_2|i,j]-v[k_1,k_2]\}\land\\
                            & v[k_1,k_2] = \min[\del_i[k_1,k_2] - \del_j[k_1,k_2],0]\},\\
\end{array}
$$

 We will simplify the notation by first 
dropping the factor $v[k_1,k_2] = \min[\del_i[k_1,k_2] - \del_j[k_1,k_2],0]$ whenever possible, and by removing the redundant
indexes $k_1,k_2,i,j$, except when distinguishing between different delays. We will also remove the factor $k_1,k_2 \in {\cal N}$.
Namely:
\begin{itemize}
 \item $\del_i[k_1,k_2]= \del_i$ and $\del_j[k_1,k_2] = \del_j$.
 \item $v[k_1,k_2] = v$.
 \item $\pd[k_1,k_2|i,j] = \pd$.
\end{itemize}

By Definition~\ref{def:priv-thr}, we have that for every $r \geq 0$, $h_r^* = \hevol[h,r|\vec{\sigma}^*]$, and $s_i^* \in \sig[\vec{\sigma}^*[h_r^*]|h_r^*]$,
\begin{equation}
\label{eq:priv-corr1}
\ds_i[j|h_{i,r+1}^*] = L_1[r+1|\vec{\sigma}^*] \cup L_2[r+1|\vec{\sigma}^*],
\end{equation}
where 
\begin{equation}
\label{eq:priv-corr1-1}
\begin{array}{ll}
L_1[r+1|\vec{\sigma}^*] = & \{(k_1,k_2,r'+1)|(k_1,k_2,r') \in \ds_{i}[j|h_{i,r}^*] \land r' +1 < \pd\}.\\
L_2[r+1|\vec{\sigma}^*] = & \{(k_1,k_2,v) | \land s_i^*[k_1,k_2] = \mbox{\emph{defect}}\}.
\end{array}
\end{equation}
Similarly, for every $r\geq0$, $h_r' = \hevol[h,r|\vec{\sigma}']$, and $s_i' \in \sig[\vec{\sigma}'[h_r']|h_r']$,
\begin{equation}
\label{eq:priv-corr1-2}
\ds_i[j|h_{i,r+1}'] = L_1[r+1|\vec{\sigma}'] \cup L_2[r+1|\vec{\sigma}'],
\end{equation}
where 
\begin{equation}
\label{eq:priv-corr1-3}
\begin{array}{ll}
L_1[r+1|\vec{\sigma}'] = & \{(k_1,k_2,r'+1)|(k_1,k_2,r') \in \ds_{i}[j|h_{i,r}'] \land r' +1 < \pd\}.\\
L_2[r+1|\vec{\sigma}'] = & \{(k_1,k_2,v) | s_i'[k_1,k_2] = \mbox{\emph{defect}}\}.
\end{array}
\end{equation}

For any $r \geq 0$, let $s_i' \in \sig[\vec{\sigma}'[h_{r}']|h_{r}']$ and $s_i^* \in \sig[\vec{\sigma}^*[h_{r}^*]|h_{r}^*]$,
where $h_r^* = \hevol[h,r|\vec{\sigma}^*]$, $h_r' =  \hevol[h,r|\vec{\sigma}']$, and $\vec{\sigma}' = \vec{\sigma}^*[h|\vec{p}']$.

By Lemma~\ref{lemma:pevol-aux}, we have:
\begin{enumerate}
 \item $s_i^*[k_1,k_2] = s_i'[k_1,k_2]$ for $k_2 \notin \cd_{k_1}[\vec{p}'|h]$ and $k_2 \in \cd_{k_1}[\vec{p}'|h]$ such that $\del_i> r$.
 \item $s_i'[k_1,k_2] = \mbox{\emph{defect}}$ and $s_i^*[k_1,k_2] = \mbox{\emph{cooperate}}$ for $k_2 \in \cd_{k_1}[\vec{p}'|h]$ such that $\del_i= r$.
 \item $s_i'[k_1,k_2] =s_i^*[k_1,k_2]= \mbox{\emph{cooperate}}$ for $k_2 \in \cd_{k_1}[\vec{p}'|h]$ such that $\del_i< r$.
\end{enumerate}

By~\ref{eq:priv-corr1-1}, and items 1), 2), and 3) above,
\begin{equation}
\label{eq:priv-corr1-4}
\begin{array}{ll}
L_2[r+1|\vec{\sigma}^*] &=  \{(k_1,k_2,v) | s_i^*[k_1,k_2] = \mbox{\emph{defect}}\}\\
			   		&=  \{(k_1,k_2,v) |  k_2 \notin \cd_{k_1}[\vec{p}'|h] \land s_i^*[k_1,k_2] = \mbox{\emph{defect}}\}\cup\\
					& \{(k_1,k_2,v) | k_2 \in \cd_{k_1}[\vec{p}'|h] \land \del_i > r  \land s_i^*[k_1,k_2] = \mbox{\emph{defect}}\},
\end{array}
\end{equation}
Equation~\ref{eq:priv-corr1-4}, 1), 2), and 3) allows us to write:
\begin{equation}
\label{eq:priv-corr1-5}
\begin{array}{ll}
L_2[r+1|\vec{\sigma}'] &= \{(k_1,k_2,v) | s_i'[k_1,k_2] = \mbox{\emph{defect}}\}\\
                                        &\\
                                         &= \{(k_1,k_2,v) | k_2 \notin \cd_{k_1}[\vec{p}',h] \land s_i'[k_1,k_2] = \mbox{\emph{defect}}\}\\
                                        & \cup \{(k_1,k_2,v) | k_2 \in \cd_{k_1}[\vec{p}'|h] \land  s_i'[k_1,k_2] = \mbox{\emph{defect}}\}\\
                                        &\\
                                        &=\{(k_1,k_2,v) | k_2 \notin \cd_{k_1}[\vec{p}'|h] \land s_i^*[k_1,k_2] = \mbox{\emph{defect}}\}\cup\\
		                     & \{(k_1,k_2,v) | k_2 \in \cd_{k_1}[\vec{p}'|h] \land \del_i > r  \land s_i^*[k_1,k_2] = \mbox{\emph{defect}}\} \cup\\
		                     & \{(k_1,k_2,v) | k_2 \in \cd_{k_1}[\vec{p}'|h] \land \del_i = r\}\\
		                     &\\
		                     & =L_2[r+1|\vec{\sigma}^*] \cup \{(k_1,k_2,v) | k_2 \in \cd_{k_1}[\vec{p}'|h] \land \del_i = r\}.
\end{array}
\end{equation}

Now proceed to the base case, for $r=0$. Since $h=\hevol[h,0|\vec{\sigma}^*] = \hevol[h,0|\vec{\sigma}']$, by~\ref{eq:priv-corr1-1} and~\ref{eq:priv-corr1-3},
it is true that: 
\begin{equation}
\label{eq:priv-corr1-6}
L_1[1|\vec{\sigma}'] = L_1[1|\vec{\sigma}^*].
\end{equation}

Furthermore, by~\ref{eq:priv-corr1-5},
\begin{equation}
\label{eq:priv-corr1-7}
\begin{array}{ll}
L_2[r+1|\vec{\sigma}'] &= L_2[r+1|\vec{\sigma}^*] \cup \{(k_1,k_2,v) | k_2 \in \cd_{k_1}[\vec{p}'|h] \land \del_i = 0\}\\
                                        &= L_2[r+1|\vec{\sigma}^*] \cup \{(k_1,k_2,r - \del_i + v) | k_2 \in \cd_{k_1}[\vec{p}'|h] \land 1 \geq \del_i +1 \land 1 < \del_i + \pd - v\}\\
                                        &= L_2[r+1|\vec{\sigma}^*] \cup \{(k_1,k_2,r- \del_i + v) | k_2 \in \cd_{k_1}[\vec{p}'|h] \land r +1 \in \{\del_i + 1 \ldots  \del_i + \pd - v\}\}.
\end{array}
\end{equation}

The base case is true by~\ref{eq:priv-corr1},~\ref{eq:priv-corr1-2},~\ref{eq:priv-corr1-6} and~\ref{eq:priv-corr1-7}.

Hence, assume the induction hypothesis for some $r \geq 0$ and consider the induction step for $r+1$, which
consists in determining the value of $\ds_i[j|h_{i,r+2}']$.

By the induction hypothesis and by~\ref{eq:priv-corr1-1},
\begin{equation}
\label{eq:priv-corr1-8}
\begin{array}{ll}
L_1[r+2|\vec{\sigma}'] & = \{(k_1,k_2,r'+1)|(k_1,k_2,r') \in (A \cup B)\land r'+1 < \pd\}\\
                                            & =  \{(k_1,k_2,r'+1)|(k_1,k_2,r') \in A \land r'+1 < \pd\} \cup \\
                                            & \{(k_1,k_2,r'+1)|(k_1,k_2,r') \in B \land r'+1 < \pd\},
\end{array}
\end{equation}

where 
$$
\begin{array}{ll}
A = & \ds_{i}[j|h_{i,r+1}^*]\\
B = &\{(k_1,k_2,r - \del_i+v) |k_2 \in \cd_{k_1}[\vec{p}'|h] \land r+1 \in \{\del_i+1 \ldots \del_i + \pd-v\}.
\end{array}
$$

We have by~\ref{eq:priv-corr1-1}
\begin{equation}
\label{eq:priv-corr1-9}
\begin{array}{ll}
 &\{(k_1,k_2,r'+1)|(k_1,k_2,r') \in A \land r'+1 < \pd\}\\
=& \{(k_1,k_2,r'+1)|(k_1,k_2,r') \in \ds_{i}[j|h_{i,r+1}^*] \land r'+1 < \pd\}\\
=L_1[r+2|\vec{\sigma}^*].
\end{array}
\end{equation}

It is also true that
\begin{equation}
\label{eq:priv-corr1-10}
\begin{array}{ll}
  &\{(k_1,k_2,r'+1)|(k_1,k_2,r') \in B \land r'+1 < \pd\}\\
=&  \{(l_1,l_2,r'+1) | (l_1,l_2,r') \in  \{(k_1,k_2,r - \del_i+v) | k_2 \in \cd_{k_1}[\vec{p}'|h] \land \\
               &r+1 \in \{\del_i +1\ldots \del_i + \pd-v\}\} \land r'+1<\pd\}\\
&\\
= & \{(k_1,k_2,r+1 - \del_i+v) |k_2 \in \cd_{k_1}[\vec{p}'|h] \land r+1 \in \{\del_i+1 \ldots \del_i + \pd-v\} \land r +1 - \del_i + v < \pd\}\\
   &\\
 =& \{(k_1,k_2,r +1- \del_i+v) |k_2 \in \cd_{k_1}[\vec{p}'|h] \land r+1 \in \{\del_i+1 \ldots \del_i + \pd-v-1\} \}\\
   &\\
   =& \{(k_1,k_2,r +1- \del_i+v) |k_2 \in \cd_{k_1}[\vec{p}'|h] \land r+2 \in \{\del_i+2 \ldots \del_i + \pd-v \}\}
 \end{array}
\end{equation}

By~\ref{eq:priv-corr1},~\ref{eq:priv-corr1-2},~\ref{eq:priv-corr1-7},~\ref{eq:priv-corr1-8},~\ref{eq:priv-corr1-9}, and~\ref{eq:priv-corr1-10},
$$
\begin{array}{ll}
\ds_i[j|h_{i,r+2}'] & = L_1[r+2|\vec{\sigma}^*] \cup\\
			 & \{(k_1,k_2,r+1 - \del_i+v) | k_2 \in \cd_{k_1}[\vec{p}'|h] \land r+2 \in \{\del_i+2 \ldots \del_i + \pd]-v \} \cup\\
			 & L_2[r+2|\vec{\sigma}^*] \cup  \{(k_1,k_2,v) |  k_2 \in \cd_{k_1}[\vec{p}'|h] \land \del_i = r+1\}\\
                             &\\
                             &= L_1[r+2|\vec{\sigma}^*] \cup L_2[r+2|\vec{\sigma}^*] \cup \\
                              & \{(k_1,k_2,r+1 - \del_i+v) | k_2 \in \cd_{k_1}[\vec{p}'|h] \land r+2 \in \{\del_i+2 \ldots \del_i + \pd-v \}\cup\\                              
                              & \{(k_1,k_2,r +1 - \del_i + v) |  k_2 \in \cd_{k_1}[\vec{p}'|h] \land \del_i+1 = r+2\}\\
                              &\\
                              &= \ds_i[j|h_{i,r+2}^*] \cup\\
 			&  \{(k_1,k_2,r+1 - \del_i+v) | k_2 \in \cd_{k_1}[\vec{p}'|h] \land r+2 \in \{\del_i+1 \ldots \del_i + \pd-v\}.
\end{array}
$$
This proves the induction step and concludes the proof.
\end{proof}

\newpage
\subsubsection{Proof of Lemma~\ref{lemma:priv-history}.}
\label{proof:lemma:priv-history}
For every $i \in {\cal N}$, $j \in {\cal N}_i$, $h \in {\cal H}$, and $h_i,h_j \in h$:
$$p_i[j|h_i] = p_i[j|h_j].$$

\begin{proof}
Fix $i$, $j$, $h$, and $h_i,h_j \in h$.
Consider any tuple $(k_1,k_2,r) \in \ds_i[j|h_i]$ and let $K_i$ and $K_j$ represent
the sets used by $i$ and $j$ to compute $p_i[j|h_i]$ and $p_i[j|h_j]$, respectively.
By Lemma~\ref{lemma:priv-corr-1} and Definitions~\ref{def:privsig} and~\ref{def:priv-thr}, 
it is true that for some $r' \geq d_i[k_1,k_2]$:
\begin{equation}
\label{eq:pvh-0}
r=r' - \del_i[k_1,k_2]+v_i,
\end{equation}
$|h_i| \geq r'+1$ and, for $s_i = h_i^{r'+1}$, 
\begin{equation}
\label{eq:pvh-1}
s_i[k_1,k_2] = \mbox{\emph{defect}},
\end{equation}
where $v_i = \min[\del_i[k_1,k_2] - \del_j[k_1,k_2],0]$.

By Definition~\ref{def:privsig}, this implies that 
$$|h_i|\geq r' + \del_i[k_1,k_2]+1,$$ 
and for $h_{k_2} \in h$ and $s_{k_2} = h_{k_2}^{r' + \del_i[k_1,k_2]+1}$:
\begin{equation}
\label{eq:pvh-2}
s_{k_2}[k_1,k_2] = \mbox{\emph{defect}}.
\end{equation}

Since $r' \geq \del_i[k_1,k_2]$, if $r <0$, then we have  by~\ref{eq:pvh-0}
$$v_i < 0 \Rightarrow \del_i[k_1,k_2] < \del_j[k_1,k_2],$$
which implies by Definition~\ref{def:privsig}, 
that $j$ has yet to observe the defection that caused $i$ to add $(k_1,k_2,v_i)$ to $\ds_i[j|h_i]$.
Consequently, $j$ has not included this tuple in $\ds_j[i|h_j]$ or in $K_j$. Also, by Definition~\ref{def:priv-thr},
$i$ does not include the tuple in $K_i$, since $r < 0$. 

Consider, now, that $r \geq0$, where by Definition~\ref{def:priv-thr} $i$ adds the tuple to $K_i$. By~\ref{eq:pvh-0},
\begin{equation}
\label{eq:pvh-3}
r' \geq \del_i[k_1,k_2] - v_i \geq \del_i[k_1,k_2] - \del_i[k_1,k_2] + \del_j[k_1,k_2] = \del_j[k_1,k_2].
\end{equation}

Furthermore, since by Definition~\ref{def:priv-thr}, $r < \pd[k_1,k_2|i,j]$, we also have by~\ref{eq:pvh-0}
\begin{equation}
\label{eq:pvh-4}
r' < \pd[k_1,k_2|i,j] + \del_i[k_1,k_2] - v_i.
\end{equation}
If $\del_i[k_1,k_2] \leq \del_j[k_1,k_2]$, then $v_i <0$, $v_j = 0$, and by~\ref{eq:pvh-3} and~\ref{eq:pvh-4} we have
$$r' < \pd[k_1,k_2|i,j] + \del_j[k_1,k_2] - v_j,$$
$$r' +1\in \{\del_j[k_1,k_2]+1 \ldots \del_j[k_1,k_2] + \pd[k_1,k_2|i,j] -v_j\},$$
where $v_j = \min[\del_j[k_1,k_2] - \del_i[k_1,k_2],0]$.
By Lemma~\ref{lemma:priv-corr-1}, $j$ adds $(k_1,k_2,r'-\del_j[k_1,k_2] + v_j)$ to $\ds_j[i|h_j]$,
such that $v_j=0$ and by~\ref{eq:pvh-3}:
\begin{equation}
\label{eq:pvh-5}
r' -\del_j[k_1,k_2] +v_j \geq 0.
\end{equation}

If $\del_i[k_1,k_2] > \del_j[k_1,k_2]$, then $v_i = 0$,
$$v_j = -(\del_i[k_1,k_2] - \del_j[k_1,k_2]),$$
hence, by~\ref{eq:pvh-4}
$$\del_j[k_1,k_2] + \pd[k_1,k_2|i,j] -v_j = \pd[k_1,k_2|i,j] + \del_i[k_1,k_2] - v_i> r'.$$
Therefore, by~\ref{eq:pvh-3},
$$r' +1\in \{\del_j[k_1,k_2]+1 \ldots \del_j[k_1,k_2] + \pd[k_1,k_2|i,j] -v_j\},$$
and by Lemma~\ref{lemma:priv-corr-1} $j$ adds $(k_1,k_2,r'-\del_j[k_1,k_2] + v_j)$ to $\ds_j[i|h_j]$.
Again, by~\ref{eq:pvh-0} and the assumption that $r \geq 0$,
\begin{equation}
\label{eq:pvh-6}
r'-\del_j[k_1,k_2] + v_j = r' - \del_i[k_1,k_2]  = r + \del_i[k_1,k_2] - \del_i[k_1,k_2]-v_i =r \geq 0.
\end{equation}
In any case, by~\ref{eq:pvh-5} and~\ref{eq:pvh-6}, $j$ adds the tuple to $K_j$.
This proves that $i$ adds the tuple to $K_i$ iff $j$ adds the tuple to $K_j$,
implying that $K_i = K_j$. Since $p_i[j|h_i]$ and $p_i[j|h_j]$ are obtained by applying the same 
deterministic functions to $K_i$ and $K_j$, respectively, we have
$$p_i[j|h_i] = p_i[j|h_j].$$
\end{proof}

\newpage

\subsection{Generic Results}
\label{sec:proof:priv-drop}

\subsubsection{Proof of Proposition~\ref{prop:priv-folk}.}
\label{proof:prop:priv-folk}
For every assessment $(\vec{\sigma}^*,\vec{\mu}^*)$, if $(\vec{\sigma}^*,\vec{\mu}^*)$ is Sequentially Rational,
then, for every $i \in {\cal N}$, $\frac{\beta_i}{\gamma_i} \geq \bar{p}_i$. 
Consequently, $\psi[\vec{\sigma}^*] \subseteq (v,\infty)$,
where $v = \max_{i \in {\cal N}} \bar{p}_i$.

\begin{proof}
Let $\vec{p}^* = \vec{\sigma}^*[\emptyset]$.
For $h_i = \emptyset$, the only history $h$ that fulfills $\mu_i^*[h|h_i]>0$ is $h = \emptyset$.
Therefore, the equilibrium utility is also
$$\pi_i[\vec{\sigma}^*|\vec{\mu}^*,\emptyset] = \sum_{r=0}^\infty\omega_i^r (1 - q_i[\vec{p}^*])(\beta_i - \gamma_i \bar{p}_i) = \frac{1 - q_i[\vec{p}^*]}{1-\omega_i}(\beta_i - \gamma_i \bar{p}_i).$$
If $\frac{\beta_i}{\gamma_i} \leq \bar{p}_i$, then 
\begin{equation}
\label{eq:priv-od-1}
\pi_i[\vec{\sigma}^*|\vec{\mu}^*,\emptyset]  \leq 0.
\end{equation}

Let $\sigma_i' \in \Sigma_i$ be a strategy such that, for every $h_i \in {\cal H}_i$, $\sigma_i'[h_i] = \vec{0}$,
and let $\vec{\sigma}' = (\sigma_i',\vec{\sigma}_{-i}^*)$, where $\vec{0}=(0)_{j \in {\cal N}_i}$. We have 
\begin{equation}
\label{eq:priv-od-2}
\pi_i[\vec{\sigma}'|\vec{\mu}^*,\emptyset] = (1-q_i[\vec{p}^*])\beta_i + \pi_i[\vec{\sigma}'|(\sig[\vec{p}'|\emptyset])] \geq (1-q_i[\vec{p}^*])\beta_i,
\end{equation}
where $\vec{p}' = (\vec{0},\vec{p}^*_{-i})$. By Lemma~\ref{lemma:pprob},
$q_i [\vec{p}^*] < 1$. Since $ \pi_i[\vec{\sigma}'|\sig[\vec{p}'|\emptyset]] \geq 0$, it is true that
$$\pi_i[\vec{\sigma}^*|\vec{\mu}^*,\emptyset] \leq 0 < \pi_i[\vec{\sigma}'|\vec{\mu}^*,\emptyset].$$
This contradicts the assumption that $\vec{\sigma}^*$ is a SPE.
\end{proof}

\newpage
\subsubsection{Proof of Lemma~\ref{lemma:priv-best-response1}.}
\label{proof:lemma:priv-best-response1}
For every $i \in {\cal N}$, $h_i \in {\cal H}_i$, $a_i \in BR[\vec{\sigma}_{-i}^*|\vec{\mu}^*,h_i]$, and $\vec{p}_i \in {\cal P}_i$ such that $a_i[\vec{p}_i] > 0$,
it is true that for every $j \in {\cal N}_i$ we have $p_i[j] \in \{0,p_i[j|h_i]\}$.

\begin{proof}
Suppose then that there exist $i \in {\cal N}$, $h_i \in {\cal H}_i$, $a_i^1 \in BR[\vec{\sigma}_{-i}^*|\vec{\mu}^*,h_i]$, 
and $\vec{p}^1_i \in {\cal P}_i$ such that $a_i^1[\vec{p}^1_i]>0$ and there exists $j \in {\cal N}_i$ such that $p_i^1[j] \notin \{0,p_i[j|h_i]\}$.
Fix any $h$ such that $h_i \in h$ and define an alternative $a_i^2 \in {\cal A}_i$:

\begin{itemize}
  \item  $\vec{a}^1 = (a_i^1,\vec{p}_{-i}^*)$ and $\vec{a}^2 = (a_i^2,\vec{p}_{-i}^*)$, where $\vec{p}^* = \vec{\sigma}^*[h]$.
  \item Define $\vec{p}_i^2 \in {\cal P}_i$ such that for every $j \in {\cal N}_i$, if $p_i^1[j] \geq p_i[j|h_i]$, then $p_i^2[j] = p_i[j|h_i]$, else, $p_i^2[j]=0$.
  \item Set $a_i^2[\vec{p}_i^2] = a_i^1[\vec{p}_i^1] + a_i^1[\vec{p}_i^2]$ and $a_i^2[\vec{p}_i^1] = 0$.
  \item For every $\vec{p}_i'' \in {\cal P}_i \setminus \{\vec{p}_i^1,\vec{p}_i^2\}$, set $a_i^2[\vec{p}_i''] = a_i^1[\vec{p}_i'']$.
  \item Define $\sigma_i^1 = \sigma_i^*[h_i|\vec{p}_i^1]$ and $\sigma_i^2 = \sigma_i^*[h_i|\vec{p}_i^2]$.
  \item Set $\vec{\sigma}^1 = (\sigma_i^1,\vec{\sigma}_{-i}^*)$ and $\vec{\sigma}^2 = (\sigma_i^2,\vec{\sigma}_{-i}^*)$.
\end{itemize}

Notice that for any $j \in {\cal N}_i$, $p_i^1[j] \geq p_i^2[j]$ and $p_i^1[j] \geq p_i[j|h_i]$ iff $p_i^2[j] \geq p_i[j|h_i]$.
Thus, by Definition~\ref{def:pubsig}, for every $s \in {\cal S}$,
\begin{equation}
\label{eq:pvbr-1}
\begin{array}{l}
pr_i[s|a_i^1,h] = pr_i[s|a_i^2,h].\\
pr[s|\vec{a}^1,h] = pr[s|\vec{a}^2,h].
\end{array}
\end{equation}

Moreover, for some $j \in {\cal N}_i$, $p_i^1[j|h_i] > p_i^2[j|h_i]$, thus, it is true that
\begin{equation}
\label{eq:pvbr-2}
u_i[\vec{a}^1] < u_i[\vec{a}^2].
\end{equation}
Recall that
$$\pi_i[\vec{\sigma}^1|h] = u_i[\vec{a}^1] + \omega_i\sum_{s \in {\cal S}} \pi_i[\vec{\sigma}^1|(h,s)]pr[s|\vec{a}^1,h],$$
$$\pi_i[\vec{\sigma}^2|h] = u_i[\vec{a}^2] + \omega_i\sum_{s \in {\cal S}} \pi_i[\vec{\sigma}^2|(h,s)]pr[s|\vec{a}^2,h].$$

By~\ref{eq:pvbr-1} and the definition of $\vec{\sigma}^1$ and $\vec{\sigma}^2$,
$$\sum_{s \in {\cal S}} \pi_i[\vec{\sigma}|(h,s)]pr[s|\vec{a}^1,h] = \sum_{s \in {\cal S}} \pi_i[\vec{\sigma}|(h,s)]pr[s|\vec{a}^2,h].$$

It follows from~\ref{eq:pvbr-2} that, for every $h \in {\cal H}$ such that $h_i \in h$,
$$\pi_i[\vec{\sigma}^1|h] < \pi_i[\vec{\sigma}^2|h].$$
Consequently,
$$\pi_i[\vec{\sigma}^1|h_i] < \pi_i[\vec{\sigma}^2|h_i].$$

This is a contradiction, since $a_i^1 \in BR[\vec{\sigma}_{-i}^*|h_i]$ by assumption, concluding the proof.
\end{proof}
\newpage

\subsubsection{Proof of Lemma~\ref{lemma:priv-best-response2}.}
\label{proof:lemma:priv-best-response2}
For every $i \in {\cal N}$ and $h_i \in {\cal H}_i$, there exists $a_i \in BR[\vec{\sigma}_{-i}^*|\vec{\mu}^*,h_i]$ and $\vec{p}_i \in {\cal P}_i$
such that $a_i[\vec{p}_i] = 1$.

\begin{proof}
Fix $i$ and $h_i$. If $BR[\vec{\sigma}_{-i}^*|\vec{\mu}^*,h_i]$ only contains pure strategies for the stage game,
since $BR[\vec{\sigma}_{-i}^*|\vec{\mu}^*,h_i]$ is not empty, the result follows.
Suppose then that there exists a mixed strategy $a_i^1 \in BR[\vec{\sigma}_{-i}^*|\vec{\mu}^*,h_i]$. 
We know from Lemma~\ref{lemma:priv-best-response1} that every such $a_i$ attributes positive probability to one of two probabilities in $\{0,p_i[j|h_i]\}$, for every $j \in {\cal N}_i$.
Denote by ${\cal P}_i^*[h_i]$ the finite set of profiles of probabilities that fulfill the condition of Lemma~\ref{lemma:best-response1}, i.e., for every $\vec{p}_i \in {\cal P}_i^*[h_i]$ and $j \in {\cal N}_i$,
$p_i[j] \in \{0,p_i[j|h_i]\}$. Define ${\cal P}^*[h]$ similarly for any $h \in {\cal H}$.

For any $h \in {\cal H}$ such that $h_i \in h$, we can write
\begin{equation}
\label{eq:priv-br2-0}
\pi_i[\vec{\sigma}^1|h]  = \sum_{\vec{p}_i \in {\cal P}^*_i[h_i]} (u_i[\vec{p}] + \omega_i \pi_i[\vec{\sigma}^1|(h,\sig[\vec{p}|h])])a_i^1[\vec{p}_i],
\end{equation}
where $\vec{p} = (\vec{p}_i,\vec{p}_{-i}^*)$ and $\vec{p}^* = \vec{\sigma}^*[h]$.

For any $\vec{p}_i^1 \in {\cal P}_i^*[h_i]$ such that $a_i[\vec{p}_i^1] >0$, let
$\sigma_i^1 = \sigma_i^*[h_i|a_i^1]$, $\vec{\sigma}^1 = (\sigma_i^1,\vec{\sigma}^*_{-i})$, $\sigma_i' = \sigma_i^*[h_i|\vec{p}_i^1]$,
and $\vec{\sigma}' = (\sigma_i',\vec{\sigma}^*_{-i})$.

There are three possibilities:
\begin{enumerate}
 \item $\pi_i[\vec{\sigma}^1|\vec{\mu}^*,h_i] = \pi_i[\vec{\sigma}'|\vec{\mu}^*,h_i]$.
 \item $\pi_i[\vec{\sigma}^1|\vec{\mu}^*,h_i] < \pi_i[\vec{\sigma}'|\vec{\mu}^*,h_i]$.
 \item $\pi_i[\vec{\sigma}^1|\vec{\mu}^*,h_i] > \pi_i[\vec{\sigma}'|\vec{\mu}^*,h_i]$.
\end{enumerate}

In possibility~1, it is true that there exists $a_i' \in BR[\vec{\sigma}_{-i}^*|\vec{\mu}^*,h_i]$ such that $a_i'[\vec{p}_i^1] = 1$ and the result follows. Possibility~2 contradicts
the assumption that $a_i^1 \in BR[\vec{\sigma}_{-i}^*|\vec{\mu}^*,h_i]$. 

Finally, consider that possibility~3 is true. Recall that $a_i^1$ being mixed implies $a_i^1[\vec{p}_i^1] < 1$.
Thus, there must exist $\vec{p}_i^2 \in {\cal P}_i^*[h_i]$, $\sigma_i''=\sigma_i^*[h_i|\vec{p}_i^2]$, and $\vec{\sigma}'' = (\sigma_i',\vec{\sigma}^*_{-i})$,
such that $a_i^1[\vec{p}_i^2] > 0$ and
\begin{equation}
\label{eq:priv-br2-1}
\pi_i[\vec{\sigma}'|\vec{\mu}^*,h_i] < \pi_i[\vec{\sigma}''|\vec{\mu}^*,h_i].
\end{equation}

Here, we can define $a_i^2 \in {\cal A}_i$ such that:
\begin{itemize}
  \item $a_i^2[\vec{p}_i^2]=a_i^1[\vec{p}_i^1] + a_i^1[\vec{p}_i^2]$; 
  \item $a_i^2[\vec{p}_i^1] = 0$.
  \item For every $\vec{p}_i'' \in {\cal P}_i^*[h_i] \setminus \{\vec{p}_i^1,\vec{p}_i^2\}$, $a_i^2[\vec{p}_i''] = a_i^1[\vec{p}_i'']$.  
\end{itemize}

Now, let $\sigma_i^2 = \sigma_i^*[h_i|a_i^2]$, and $\vec{\sigma}^2 = (\sigma_i^2,\vec{\sigma}^*_{-i})$.

By~\ref{eq:priv-br2-0}, it holds that for every $h \in {\cal H}$ such that $h_i \in h$:
$$\pi_i[\vec{\sigma}^1|h] = l_1 + \pi_i[\vec{\sigma}'|h]a_i^1[\vec{p}_i^1] + \pi_i[\vec{\sigma}''|h]a_i^1[\vec{p}_i^2],$$
$$\pi_i[\vec{\sigma}^2|h] = l_2 + \pi_i[\vec{\sigma}''|h]a_i^2[\vec{p}_i^2],$$
where 
$$l_1 = \sum_{\vec{p}'' \in {\cal P}^*[h] \setminus \{\vec{p}_i^1,\vec{p}_i^2\}}  (u_i[\vec{p}''] + \omega_i \pi_i[\vec{\sigma}^1|h,s''])a_i^1[\vec{p}_i''],$$
$$l_2 = \sum_{\vec{p}'' \in {\cal P}^*[h] \setminus \{\vec{p}_i^1,\vec{p}_i^2\}}  (u_i[\vec{p}''] + \omega_i \pi_i[\vec{\sigma}^2|h,s''])a_i^2[\vec{p}_i''],$$
and $s'' = \sig[\vec{p}''|h]$.

By the definition of $a_i^2$, we have that $l_1 = l_2$. 
It follows that
$$
\begin{array}{ll}
\pi_i[\vec{\sigma}^1|h] - \pi_i[\vec{\sigma}^2|h] & = \pi_i[\vec{\sigma}'|h]a_i^1[\vec{p}_i^1]  + \pi_i[\vec{\sigma}''|h]a_i^1[\vec{p}_i^2] - \pi_i[\vec{\sigma}''|h](a_i^1[\vec{p}_i^2] + a_i^1[\vec{p}_i^1]) \\
                                                                                   & = (\pi_i[\vec{\sigma}'|h] - \pi_i[\vec{\sigma}''|h])a_i^1[\vec{p}_i^1].
\end{array}
$$

Consequently, by~\ref{eq:priv-br2-1},
$$
\begin{array}{ll}
\pi_i[\vec{\sigma}^1|\vec{\mu}^*,h_i] - \pi_i[\vec{\sigma}^2|\vec{\mu}^*,h_i] &= \sum_{h \in {\cal H}}\mu_i^*[h|h_i](\pi_i[\vec{\sigma}^1|h] - \pi_i[\vec{\sigma}^2|h])\\
														  &=\sum_{h \in {\cal H}}\mu_i^*[h|h_i](\pi_i[\vec{\sigma}'|h] - \pi_i[\vec{\sigma}''|h])a_i^1[\vec{p}_i^1]\\
														  &=(\pi_i[\vec{\sigma}'|\vec{\mu}^*,h_i] - \pi_i[\vec{\sigma}''|\vec{\mu}^*,h_i])a_i^1[\vec{p}_i^1]\\
														  &<0.
\end{array}
$$
Thus,
$$\pi_i[\vec{\sigma}^1|\vec{\mu}^*,h_i] < \pi_i[\vec{\sigma}^2|\vec{\mu}^*,h_i],$$

contradicting the assumption that $a_i^1 \in BR[\vec{\sigma}_{-i}^*|\vec{\mu}^*,h_i]$. This concludes the proof.
\end{proof}

\newpage
\subsubsection{Proof of Lemma~\ref{lemma:priv-best-response}.}
\label{proof:lemma:priv-best-response}
For every $i \in {\cal N}$ and $h_i \in {\cal H}_i$, there exists $\vec{p}_i \in {\cal P}_i$ and a pure strategy $\sigma_i=\sigma_i^*[h_i|\vec{p}_i]$ such that:
\begin{enumerate}
 \item For every $j \in {\cal N}_i$, $p_i[j] \in \{0,p_{i}[j|h_i]\}$.
 \item For every $a_i \in {\cal A}_i$, $\pi_i[\sigma_i,\vec{\sigma}_{-i}^*|\vec{\mu}^*,h_i] \geq \pi_i[\sigma_i',\vec{\sigma}_{-i}^*|\vec{\mu}^*,h_i]$,
where $\sigma_i' = \sigma_i^*[h_i|a_i]$.
\end{enumerate}

\begin{proof}
Consider any $i \in {\cal N}$ and $h_i \in {\cal H}$.
From Lemma~\ref{lemma:priv-best-response2}, it follows that there exists $a_i \in BR[\vec{\sigma}_{-i}^*|\vec{\mu}^*,h_i]$
and $\vec{p}_i \in {\cal P}_i$ such that $a_i[\vec{p}_i] = 1$. By Lemma~\ref{lemma:priv-best-response1},
every such $a_i$ and $\vec{p}_i$ such that $a_i[\vec{p}_i]=1$ fulfill Condition~$1$.
Condition~$2$ follows from the definition of $BR[\vec{\sigma}_{-i}^*|\vec{\mu}^*,h_i]$.
\end{proof}

\subsubsection{Proof of Lemma~\ref{lemma:priv-drop-suff}.}
\label{proof:lemma:priv-drop-suff}
If the PDC Condition is fulfilled and $(\vec{\sigma}^*,\vec{\mu}^*)$ is Preconsistent, then $(\vec{\sigma}^*,\vec{\mu}^*)$ is Sequentially Rational.

\begin{proof}
Assume that Inequality~\ref{eq:priv-drop} holds for every player $i$, history $h_i$ and
$D \subseteq {\cal N}_{i}[h_i]$. In particular, these assumptions imply that,
for each $\vec{p}_i \in {\cal P}_i$ such that $p_i[j] \in \{0,p_i[j|h_i]\}$ for every $j \in {\cal N}_i$,
we have 
\begin{equation}
\label{eq:priv-gcs}
\pi_i[\vec{\sigma}^*|h_i] \geq \pi_i[\sigma_i,\vec{\sigma}^*_{-i}|h_i],
\end{equation}
where $\sigma_i = \sigma_i^*[h_i|\vec{p}_i]$. By Lemma~\ref{lemma:priv-best-response}, 
there exists one such $\vec{p}_i$ such that $\sigma_i$ is a local best response.
Consequently, by~\ref{eq:priv-gcs}, for every $a_i \in {\cal A}_i$ and $\sigma_i'=\sigma_i^*[h|a_i]$,
$$\pi_i[\vec{\sigma}^*|h_i] \geq \pi_i[\sigma_i',\vec{\sigma}^*_{-i}|h_i].$$
By Property~\ref{prop:priv-one-dev}, $(\vec{\sigma}^*,\vec{\mu}^*)$ is Sequentially Rational.
\end{proof}

\newpage

\subsubsection{Proof of Lemma~\ref{lemma:paths}.}
\label{proof:lemma:paths}
If the assessment $(\vec{\sigma}^*,\vec{\mu}^*)$ is Preconsistent and 
Sequentially Rational and $G$ is redundant, then for every $i \in {\cal N}$ and $j \in {\cal N}_i$,
there exists $k \in {\cal N} \setminus \{i\}$, $x \in \pth[s,i]$, and $x' \in \pth[j,k]$, such that $k \in x$ and $i \notin x'$.

\begin{proof}
Suppose that there exists a player $i \in {\cal N}$, and a neighbor $j \in {\cal N}_i$ 
such that for every $k \in {\cal N} \setminus \{i\}$, $x' \in \pth[j,k]$, 
and $x \in \pth[s,i]$, we have $k \notin x$ or $i \in x'$. Define $D_j \subseteq {\cal N}$ and $D \subseteq {\cal N}_i$ as:
\begin{equation}
\label{eq:pths1}
\begin{array}{l}
D_j = \{k \in {\cal N} \setminus \{i\} | \exists_{x \in \pth[j,k]} i \notin x\}.\\
D = \{j \in {\cal N}_i | \forall_{k \in D_j}\forall_{x \in \pth[s,i]} \forall_{x' \in \pth[j,k]}  k \notin x \vee i \in x'\}.
\end{array}
\end{equation}

Let $\rs_D = \cup_{j \in D} D_j$. By our assumptions, $D$ is not empty. Define $\sigma_i' = \sigma_i^*[h_i|\vec{p}_i']$ for every $h_i \in {\cal H}$ such that:
\begin{itemize}
  \item For every $j \in D$, $p_i'[j] = 0$.
  \item For every $j \in {\cal N}_i \setminus D$, $p_i'[j] = p_i[j|\emptyset]$.
\end{itemize}
Let $\vec{\sigma}' = (\sigma_i',\vec{\sigma}_{-i}^*)$.

Notice that, for every $k \in {\cal N} \setminus \{i\}$,
\begin{equation}
\label{eq:pths4}
\begin{array}{ll}
\cd_i[\vec{p}'|h] = D.\\
\cd_k[\vec{p}'|h] = \emptyset.
\end{array}
\end{equation}

For every $j \in D$ and $k \in {\cal N} \setminus (\rs_D \cup \{i\})$, we have that $\del_k[i,j]= \infty$. Therefore, by Lemma~\ref{lemma:priv-corr-1} and by~\ref{eq:pths4},
for every $l \in {\cal N}_k$, $r\geq0$,
\begin{equation}
\label{eq:pths2}
\ds_k[l|h_{k,r}'] = \ds_k[l|h_{k,r}^*],
\end{equation}
where $h_{k,r}' \in \hevol[\emptyset,r|\vec{\sigma}']$ and $h_{k,r}^* \in \hevol[\emptyset,r|\vec{\sigma}']$.

By Definition~\ref{def:priv-thr} and~\ref{eq:pths2}, we have that for every $r \geq 0$,
\begin{equation}
\label{eq:pths5}
\begin{array}{ll}
p_k[l|h_{k,r}'] &= p_k[l|h_{l,r}^*].
\end{array}
\end{equation}

Consequently, by~\ref{eq:pths1} and \ref{eq:pths5},
for every $r >0$, there exist $\vec{p}^*= \vec{\sigma}^*[\hevol[\emptyset,r-1|\vec{\sigma}^*]]$ and $\vec{p}' = \vec{\sigma}'[\hevol[\emptyset,r-1|\vec{\sigma}']]$
such that for every $x \in \pth[s,i]$ and $k \in x \setminus \{i\}$ we have $\vec{p}_k^* = \vec{p}_k'$.

It follows from Lemma~\ref{lemma:bottleneck-impact} of Appendix~\ref{sec:epidemic} that
for every $r \geq 0$
\begin{equation}
\label{eq:ph-2}
q_i[\emptyset,r|\vec{\sigma}^*] = q_i[\emptyset,r|\vec{\sigma}'].
\end{equation}

By the definition of $\vec{p}_i'$, for every $r \geq 0$,
\begin{equation}
\label{eq:pths8}
\bar{p}_i[\emptyset,r|\vec{\sigma}'] < \bar{p}_i[\emptyset,r|\vec{\sigma}^*].
\end{equation}
By the definition of $u_i[h,r|\vec{\sigma}]$, from~\ref{eq:ph-2} and~\ref{eq:pths8},
we have for every $r\geq 0$
$$u_i[\emptyset,r|\vec{\sigma}^*] < u_i[\emptyset,r|\vec{\sigma}'].$$ 
This implies that
$$\sum_{r=0}^{\infty}\omega_i^r (u_i[\emptyset,r|\vec{\sigma}^*] - u_i[\emptyset,r|\vec{\sigma}']) < 0.$$
Since $h=\emptyset$ is the only history such that $\mu_i[h|\emptyset]=1$,
by Theorem~\ref{theorem:priv-drop}, $(\vec{\sigma}^*,\vec{\mu}^*)$ cannot be Sequentially Rational, which is a contradiction.
\end{proof}

\newpage
\subsection{Redundancy may Reduce Effectiveness}
\label{sec:proof:coord} 

\subsubsection{Proof of Theorem~\ref{theorem:problem}.}
\label{proof:theorem:problem}
If $G$ is redundant and there exist $i \in {\cal N}$, $j \in {\cal N}_i$, and $k \in {\cal N}_i^{-1}$ such that for every $x \in \pth[j,k]$ we have $i \in x$,
then Equality~\ref{eq:problem} holds.

\begin{proof}
Assume that there exist $i \in {\cal N}$, $j \in {\cal N}_i$, and $k \in {\cal N}_i^{-1}$ such that for every $x \in \pth[j,k]$ we have $i \in x$.
This implies by Definition~\ref{def:privsig} that
\begin{equation}
\label{eq:noally}
\del_k[i,j] = \infty.
\end{equation}

Define $\sigma_i' = \sigma_i^*[h_i|\vec{p}_i']$ for every $h_i \in {\cal H}_i$, such that $p_i'[j] = 0$ and $p_i'[l] = p_i[l|h_i]$ for every $l \in {\cal N}_i \setminus \{j\}$
and let $\vec{\sigma}' = (\sigma_i',\vec{\sigma}_{-i}^*)$. 

Notice that
$$\bar{p}_i[\emptyset,r|\vec{\sigma}'] = \bar{p}_i[\emptyset,r|\vec{\sigma}^*] - p_i[j|\emptyset].$$
By Theorem~\ref{theorem:priv-drop}, if $(\vec{\sigma}^*,\vec{\mu}^*)$ is Sequentially Rational, 
then, for every $r \geq 0$ and the empty history:
\begin{equation}
\label{eq:prb-0}
\begin{array}{ll}
\sum_{r=0}^{\infty} \omega_i (u_i[\emptyset,r|\vec{\sigma}^*] - u_i[\emptyset,r|\vec{\sigma}']) & \geq 0\\
&\\
\sum_{r=0}^{\infty} \omega_i ((1-q_i[\emptyset,r|\vec{\sigma}^*])(\beta_i - \gamma_i \bar{p}_i[\emptyset,r|\vec{\sigma}^*]) -&\\
 (1-q_i[\emptyset,r|\vec{\sigma}'])(\beta_i - \gamma_i \bar{p}_i[\emptyset,r|\vec{\sigma}'])) & \geq 0\\
 &\\
\sum_{r=0}^{\infty} \omega_i ((q_i[\emptyset,r|\vec{\sigma}']-q_i[\emptyset,r|\vec{\sigma}^*])(\beta_i - \gamma_i \bar{p}_i[\emptyset,r|\vec{\sigma}^*]) \\
  -(1-q_i[\emptyset,r|\vec{\sigma}']) \gamma_i p_i[j|\emptyset]  & \geq 0.
\end{array}
\end{equation}

By Lemma~\ref{lemma:prob-1}, since $G$ is connected from $s$ and $q_i$ is continuous in [0,1], for every $r \geq 0$:
\begin{equation}
\label{eq:prb-1}
\lim_{\vec{\sigma}^* \to \vec{1}} q_i[\emptyset,r|\vec{\sigma}^*]= q_i[\vec{1}] = 0.
\end{equation}

Now, let $\vec{p}^* = \vec{\sigma}^*[\emptyset]$, $\vec{p}' = (\vec{p}_i',\vec{p}_{-i}^*)$, $h_{k,r}' \in \hevol[\emptyset,r|\vec{\sigma}']$, and $h_{k,r}^* \in \hevol[\emptyset,r|\vec{\sigma}^*]$. 
We have that for every $l \in {\cal N} \setminus \{i\}$:
\begin{equation}
\label{eq:prb-2}
\cd_l[\vec{p}'|\emptyset] = \emptyset,
\end{equation}
and
\begin{equation}
\label{eq:prb-3}
\cd_i[\vec{p}'|\emptyset] = \{j\}.
\end{equation}

It follows immediately by Definition~\ref{def:priv-thr} and Lemma~\ref{lemma:priv-corr-1} that, for every $k \in {\cal N}$ and $l \in {\cal N}_k$ such that $\del_k[i,j]=\infty$, and $r\geq0$,
\begin{equation}
\label{eq:prb-4}
\begin{array}{l}
\ds_k[l|h_{k,r}'] = \ds_k[l|h_{k,r}^*].\\
p_{k}[l|h_{k,r}'] = p_{k}[l|h_{k,r}^*].
\end{array}
\end{equation}

For any $r \geq 0$, let
$$\vec{p}'' = \lim_{\vec{\sigma}^* \to \vec{1}} \vec{\sigma}'[h_{r}'].$$

Since $G$ is redundant, there is a path $x \in \pth[s,k]$ such that $i \notin x$. Furthermore, by~\ref{eq:noally},
for every $l \in x \setminus \{i\}$, $d_l[i,j] = \infty$. By~\ref{eq:prb-4}, this implies that $p_l''[a] = 1$ for every $a \in {\cal N}_l$.

Therefore, by Lemma~\ref{lemma:prob-1} and, since $G$ is connected from $s$ and $q_i$ is continuous in $[0,1]$, for every $r \geq 0$,
\begin{equation}
\label{eq:prb-6}
\lim_{\vec{\sigma}^* \to \vec{1}} q_i[\emptyset,r|\vec{\sigma}']= q_i[\vec{p}'']= 0.
\end{equation}

By~\ref{eq:prb-0},~\ref{eq:prb-1}, and~\ref{eq:prb-6},
\begin{equation}
\label{eq:prb-7}
\begin{array}{ll}
\lim_{\vec{\sigma}^* \to \vec{1}} \sum_{r=0}^{\infty} \omega_i (u_i[\emptyset,r|\vec{\sigma}^*] - u_i[\emptyset,r|\vec{\sigma}'])) & =\\
&\\
\lim_{\vec{\sigma}^* \to \vec{1}} \sum_{r=0}^{\infty} \omega_i ((q_i[\emptyset,r|\vec{\sigma}']-q_i[\emptyset,r|\vec{\sigma}^*])(\beta_i - \gamma_i \bar{p}_i[\emptyset,r|\vec{\sigma}^*]) &\\
- (1-q_i[\emptyset,r|\vec{\sigma}']) \gamma_i p_i[j|\emptyset] ) &=\\
&\\
\lim_{\vec{\sigma}^* \to \vec{1}} - (1-q_i[\emptyset,r|\vec{\sigma}']) \gamma_i p_i[j|\emptyset] ) &=\\
&\\
-\gamma_i & < 0.
\end{array}
\end{equation}

Therefore, in the limit, the PDC Condition is never fulfilled for any values of $\beta_i$, $\gamma_i$, and $\omega_i \in (0,1)$,
which implies by Theorem~\ref{theorem:priv-drop} that $(\vec{\sigma}^*,\vec{\mu}^*)$ is not Sequentially Rational for an arbitrarily large reliability 
and
$$\lim_{\vec{\sigma}^* \to \vec{1}} \psi[\vec{\sigma}^*|\vec{\mu}^*] = \emptyset.$$
\end{proof}

\newpage

\subsection{Coordination is Desirable}

\subsubsection{Proof of Theorem~\ref{theorem:problem-coord}.}
\label{proof:theorem:problem-coord}
If the graph is redundant and $\vec{\sigma}^*$ does not enforce coordination, then there is
a definition of $\vec{\sigma}^*$ such that:
$$\lim_{\vec{\sigma}^* \to \vec{1}} \psi[\vec{\sigma}^*|\vec{\mu}^*] = \emptyset.$$
\begin{proof}
In the aforementioned circumstances, define $\vec{\sigma}^*$ such that
for every $i \in {\cal N}$, $j \in {\cal N}_i$, and $h_i \in {\cal H}_i$,
\begin{equation}
\label{eq:coo-1}
p_i[j|h_i]>0 \equiv p_i[j|h_i] = p_i[j|\emptyset].
\end{equation}

By the assumption that $\vec{\sigma}^*$ does not enforce coordination,
there exists $i \in {\cal N}$ and $j \in {\cal N}_i$ such that for every $r>0$ there is $k \in {\cal N}_i^{-1}$
for which
\begin{equation}
\label{eq:coo-2}
r \leq \del_k[i,j] \vee r \geq \del_k[i,j] + \pd[i,j|k,i] + 1.
\end{equation}

Fix $r$. Let $\sigma_i' = \sigma_i^*[\emptyset|\vec{p}_i']$ and $\vec{\sigma}' = (\sigma_i',\vec{\sigma}_{-i}^*)$, such that
\begin{itemize}
  \item $p_i'[j] = 0$.
  \item For every $k \in {\cal N}_i \setminus \{j\}$, $p_i'[k] = p_i[k|\emptyset]$.
\end{itemize}
For $\vec{p}' = \vec{\sigma}'[\emptyset]$, we have
\begin{equation}
\label{eq:coo-3}
\cd_i[\vec{p}'|\emptyset] = \{j\},
\end{equation}
and for every $k \in {\cal N} \setminus \{i\}$
\begin{equation}
\label{eq:coo-4}
\cd_k[\vec{p}'|\emptyset] = \emptyset.
\end{equation}

By Lemma~\ref{lemma:priv-corr-1} and Definition~\ref{def:priv-thr},
and by~\ref{eq:coo-3}, and~\ref{eq:coo-4}, we have for every $a \in {\cal N}$ and $b \in {\cal N}_a$, $h_{a,r}' \in \hevol[\emptyset,r|\vec{\sigma}']$,
and $h_{a,r}^* \in \hevol[\emptyset,r|\vec{\sigma}^*]$:
\begin{equation}
\label{eq:coo-5}
\begin{array}{ll}
\ds_a[b|h_{a,r}'] =& \ds_a[b|h_{a,r}^*] \cup \{(k_1,k_2,r-1 - \del_a[k_1,k_2]+v[k_1,k_2]) | k_1,k_2 \in {\cal N} \land \\
                            & k_2 \in \cd_{k_1}[\vec{p}'|\emptyset] \land r \in \{\del_a[k_1,k_2] +1 \ldots \del_a[k_1,k_2] + \pd[k_1,k_2|a,b]-v[k_1,k_2]\}\land\\
                            & v[k_1,k_2] = \min[\del_a[k_1,k_2] - \del_b[k_1,k_2],0]\}\\
                            &\\
                            &= \emptyset \cup \{(i,j,r-1 - \del_a[i,j]+v[i,j]) | v[i,j] = \min[\del_a[i,j] - \del_b[i,j],0]\}.
\end{array}
\end{equation}

For $\ds_k[i|h_{k,r}']$, we have $v[i,j] = 0$, since $\del_i[i,j] = 0$. Therefore, by~\ref{eq:coo-5},
\begin{equation}
\label{eq:coo-6}
p_a[b|h_{a,r}] = p_a[b|\emptyset].
\end{equation}

Also, for every $a \in {\cal N} \setminus ({\cal N}_i \cup {\cal N}_i^{-1} \cup\{k,i\})$ and $b \in {\cal N}_a$, by Definition~\ref{def:priv-thr} and the definition of $\vec{\sigma}^*$
in this context, $a$ does not react to a defection of $i$ from $j$, which implies that:
\begin{equation}
\label{eq:coo-7}
\begin{array}{l}
\ds_a[b|h_{a,r}'] =\{(i,j,r-1 - \del_a[i,j]+v[i,j]) | v[i,j] = \min[\del_a[i,j] - \del_b[i,j],0]\}.\\
p_a[b|h_{a,r}'] = p_a[b|h_{a,r}^*] = p_a[b|\emptyset].
\end{array}
\end{equation}

Since the graph is redundant, there exists $x \in \pth[s,k]$ such that, for every $a \in ({\cal N}_i \cup {\cal N}_i^{-1} \cup \{i\}) \setminus \{k\}$, 
$a \notin x$. Thus, by~\ref{eq:coo-6} and~\ref{eq:coo-7}, for 
$$\vec{p}' = \lim_{\vec{\sigma}^* \to \vec{1}} \vec{\sigma}'[\hevol[\emptyset,r|\vec{\sigma}']],$$ 
there exists a path $x \in \pth[s,i]$ such that for every $a \in x \setminus \{i\}$ and $b \in {\cal N}_a$ we have
$$p_a'[b] = 1.$$

By Lemma~\ref{lemma:prob-1}, given that $q_i$ is continuous in $[0,1]$, for every $r>0$,
\begin{equation}
\label{eq:coo-8}
\begin{array}{ll}
\lim_{\vec{\sigma}^* \to \vec{1}} (q_i[\emptyset,r|\vec{\sigma}^*] - q_i[\emptyset,r|\vec{\sigma}']) &= 0.\\
&\\
\lim_{\vec{\sigma}^* \to \vec{1}} (u_i[\emptyset,r|\vec{\sigma}^*] - u_i[\emptyset,r|\vec{\sigma}']) &=\\
\lim_{\vec{\sigma}^* \to \vec{1}} (1-q_i[\emptyset,r|\vec{\sigma}^*])(\beta_i - \gamma_i \bar{p}_i[\emptyset,r|\vec{\sigma}^*]) - (1-q_i[\emptyset,r|\vec{\sigma}'])(\beta_i - (\gamma_i \bar{p}_i[\emptyset,r|\vec{\sigma}^*] + p_i[j|h_{ir}^*]))  &=\\
\lim_{\vec{\sigma}^* \to \vec{1}} -\gamma_i p_i[j|\emptyset]  &<0.
\end{array}
\end{equation}

Therefore, in the limit, the PDC Condition is never fulfilled for any values of $\beta_i$ and $\gamma_i$,
which implies by Theorem~\ref{theorem:priv-drop} that $(\vec{\sigma}^*,\vec{\mu}^*)$ is never Sequentially Rational and:
$$\lim_{\vec{\sigma}^* \to \vec{1}} \psi[\vec{\sigma}^*|\vec{\mu}^*] = \emptyset.$$
\end{proof}

\newpage
\subsection{Impact of Delay}
\label{sec:proof:priv-indir}

\subsubsection{Auxiliary Lemmas.}
Lemma~\ref{lemma:aux:impdel-1} shows that the utility obtained by $i$ during the $\pd$ stages that follow stage $\mdel_i$
after any defection by $i$ is null. This is because during that period $i$ is necessarily punished by every in-neighbor.

\begin{lemma}
\label{lemma:aux:impdel-1}
For every $i \in {\cal N}$, $h \in {\cal H}$ and $h_i \in h$, $D \subseteq {\cal N}_i[h_i]$,
and $r \in \{\mdel_i +1 \ldots \mdel_i + \pd\}$, $u_i[h,r|\vec{\sigma}'] = 0$,
where $\vec{\sigma}' = (\sigma_i^*[h_i|\vec{p}_i'],\vec{\sigma}_{-i}^*)$
and $i$ drops every node from $D$ in $\vec{p}_i'$.
\end{lemma}
\begin{proof}
By Definition~\ref{def:overlap}, since $\del_i[i,j] = 0$ for every $i$,
then, for every $k \in {\cal N}_i^{-1}$, $\max[\del_k[i,j],\del_i[i,j]] = \del_k[i,j]$ and:
\begin{equation}
\label{eq:aid1-1}
\begin{array}{l}
\pd[i,j|k,i] \leq \mdel_i.\\
\pd[i,j|k,i] = \mdel_i - \del_k[i,j] + \pd.
\end{array}
\end{equation}

Notice that for $\vec{p}' = \vec{\sigma}'[h]$
\begin{equation}
\label{eq:aid1-2}
\cd_i[\vec{p}'|h] = D,
\end{equation}
and for every $j \in {\cal N} \setminus \{i\}$
\begin{equation}
\label{eq:aid1-3}
\cd_j[\vec{p}'|h] = \emptyset.
\end{equation}

By Lemma~\ref{lemma:priv-corr-1} and by~\ref{eq:aid1-1}, for every $k \in {\cal N}_i^{-1}$ and $r \in \{\mdel_i+1 \ldots \mdel_i + \pd\}$:
\begin{equation}
\label{eq:aid1-4}
\begin{array}{ll}
\ds_k[i|h_{k,r}'] =& \ds_k[i|h_{k,r}^*] \cup \{(k_1,k_2,r-1 - \del_k[k_1,k_2]+v[k_1,k_2]) | k_1,k_2 \in {\cal N} \land \\
                            & k_2 \in \cd_{k_1}[\vec{p}'|h] \land r \in \{\del_k[k_1,k_2] + 1\ldots \del_i[k_1,k_2] + \pd[k_1,k_2|k,i]-v[k_1,k_2]\}\land\\
                            & v[k_1,k_2] = \min[\del_k[k_1,k_2] - \del_i[k_1,k_2],0]\},\\
                            &\\
                            & =\ds_k[i|h_{k,r}^*] \cup \{(i,j,r-1 - \del_k[i,j]+v[i,j]) | j \in D \land  \\
                            & r \in\{\del_k[i,j]+1 \ldots \del_k[i,j] + \pd[i,j|k,i]-v[i,j]\}\land\\
                            & v[i,j] = \min[\del_k[i,j] - \del_i[i,j],0]\}\\
                            &\\
                            & =\ds_k[i|h_{k,r}^*] \cup \{(i,j,r-1 - \del_k[i,j]+v[i,j]) | j \in D \land \\
                            &r \in\{\del_k[i,j]+1 \ldots \del_k[i,j] + \pd[i,j|k,i] - v[i,j]\} \land v[i,j] = 0\}\\
                            &\\
                            &\supseteq \ds_k[i|h_{k,r}^*] \cup \{(i,j,r-1 - \del_k[i,j]) | j \in D \land  r \in\{\mdel_i+1 \ldots \mdel_i + \pd\}\}\\
                            &\\
                            & = \ds_k[i|h_{k,r}^*] \cup \{(i,j,r-1 - \del_k[i,j]+v[i,j])\}
\end{array}
\end{equation}
where $h_{k,r}^* \in \hevol[h,r|\vec{\sigma}^*]$ and $h_{k,r}' \in \hevol[h,r| \vec{\sigma}']$.

By~\ref{eq:aid1-4} and Definition~\ref{def:priv-thr}, it follows that, for every $r \in \{\mdel_i+1 \ldots \mdel_i + \pd\}$,
$$p_k[i|h_{k,r}'] = 0,$$
which by Lemma~\ref{lemma:noneib} leads to
\begin{equation}
\label{eq:aid1-5}
\begin{array}{l}
q_i[h,r|\vec{\sigma}'] = 1.\\
u_i[h,r|\vec{\sigma}'] = 0.
\end{array}
\end{equation}
This concludes the proof.
\end{proof}

\newpage

Lemma~\ref{lemma:aux:impdel-2} proves that all punishments of a node $i$ are concluded after stage $\mdel_i + \pd$
that follows any defection of $i$.
\begin{lemma}
\label{lemma:aux:impdel-2}
For every $i \in {\cal N}$, $h \in {\cal H}$ and $h_i \in h$, $D \subseteq {\cal N}_i[h_i]$,
and $r > \mdel_i + \pd$, 
$$u_i[h,r|\vec{\sigma}'] = u_i[h,r|\vec{\sigma}^*],$$
where $\vec{\sigma}' = (\sigma_i^*[h_i|\vec{p}_i'],\vec{\sigma}_{-i}^*)$
and $i$ drops every node from $D$ in $\vec{p}_i'$.
\end{lemma}
\begin{proof}
Notice that for $\vec{p}' = \vec{\sigma}'[h]$ and $j \in {\cal N} \setminus \{i\}$
\begin{equation}
\label{eq:aid2-1}
\begin{array}{l}
\cd_i[\vec{p}'|h] = D.\\
\cd_j[\vec{p}'|h] = \emptyset.
\end{array}
\end{equation}

By Lemma~\ref{lemma:priv-corr-1} and by~\ref{eq:aid2-1}, for every $k \in {\cal N}$, $l \in {\cal N}_k$, and $r > \mdel_i + \pd$,
\begin{equation}
\label{eq:aid2-2}
\begin{array}{ll}
\ds_k[l|h_{k,r}'] =& \ds_k[l|h_{k,r}^*] \cup \{(k_1,k_2,r-1 - \del_k[k_1,k_2]+v[k_1,k_2]) | k_1,k_2 \in {\cal N} \land \\
                            & k_2 \in \cd_{k_1}[\vec{p}'|h] \land r \in \{\del_k[k_1,k_2] + 1\ldots \del_k[k_1,k_2] + \pd[k_1,k_2|k,l]-v[k_1,k_2]\}\land\\
                            & v[k_1,k_2] = \min[\del_k[k_1,k_2] - \del_l[k_1,k_2],0]\},\\
                            &\\
                            & =\ds_k[l|h_{k,r}^*] \cup \{(i,j,r-1 - \del_k[i,j]+v[i,j]) | j \in D \land  \\
                            & r \in\{\del_k[i,j]+1 \ldots \del_k[i,j] + \pd[i,j|k,l]-v[i,j]\}\land\\
                            & v[i,j] = \min[\del_k[i,j] - \del_l[i,j],0]\}\\
                            &\\
                            & =\ds_k[l|h_{k,r}^*] \cup A\\
\end{array}
\end{equation}
where $h_{k,r}^* \in \hevol[h,r|\vec{\sigma}^*]$ and $h_{k,r}' \in \hevol[h,r| \vec{\sigma}']$.

The goal now is to show that for every $(i,j,r') \in A$, we have $r' < 0$. Fix any $j$.

By Definition~\ref{def:overlap}, for every $k \in {\cal N}$ and $l \in {\cal N}_k$ such that
\begin{equation}
\label{eq:aid2-4}
g=\max[\del_k[i,j],\del_l[i,j]] > \mdel + \pd, 
\end{equation}
we have
$$\pd[i,j|k,l] = 0.$$
Thus, by~\ref{eq:aid2-4}, if $g= \del_k[i,j]$, then we have $v[i,j] = 0$ and
$$
\begin{array}{l}
\del_k[i,j] + \pd[i,j|k,l] - v[i,j] =  \del_k[i,j].
\end{array}
$$
Thus, there is no $r$ such that
$$r \in\{\del_k[i,j]+1 \ldots \del_k[i,j] + \pd[i,j|k,l]-v[i,j]\},$$
which implies that $A = \emptyset$ and the result is true.

If $g = \del_l[i,j]$,  then we have $v[i,j] = \del_k[i,j] - \del_l[i,j]$ and
$$
\begin{array}{l}
\del_k[i,j] + \pd[i,j|k,l] - v[i,j] = \del_l[i,j].\\
r' = r - 1 - \del_k[i,j] + v[i,j] = r-1 - \del_l[i,j].
\end{array}
$$
Here, if $r'\geq0$, then $\del_l[i,j] \leq r-1$ and 
there is no $r > \mdel_i + \pd$ such that
$$r \in\{\del_k[i,j]+1 \ldots \del_k[i,j] + \pd[i,j|k,l]-v[i,j]\},$$
which implies that $A = \emptyset$. So, either $A= \emptyset$ or $r'<0$, which concludes
the step for any $k \in {\cal N}$ and $l \in {\cal N}_k$ that fulfill~$\ref{eq:aid2-4}$.

Consider now that 
\begin{equation}
\label{eq:aid2-5}
g=\max[\del_k[i,j],\del_l[i,j]] < \mdel_i + \pd,
\end{equation}
which results by Definition~\ref{def:overlap} in
$$\pd[i,j|k,l] = \mdel_i + \pd - g.$$
Notice that
$$\del_k[i,j] -  v[i,j] = g,$$
which implies that
$$\del_k[i,j] + \pd[i,j|k,l]-v[i,j] = g + \mdel_i + \pd - g = \mdel_i + \pd.$$
Thus, by~\ref{eq:aid2-4} there is no $r > \mdel_i + \pd$ such that
$$r \in\{\del_k[i,j]+1 \ldots \del_k[i,j] + \pd[i,j|k,l]-v[i,j]\},$$
which implies that $A = \emptyset$.

This allows us to conclude that for every $k \in {\cal N}$, $l \in {\cal N}_k$, and $r > \mdel_i + \pd$,
$$\ds_k[l|h_{k,r}']  = \ds_k[l|h_{k,r}^*] \cup A,$$
where for every $(i,j,r') \in A$ we have $r' < 0$.
By Definition~\ref{def:priv-thr}, $i$ adds $(k_1,k_2,r'') \in \ds_k[l|h_{k,r}^*]$ to $K$ if and only if $i$ adds $(k_1,k_2,r'')$ to $\ds_k[l|h_{k,r}']$.
Therefore,
\begin{equation}
\label{eq:aid2-6}
\begin{array}{l}
p_k[l|h_{k,r}'] = p_k[l|h_{k,r}^*].\\
q_i[h,r|\vec{\sigma}'] = q_i[h,r|\vec{\sigma}^*].\\
u_i[h,r|\vec{\sigma}'] = u_i[h,r|\vec{\sigma}^*].
\end{array}
\end{equation}
This concludes the proof.
\end{proof}

\newpage

\subsubsection{Proof of Lemma~\ref{lemma:delay-equiv}.}
\label{proof:lemma:delay-equiv}
If $(\vec{\sigma}^*,\mu^*)$ is Preconsistent, Assumption~\ref{def:non-neg} holds, and Inequality~\ref{eq:delay-equiv} 
is fulfilled for every $i \in {\cal N}$, $h_i \in {\cal H}_i$, and $h \in {\cal H}$ such 
that $\mu_i^*[h|h_i]>0$, then $(\vec{\sigma}^*,\vec{\mu}^*)$ is Sequentially Rational:
$$
- \sum_{r=0}^{\mdel_i} \omega_i^r ((1-q_i[h,r|\vec{\sigma}^*])\gamma_i \bar{p}_i[h,r|\vec{\sigma}^*]  + \epsilon \beta_i)+ \sum_{r=\mdel_i+1}^{\mdel_i + \pd} \omega_i^r u_i[h,r | \vec{\sigma}^*]  \geq 0.
$$

\begin{proof}
Fix $i$, $h_i$, and $h$. Define $\vec{\sigma}' = (\sigma_i^*[h_i|\vec{p}_i'],\vec{\sigma}_{-i}^*)$ for any $D \subseteq {\cal N}_i[h_i]$, where:
\begin{itemize}
  \item For every $j \in D$, $p_i'[j] = 0$.
  \item For every $j \in {\cal N}_i \setminus D$, $p_i'[j] = p_i[j|h_i]$.
\end{itemize}

By Assumption~\ref{def:non-neg}, for every $r \in \{0 \ldots \mdel_i\}$,
\begin{equation}
\label{eq:de-1}
\begin{array}{lll}
\bar{p}_i[h,r|\vec{\sigma}'] & \geq 0.\\
q_i[h,r|\vec{\sigma}^*] - q_i[h,r|\vec{\sigma}']  &\geq \epsilon.\\
u_i[h,r|\vec{\sigma}^*] - u_i[h,r|\vec{\sigma}'] & = (1-q_i[h,r|\vec{\sigma}^*])(\beta_i -\gamma_i \bar{p}_i[h,r|\vec{\sigma}^*]) -(1-q_i[h,r|\vec{\sigma}'])(\beta_i - \gamma_i \bar{p}_i[h,r|\vec{\sigma}'])\\
                                      				        & \geq -(1-q_i[h,r|\vec{\sigma}^*])\gamma_i - (q_i[h,r|\vec{\sigma}^*]- q_i[h,r|\vec{\sigma}'])\beta_i\\
					      			        & \geq -(1-q_i[h,r|\vec{\sigma}^*])\gamma_i - \epsilon \beta_i.
\end{array}
\end{equation}

By Lemma~\ref{lemma:aux:impdel-1}, for every $r \in \{\mdel_i+1 \ldots \mdel_i + \pd\}$,
\begin{equation}
\label{eq:de-2}
u_i[h,r|\vec{\sigma}'] = 0.
\end{equation}

Finally, by Lemma~\ref{lemma:aux:impdel-2}, for every $r \geq \mdel_i+\pd+1$,
\begin{equation}
\label{eq:de-3}
\begin{array}{l}
u_i[h,r|\vec{\sigma}^*] = u_i[h,r|\vec{\sigma}'].\\
\end{array}
\end{equation}

It follows from~\ref{eq:de-1},~\ref{eq:de-2}, and~\ref{eq:de-3} that:
$$
\begin{array}{ll}
\sum_{r = 0}^{\infty} \omega_i^r(u_i[h,r|\vec{\sigma}^*] - u_i[h,r|\vec{\sigma}']) & \geq\\
- \sum_{r=0}^{\mdel_i} \omega_i^r ((1-q_i^*[h,r|\vec{\sigma}^*])\gamma_i \bar{p}_i[h,r|\vec{\sigma}^*]  + \epsilon \beta_i)+ \sum_{r=\mdel_i+1}^{\mdel_i + \pd} \omega_i^r u_i^*[h,r | \vec{\sigma}^*] .
\end{array}
$$
Therefore, if Inequality~\ref{eq:delay-equiv} is fulfilled for every $i \in {\cal N}$, $h_i \in {\cal H}_i$, and $h \in {\cal H}$ such 
that $\mu_i^*[h|h_i]>0$, then the PDC Condition holds. Consequently, by Theorem~\ref{theorem:priv-drop},
$(\vec{\sigma}^*,\vec{\mu}^*)$ is Sequentially Rational.
\end{proof}

\newpage

\subsubsection{Proof of Lemma~\ref{lemma:priv-suff}.}
\label{proof:lemma:priv-suff}
If $(\vec{\sigma}^*,\vec{\mu}^*)$ is Preconsistent, 
Assumptions~\ref{def:non-neg} and~\ref{def:priv-assum} hold, 
and Inequality~\ref{eq:priv-suff} is fulfilled for every $h$, $i \in {\cal N}$, and $r,r' \leq \mdel_i + \pd$ such that $q_i[h,r'|\vec{\sigma}^*] < 1$,
then there exist $\omega_i \in (0,1)$ for every $i \in {\cal N}$ such that $(\vec{\sigma}^*,\vec{\mu}^*)$ is Sequentially Rational:
$$\frac{\beta_i}{\gamma_i}> \bar{p}_i[h,r|\vec{\sigma}^*] \frac{1}{A}+\bar{p}_i[h,r'|\vec{\sigma}^*]\frac{1}{B - C},$$
where
\begin{itemize}
  \item $A = 1 - \frac{\epsilon(\mdel_i+1)}{(1-q_i[h,r|\vec{\sigma}^*])\pd}$.
  \item $B=\frac{\pd}{c}$.
  \item $C= \frac{\epsilon(\mdel_i+1)}{1-q_i[h,r'|\vec{\sigma}^*]}$.
\end{itemize}

\begin{proof}
Consider the above assumptions and assume by contradiction that $(\vec{\sigma}^*,\vec{\mu}^*)$ is not Sequentially Rational.

The proof considers history $h_1$ that minimizes the first component of Inequality~\ref{eq:delay-equiv} and $h_2$ that minimizes the second component,  for any history $h \in {\cal H}$.
More precisely, fix $h$ and $i$:
\begin{equation}
\label{eq:ps-1}
\begin{array}{ll}
h_1 = \mbox{\emph{argmin}}_{\hevol[h,r|\vec{\sigma}^*]| r \in \{0 \ldots \mdel_i\}} -((1-q_i[h,r|\vec{\sigma}^*])\gamma_i \bar{p}_i[h,r|\vec{\sigma}^*]  + \epsilon \beta_i).\\
h_2 = \mbox{\emph{argmin}}_{\hevol[h,r|\vec{\sigma}^*]| r \in \{\mdel_i +1 \ldots \mdel_i + \pd\}} u_i[h,r|\vec{\sigma}^*].
\end{array}
\end{equation}

Let $u_i^{h_2} = u_i[h_2,0|\vec{\sigma}^*]$. We can write:
\begin{equation}
\label{eq:ps-2}
\begin{array}{ll}
- \sum_{r=0}^{\mdel_i} \omega_i^r ((1-q_i^*[h,r|\vec{\sigma}^*])\gamma_i \bar{p}_i[h,r|\vec{\sigma}^*]  + \epsilon \beta_i)+ \sum_{r=\mdel_i+1}^{\mdel_i + \pd} \omega_i^r u_i^*[h,r | \vec{\sigma}^*]  & \geq \\
-\sum_{r = 0}^{\mdel_i} \omega_i((1-q_i[h_1,0|\vec{\sigma}^*])\gamma_i \bar{p}_i[h_1,0|\vec{\sigma}^*] + \epsilon \beta_i) + \sum_{r = 1}^{\mdel_i + \pd} \omega_i^r u_i^{h_2} & =\\
- a \frac{1-\omega_i^{\mdel_i +1}}{1-\omega_i} + \frac{\omega_i^{\mdel_i + 1} - \omega_i^{\mdel_i + \pd+1}}{1 - \omega_i} u_i^{h_2},
\end{array}
\end{equation}
where 
$$a = (1-q_i[h_1,0|\vec{\sigma}^*])\gamma_i \bar{p}_i[h_1,0|\vec{\sigma}^*] + \epsilon \beta_i.$$

We want to fulfill
\begin{equation}
\label{eq:ps-3}
\begin{array}{ll}
- a \frac{1-\omega_i^{\mdel_i +1}}{1-\omega_i} + \frac{\omega_i^{\mdel_i + 1} - \omega_i^{\mdel_i + \pd+1}}{1 - \omega_i} u_i^{h_2} \geq 0\\
-a  + \omega_i^{\mdel_i + 1}(u_i^{h_2} + a) - \omega_i^{ \mdel_i + \pd+1} u_i^{h_2} & \geq 0.
\end{array}
\end{equation}

Again, this inequality corresponds to a polynomial with degree $\pd+1$. If $q_i^*[h_1,0|\vec{\sigma}^*] = 1$, then by our assumptions $q_i[h_2,0|\vec{\sigma}^*] =1$, $a=0$, and
the Inequality holds. Suppose then that 
$$q_i[h_1,0|\vec{\sigma}^*], q_i[h_2,0|\vec{\sigma}^*] < 1.$$

The polynomial has a zero in $\omega_i = 1$. If $a=0$, then the Inequality holds.
Consider, then, that $a > 0$. In these circumstances, a solution to~\ref{eq:ps-3}
exists for $\omega_i \in (0,1)$ iff the polynomial is strictly concave and has another zero in $(0,1)$. This is true iff
the polynomial has a maximum in $(0,1)$. The derivatives yield the following conditions:
\begin{enumerate}
 \item $\exists_{\omega_i \in (0,1)} (\mdel_i+1)(u_i^{h_2} + a) - (\pd+\mdel_i + 1) \omega_i u_i^{h_2} = 0 \Rightarrow \exists_{\omega_i \in (0,1)} \omega_i = \frac{(\mdel_i+1)(u_i^{h_2}+a)}{(\pd+\mdel_i+1)u_i^{h_2}}$.
 \item $-(\pd+\mdel_i + 1)\pd u_i^{h_2} < 0 \Rightarrow u_i^{h_2} > 0$.
\end{enumerate}

By our assumptions, Condition~$1$ implies that:
\begin{equation}
\label{eq:ps-4}
\begin{array}{ll}
u_i^{h_2}\pd& > (\mdel_i+1)a\\
&\\
(1-q_i[h_2,0|\vec{\sigma}^*]) \beta_i \pd &> (1-q_i[h_2,0|\vec{\sigma}^*])\gamma_i\bar{p}_i[h_2,0|\vec{\sigma}^*] \pd + (1-q_i[h_1,0|\vec{\sigma}^*])\gamma_i\bar{p}_i[h_1,0|\vec{\sigma}^*] + \epsilon \beta_i\\
&\\
((1-q_i[h_2,0|\vec{\sigma}^*])\pd - \epsilon(\mdel_i +1)) \beta_i &> (1-q_i[h_2,0|\vec{\sigma}^*])\gamma_i\bar{p}_i[h_2,0|\vec{\sigma}^*] \pd + (1-q_i[h_1,0|\vec{\sigma}^*])\gamma_i\bar{p}_i[h_1,0|\vec{\sigma}^*]\\
&\\
\frac{\beta_i}{\gamma_i} ((1-q_i[h_2,0|\vec{\sigma}^*])\pd - \epsilon(\mdel_i +1)) & > (1-q_i[h_2,0|\vec{\sigma}^*])\bar{p}_i[h_2,0|\vec{\sigma}^*] \pd + (1-q_i[h_1,0|\vec{\sigma}^*])\bar{p}_i[h_1,0|\vec{\sigma}^*].
\end{array}
\end{equation}

Solving in order to the benefit-to-cost ratio,
$$
\begin{array}{l}
\frac{(1-q_i[h_2,0|\vec{\sigma}^*])\bar{p}_i[h_2,0|\vec{\sigma}^*]\pd}{(1-q_i[h_2,0|\vec{\sigma}^*])\pd - \epsilon(\mdel_i +1)}=\bar{p}_i[h_2,0|\vec{\sigma}^*] \frac{1}{1 - \frac{\epsilon(\mdel_i +1)}{(1-q_i[h_2,0|\vec{\sigma}^*])\pd}}=\bar{p}_i[h_2,0|\vec{\sigma}^*] \frac{1}{A}.
\end{array}
$$
where $A = 1 - \frac{\epsilon(\mdel_i+1)}{(1-q_i[h_2,0|\vec{\sigma}^*])\pd}$.

Continuing, by Assumption~\ref{def:priv-assum}, it is true that
$$\frac{(1-q_i[h_2,0|\vec{\sigma}^*])}{(1-q_i[h_1,0|\vec{\sigma}^*])}\geq \frac{1}{c}.$$
Therefore,
$$
\begin{array}{l}
\frac{(1-q_i[h_1,0|\vec{\sigma}^*])\bar{p}_i[h_1,0|\vec{\sigma}^*]}{(1-q_i[h_2,0|\vec{\sigma}^*])\pd - \epsilon(\mdel_i +1)}= \bar{p}_i[h_1,0|\vec{\sigma}^*] \frac{1}{\frac{(1-q_i[h_2,0|\vec{\sigma}^*])\pd}{(1-q_i[h_1,0|\vec{\sigma}^*])} - \frac{\epsilon(\mdel_i +1)}{(1-q_i[h_1,0|\vec{\sigma}^*])}}\\
\\
\leq  \bar{p}_i[h_1,0|\vec{\sigma}^*] \frac{1}{\frac{\pd}{c} - \frac{\epsilon(\mdel_i+1)}{1-q_i[h_1,0|\vec{\sigma}^*]}} =  \bar{p}_i[h_1,0|\vec{\sigma}^*] \frac{1}{B -C}.
\end{array}
$$
where:
\begin{itemize}
  \item $B = \frac{\pd}{c}$.
  \item $C= \frac{\epsilon(\mdel_i+1)}{1-q_i[h_1,0|\vec{\sigma}^*]}$.
\end{itemize}

In summary, we have
\begin{equation}
\label{eq:ps-5}
\begin{array}{ll}
\frac{\beta_i}{\gamma_i} &> \bar{p}_i[h_2,0|\vec{\sigma}^*] \frac{1}{1 - A} +  \bar{p}_i[h_1,0|\vec{\sigma}^*] \frac{1}{B -C} \Rightarrow\\
\frac{\beta_i}{\gamma_i} ((1-q_i[h_2,0|\vec{\sigma}^*])\pd - \epsilon(\mdel_i +1)) & > (1-q_i[h_2,0|\vec{\sigma}^*])\bar{p}_i[h_2,0|\vec{\sigma}^*] \pd + (1-q_i[h_1,0|\vec{\sigma}^*])\bar{p}_i[h_1,0|\vec{\sigma}^*].
\end{array}
\end{equation}

Consequently, if Inequality~\ref{eq:priv-suff} is true, then so is~\ref{eq:ps-4}.
Furthermore, it also holds that
$$\beta_i > \gamma_i \bar{p}_i[h_1,0|\vec{\sigma}^*] \Rightarrow u_i^h >0.$$

That is, Inequality~\ref{eq:priv-suff} implies Conditions~1 and~2 of the polynomial for any $h$ and some $\omega_i \in (0,1)$,
which by transitivity implies that~\ref{eq:ps-3} is true. By~\ref{eq:ps-2},
Inequality~\ref{eq:delay-equiv} is fulfilled for every history $h$. Lemma~\ref{lemma:delay-equiv},
allows us to conclude that $(\vec{\sigma}^*,\vec{\mu}^*)$ is Sequentially Rational.
This is a contradiction, proving the result.
\end{proof}

\newpage

\subsubsection{Proof of Theorem~\ref{theorem:priv-effect}.}
\label{proof:theorem:priv-effect}
If $(\vec{\sigma}^*,\vec{\mu}^*)$ is Preconsistent, Assumptions~\ref{def:non-neg} and~\ref{def:priv-assum} hold for $\epsilon \ll 1$,
and $\pd \geq \mdel +1$, then there exists a constant $c>0$
such that $\psi[\vec{\sigma}^*|\vec{\mu}^*] \supseteq (v,\infty)$, where
$$
v = \max_{i \in {\cal N}} \max_{h \in {\cal H}}\bar{p}_i[h,0|\vec{\sigma}^*](1+c).
$$

\begin{proof}
The idea is to simplify Inequality~\ref{eq:priv-suff} for $\epsilon \ll 1$ and $\pd \geq \mdel + 1$.

Recall that 
$$A = 1 - \frac{\epsilon(\mdel_i+1)}{(1-q_i[h,r|\vec{\sigma}^*])\pd}.$$
Thus, this yields
$$\frac{1}{A} = \frac{(1-q_i[h,r|\vec{\sigma}^*]) \pd}{(1-q_i[h,r|\vec{\sigma}^*])\pd (1 - \epsilon \frac{\mdel_i+1}{\pd})} \leq \frac{(1-q_i[h,r|\vec{\sigma}^*]) \pd}{(1-q_i[h,r|\vec{\sigma}^*]) \pd(1-\epsilon)} \approx 1.$$

Moreover, by Assumption~\ref{def:priv-assum},
$$
\begin{array}{l}
\frac{1}{B - C} = \frac{1}{\frac{\pd}{c} -  \frac{\epsilon(\mdel_i+1)}{1-q_i[h,r|\vec{\sigma}^*]}}=\\
\\
\frac{1-q_i[h,r|\vec{\sigma}^*]}{(1-q_i[h,r|\vec{\sigma}^*])\frac{\pd}{c} - \epsilon(\mdel_i+1)}\leq\\
\\
\frac{(1-q_i[h,r|\vec{\sigma}^*])}{(1-q_i[h,r+\pd|\vec{\sigma}^*])\pd - \epsilon(\mdel_i+1)}\leq\\
\\
\frac{(1-q_i[h,r|\vec{\sigma}^*])}{(1-q_i[h,r+\pd|\vec{\sigma}^*]- \epsilon)(\mdel_i+1)}\approx\\
\\
\frac{(1-q_i[h,r|\vec{\sigma}^*])}{(1-q_i[h,r+\pd|\vec{\sigma}^*])(\mdel_i+1)}\leq\\
\\
\frac{(1-q_i[h,r|\vec{\sigma}^*])}{(1-q_i[h,r+\pd|\vec{\sigma}^*](\mdel_i+1))} \leq\\
\\
\frac{c}{\mdel_i+1}.
\end{array}
$$

Thus, for any $r,r' \geq 0$,
$$\bar{p}_i[h,r|\vec{\sigma}^*] \frac{1}{A}+\bar{p}_i[h,r'|\vec{\sigma}^*]\frac{1}{B - C} \leq \bar{p}_i[h,r|\vec{\sigma}^*] + \bar{p}_i[h,r'|\vec{\sigma}^*] \frac{c}{\mdel_i +1}.$$

Thus, there exists a constant $c'= \frac{c}{\mdel_i + 1}$ such that if for every $i$ we have
$$\frac{\beta_i}{\gamma_i} > \max_{h \in {\cal H}} \bar{p}_i[h|\vec{\sigma}^*] (1 + c') \geq \bar{p}_i[h,r|\vec{\sigma}^*] + \bar{p}_i[h,r'|\vec{\sigma}^*] \frac{c}{\mdel_i +1},$$
then Inequality~\ref{eq:priv-suff} is fulfilled for every $h$, $r$, and $r'$, and for some $\omega_i \in (0,1)$.
By Lemma~\ref{lemma:priv-suff}, this implies $(\vec{\sigma}^*,\vec{\mu}^*)$ is Sequentially Rational
and the result follows.
\end{proof}

\end{document}